\documentclass[11pt]{article}
\usepackage{amsthm,amsgen,amsmath,amstext,amsbsy,amsopn,amssymb,latexsym,mathrsfs}
\usepackage{array}
\usepackage{multirow}
\usepackage{graphicx}
\usepackage{amssymb,bm,bbm}
\usepackage{titling}
\usepackage{url,bm}
\usepackage{algpseudocode,float,algorithm}
\usepackage[round]{natbib}
\usepackage{bbm}
\usepackage{verbatim}
\usepackage[normalem]{ulem}
\usepackage{enumerate}
\usepackage{subcaption}
\usepackage{authblk}
\usepackage[dvipsnames,table,dvipsnames*, svgnames*]{xcolor}

\usepackage{hyperref}
\hypersetup{
    colorlinks=true,
    linkcolor=blue,
    filecolor=blue,
    urlcolor=blue,
    citecolor=blue,
}
\usepackage{cleveref}
\usepackage{pifont}

\allowdisplaybreaks

\newcommand{\argmin}{\mathop{\rm arg\min}}
\newcommand{\argmax}{\mathop{\rm arg\max}}
\newcommand{\norm}[1]{\left\|#1\right\|}
\newcommand{\abs}[1]{\left|#1\right|}
\newcommand{\sgn}{{\rm sgn}}

\newtheorem{Corollary}{Corollary}
\newtheorem{Proposition}{Proposition}
\newtheorem{Lemma}{Lemma}
\newtheorem{Definition}{Definition}
\newtheorem{Theorem}{Theorem}
\newtheorem{Example}{Example}
\newtheorem{Remark}{Remark}

\newtheorem{Assumption}{Assumption}
\newtheorem{Claim}{Claim}

\newcommand{\R}{{\mathbb{R}}}
\newcommand{\E}{{\mathbb{E}}}
\newcommand{\N}{{\mathbb{N}}}
\newcommand{\Prob}{{\mathbb{P}}}
\newcommand{\1}{{\mathbbm{1}}}
\newcommand{\U}{{\mathcal{U}}}
\newcommand{\A}{{\mathcal{A}}}
\newcommand{\D}{{\mathcal{D}}}
\newcommand{\X}{{\mathcal{X}}}

\newcommand{\HH}{{\mathcal{H}}}
\newcommand{\II}{{\mathcal{I}}}

\newcommand{\KK}{\mathcal{K}}
\newcommand{\EE}{\mathcal{E}}
\newcommand{\app}{{\rm app}}

\newcommand{\Var}{{\rm Var}}

\newcommand{\EOO}{{\rm EOO}}

\newcommand{\risk}{{\mathcal{R}}}

\graphicspath{{Figures/}}

\begin{document}

\title{Finite-Sample and Distribution-Free Fair Classification: Optimal Trade-off Between Excess Risk and Fairness, and the Cost of Group-Blindness}

\author[1]{Xiaotian Hou}
\author[2]{Linjun Zhang}
\affil[1]{Department of Biostatistics, Epidemiology and Informatics, University of Pennsylvania}
\affil[2]{Department of Statistics, Rutgers University}
\date{}
\maketitle

\begin{abstract}
Algorithmic fairness has become a central concern in modern machine learning and AI applications. However, two pressing challenges remain: (1) The fairness guarantees of existing methods often rely on specific data distributional assumptions and large sample sizes, which can lead to fairness violations in practice. (2) Due to legal and societal considerations, using sensitive group attributes during decision-making (referred to as the group-blind setting) may not always be feasible.

In this work, we quantify the impact of enforcing algorithmic fairness and group-blindness/awareness in binary classification under group fairness constraints. Specifically, we propose a unified framework for fair classification that provides distribution-free and finite-sample fairness guarantees with controlled excess risk. This framework is applicable to various group fairness notions in both group-aware and group-blind scenarios. Our approach is based on a post-processing procedure that can be applied to arbitrary black-box models, making it directly compatible with modern machine learning pipelines. Furthermore,  we establish a minimax lower bound showing the minimax rate-optimality of our proposed algorithm up to logarithmic factors.

Through extensive synthetic and real data studies, we further demonstrate the competitive or superior performance of our algorithm compared to existing methods, and provide empirical support for our theoretical findings.
\end{abstract}
\tableofcontents
\section{Introduction}\label{sec:introduction}

Machine learning algorithms are increasingly used in high-stakes domains such as university admissions \citep{waters2014grade}, hiring \citep{pimpalkar2023job}, and criminal justice \citep{berk2012criminal}. However, empirical studies show that these systems can inherit or even amplify existing data biases, disproportionately impacting historically underrepresented or disadvantaged groups \citep{angwin2022machine, tolan2019machine}.

These concerns have motivated extensive research on bias mitigation and algorithmic fairness \citep{dwork2012fairness, hardt2016equality, agarwal2018reductions}. Yet, existing fairness guarantees often rely on large samples and restrictive distributional assumptions, such as sub-Gaussianity, limiting their practicality for complex data or small-sample settings. This highlights the need for algorithms that ensure fairness without any distributional assumptions (which we term ``distribution-free'') and in a finite-sample manner. 

Another practical challenge is the group-blind setting, where sensitive attributes are accessible during the training but not during the test (or decision-making) time due to regulatory or contractual restrictions \citep{lipton2018does}. For instance, the U.S. Supreme Court has banned the use of race in college admissions \citep{rice2023supreme, bather2023unpacking}.

In this paper, we aim to answer the following fundamental questions:
\begin{center}
\it How to achieve finite-sample and distribution-free fairness,  what is the optimal approach, and how does the optimal performance depend on the observability of sensitive attributes?
\end{center}

Motivated by these challenges, we study binary classification under various group fairness notions in both group-aware and group-blind scenarios. We develop a unified framework that (1) derives the Bayes optimal fair classifier, (2) constructs a post-processing procedure that can be applied to arbitrary black-box models and yields classifiers with finite-sample, distribution-free fairness guarantees, an important feature for modern AI systems which are often complex and opaque, and (3) analyzes their excess risk. For binary sensitive attributes, we establish minimax excess risk lower bounds, showing that our framework is minimax optimal up to logarithmic factors. This analysis elucidates key fairness–accuracy trade-offs: the optimal excess risk quantifies the inevitable cost of enforcing finite-sample fairness, while comparisons between group-aware and group-blind settings reveal the cost of group-blindness due to errors in predicting unobserved sensitive attributes. When fairness constraints are overly stringent, the group-blind excess risk may approach a constant, precluding meaningful predictive performance. Finally, since the triangle inequality fails in our setting, standard tools such as Le Cam’s method, Fano’s lemma, and Assouad’s lemma become inapplicable; we thus develop a new proof technique to obtain tight minimax bounds, which may be of independent interest.

In summary, our contributions are threefold:
\begin{enumerate}
    \item For various fairness notions in both group-aware and group-blind scenarios, we propose a unified framework that derives Bayes optimal fair classifiers, constructs classifiers with finite-sample, distribution-free fairness guarantees, and analyzes excess risks with optimality guarantees. A key component of our framework is a unified post-processing algorithm that applies to arbitrary pretrained black-box classifiers. To the best of our knowledge, this is the first framework that achieves finite-sample and distribution-free fairness with minimax-optimal excess risk control. 
    
    \item For binary sensitive attributes, we establish a minimax lower bound on the excess risk using a novel proof technique that remains effective even when the triangle inequality fails due to fairness constraints, establishing the minimax optimal rate for excess risk in fair classification. The new proof technique may be of independent interest.
    
    \item Under our proposed framework, we are able to characterize the exact fundamental trade-off between fairness and excess risk, revealing the unavoidable increase in excess risk due to group-blindness.
\end{enumerate}

\subsection{Related Works}\label{sec:literature}

Algorithms for group fairness generally fall into three categories: pre-processing, in-processing, and post-processing. Pre-processing methods adjust data distributions to mitigate bias against protected groups while retaining information \citep{calmon2017optimized, feldman2015certifying, johndrow2019algorithm, zeng2024bayes}. In-processing approaches incorporate fairness constraints or penalties during training to balance fairness and accuracy \citep{calders2009building, celis2019classification, cho2020fair, donini2018empirical, kamishima2012fairness, narasimhan2018learning, wadsworth2018achieving, zhang2018mitigating, zeng2024bayes}. Post-processing techniques modify outputs of unconstrained models to reduce discrimination across demographic groups \citep{chzhen2019leveraging, schreuder2021classification, xian2023fair, li2022fairee, zeng2024bayes, chen2024posthoc, xian2024unified}. A comprehensive overview is provided in \cite{caton2024fairness}.

Several studies have examined the form of Bayes optimal classifiers under specific fairness constraints \citep{corbett2017algorithmic, celis2019classification, menon2018cost, chzhen2019leveraging, chzhen2022minimax, xian2023fair, zeng2024bayes, chen2024posthoc, xian2024unified}. 
Closely related studies include \cite{chen2024posthoc}, \cite{zeng2024bayes}, and \cite{xian2024unified}. {\cite{chen2024posthoc} characterized group-blind Bayes optimal fair classifiers implicitly and proposed a post-processing algorithm based on empirical risk minimization, with fairness guarantees holding only asymptotically.
In both group-aware and group-blind scenarios, \cite{zeng2024bayes} explicitly characterized Bayes optimal classifiers under linear disparity constraints and developed pre-, in-, and post-processing algorithms. Their analysis mainly addresses binary sensitive attributes, extending to multi-class cases only for demographic parity, and doesn't provide theoretical guarantees.
\cite{xian2024unified} examined multiclass classification in both group-aware and group-blind settings. They characterized Bayes optimal classifiers using pointwise linear programming and proposed a post-processing algorithm that requires solving one such optimization per prediction. Their fairness guarantees depend on the multicalibration error of a pretrained model and therefore do not provide finite-sample, distribution-free validity.}

\cite{zeng2024minimax} studied a related problem, where they only focus on group-aware classification under demographic parity with binary sensitive attributes and introduce a fairness-aware excess risk measure. In comparison, we study the excess risk among fair classifiers with general fairness notions under both group-aware and blind settings. Their fairness-aware excess risk is always smaller and may underestimate the intrinsic difficulty of fair classification. In addition, they derive group-aware minimax rates for their risk and propose an asymptotically fair algorithm, whereas our work establishes both group-aware and group-blind minimax excess risks and achieves finite-sample, distribution-free fairness under general fairness notions. We summarize the comparison of our work with these closely related studies in Table \ref{tab:literature}, see Sections \ref{sec:preliminary} and \ref{sec:unify} for detailed discussions.

\begin{table}
    \scriptsize
    \hspace*{-2cm}
    \begin{tabular}{l|ccccc}
        Properties & \cite{chen2024posthoc} & \cite{zeng2024minimax} & \cite{zeng2024bayes} & \cite{xian2024unified} & Ours \\\hline
        Unified framework for general fairness notion & \checkmark & \ding{55} &\checkmark& \checkmark & \checkmark \\
        Finite-sample, distribution-free fairness & \ding{55} & \ding{55}& \ding{55} &\ding{55} & \checkmark \\
        Optimal group-aware excess risk & \ding{55}  & \checkmark& \ding{55} & \ding{55} & \checkmark \\
        Optimal group-blind excess risk & \ding{55} & \ding{55}& \ding{55} & \ding{55} & \checkmark\\
    \end{tabular}
    \caption{Comparison with existing methods.} 
    \label{tab:literature}
\end{table}

\textbf{Organization. }
The remainder of the paper is organized as follows. Section~\ref{sec:preliminary} formulates the fair classification problem and introduces notation. Section~\ref{sec:unify} presents a unified framework for binary sensitive attributes, providing fairness and excess risk guarantees. Section~\ref{sec:binary_eoo} applies this framework to equality of opportunity and establishes minimax lower bounds for excess risk in both group-aware and group-blind settings. Section~\ref{sec:numerical} evaluates the algorithm’s numerical performance. Applications to other fairness notions, extensions to multi-class sensitive attributes, and all proofs are deferred to the Supplementary Material.

\section{Preliminaries}\label{sec:preliminary}

\subsection{Model Set-up}\label{sec:setup}
We observe $n$ i.i.d. samples $\D=\big\{(X_i,A_i,Y_i):i\in[n]\big\}$ from $P_{X,A,Y}$, where each sample $(X_i,A_i,Y_i)$ consists of non-sensitive covariates $X_i\in\X\subset\R^d$, a categorical sensitive attribute $A_i\in[K]$, and a binary label $Y_i\in\{0,1\}$.

To analyze fair classification, we consider randomized classifiers \citep{li2022fairee,zeng2024bayes}, defined as follows. 

\begin{Definition}[Randomized Classifier]\label{def:randomized_classifier}
    A randomized classifier $f$ is a measurable function $f:\R^d\times[K]\rightarrow[0,1]$, where $f(X,A)=\Prob(Y_f(X,A)=1|X,A)$ and $Y_f(X,A)\in\{0,1\}$ denotes the predicted label induced by $f(X,A)$.
\end{Definition}
The set of classifiers mapping $\X\times[K]$ to $\{0,1\}$ is nonconvex, while the set of randomized classifiers forms its convex hull. When a randomized classifier $f$ takes only values in $\{0,1\}$, it reduces to a standard classifier with $Y_f(X,A)=f(X,A)$.

Given training data $\D$, our goal is to construct a randomized classifier $\hat f$ to predict $Y$ from $(X,A)$ for a new sample $(X,A,Y)\sim P_{X,A,Y}$. While the sensitive attribute $\{A_i:i\in[n]\}$ in the historical data may be used during training, it is not always available at prediction time. We therefore categorize the classification problems into two cases:

\begin{enumerate}
    \item Group-aware: $\hat f^{\rm aware}:\R^d\times[K]\rightarrow [0,1]$ use both the non-sensitive covariates $X$ and the sensitive attribute $A$ as inputs.
    \item Group-blind: $\hat f^{\rm blind}:\R^d\rightarrow [0,1]$ predicts based only on the non-sensitive covariate $X$.
\end{enumerate}
To unify notation, we slightly abuse it as follows. \textbf{For any group-blind function $f^{\rm blind}(X)$ with domain $\R^d$, we write $f^{\rm blind}(X,A)$ and use the superscript to highlight that it depends only on the non-sensitive covariates $X$.} Accordingly, for functions on $\R^d\times[K]$, we use a unified form $f^G(X,A)$ with $G\in\{{\rm aware, blind}\}$ to denote the group-aware and group-blind scenarios, where $f^{\rm blind}$ ignores the sensitive attribute $A$.

To quantify algorithmic fairness in classification, various group fairness notions have been proposed \citep{calders2009building, hardt2016equality, corbett2017algorithmic, berk2021fairness}. The methods developed here apply broadly across these notions. As an illustration, we present equality of opportunity (EOO) with binary sensitive attributes, deferring other common definitions with multiclass sensitive attributes to Section A of the Supplement.

\begin{Definition}[Unfairness Measure in terms of EOO]\label{def:unfairness_eoo_binary}
    For binary sensitive attribute $K=2$ and randomized classifier $f$, the unfairness of $f$ under EOO is
    $\U_{\rm EOO}(f)=\big|\Prob\big(Y_f(X,A)=1|A=1,Y=1\big)-\Prob\big(Y_f(X,A)=1|A=2,Y=1\big)\big|$,
    where probabilities are over the randomness of the test sample $(X, A, Y)$ and of $Y_f(X,A)$ given $f(X,A)$.
\end{Definition}
Throughout the work, we consider the following fairness constraint.
\begin{Definition}[$(\alpha,\delta)$-Fairness Constraint]\label{def:fair}
    For an unfairness measure $\U$ that maps a classifier to $\R_+$, a constructed classifier $\hat f$ satisfies the $(\alpha,\delta)$-fairness constraint if 
    \begin{equation}\label{eq:alpha_delta_fair}
    \Prob(\U(\hat f)\le\alpha)\ge1-\delta, 
    \end{equation}
    where $\U$ measures the unfairness of $\hat f$ on an independent test sample, and the probability $\Prob$  is taken over all randomness of $\hat f$, including that from the training data and (possibly) randomization introduced in the algorithm. 
\end{Definition}
\begin{Remark}[Finite-Sample Fairness]
    If a classifier $\hat f$ trained on $n$ samples achieves $(\alpha,\delta)$-fairness guarantee, we say $\hat f$ achieves finite-sample fairness. In contrast, existing works \citep{chen2024posthoc,zeng2024minimax,xian2024unified} typically ensure $(\alpha+o(1),\delta)$-fairness, achieving $(\alpha,\delta)$-fairness only asymptotically.
\end{Remark}
For $G\in\{{\rm aware, blind}\}$, define the misclassification error $\risk(f^G)=\Prob(Y\ne Y_{f^G}(X,A)),$
where $\Prob$ is over both the independent sample $(X,A,Y)$ and the randomness of $Y_{f^G}(X,A)$ given $f^G(X,A)$. Our goal is to estimate the Bayes optimal $\alpha$-fair classifier $f^{*G}_\alpha$,  
\begin{equation}\label{eq:bayes_obj}
    f^{*G}_\alpha \in\argmin_{f^G\in[0,1]^{\R^d\times[K]}}\risk(f^G),\quad {\rm s.t.}\quad\U(f^G)\le\alpha.
\end{equation}
Estimating $f^{*G}_\alpha$ is challenging because, although Problem~\eqref{eq:bayes_obj} may be convex in $f^G$, it is generally nonconvex in the parameterization of $f^G$ \citep{wu2019convexity,celis2019classification,caton2024fairness}, and solving the empirical version of Problem~\eqref{eq:bayes_obj} does not guarantee $(\alpha,\delta)$-fairness \eqref{eq:alpha_delta_fair}. As shown in Section~\ref{sec:unify} and Section A of the supplement, to address these problems, we propose a post-processing algorithm that modifies any (black-box) classifier trained without fairness constraints, reducing the original nonconvex Problem~\eqref{eq:bayes_obj} over possibly complex function classes to a one-dimensional (resp. $K$-dimensional) nonconvex optimization for binary (resp. $K$-class) sensitive attributes.

\subsection{Notation}\label{sec:notations}
For $n,m\in\mathbb{N}_+$, let $[n]=\{1,\ldots,n\}$ and $m+[n]=\{m+1,\ldots, m+n\}$. For spaces $\X$, $\mathcal{Y}$, denote $\mathcal{Y}^{\X}$ as the set of all functions mapping from $\X$ to $\mathcal{Y}$. Define $\eta^{\rm blind}(X,A)=\Prob(Y=1|X)$ and $\eta^{\rm aware}(X,A)=\Prob(Y=1|X,A)$ as the best predictions of $Y$ given $X$ and $(X,A)$, respectively. For $a\in[K], y\in\{0,1\}$, let $\rho_a(X)=\Prob(A=a|X)$ and $\rho_{a|y}(X)=\Prob(A=a|Y=y,X)$ to be the conditional distributions of $A$ given $X$ and $(X,Y)$, respectively. Denote $p_a=\Prob(A=a)$, $p_{y,a}=\Prob(Y=y,A=a)$, and $p_Y=\Prob(Y=1)$ to be the probability measures of $A$, $(Y,A)$, and $Y$, separately. For any random vector $Z$, $P_Z$ denotes its joint distribution. For any function $f$ of $x$, define the $L_\infty$ norm $\|f\|_\infty$ to be the supremum value of $|f(x)|$ on the support of $P_X$, i.e., $\|f\|_\infty=\sup_{x\in\X}|f(x)|$. 
Let ${\rm Leb}(\cdot)$ denote the Lebesgue measure on $\R^d$ and $B_q(c,r)=\{x\in\R^d:\|x-c\|_q\le r\}$ the $l_q$ ball in $\R^d$ centered at $c$ with radius $r$. For $a,b\in\R$, write $a\wedge b=\min\{a,b\}$, $a\vee b=\max\{a,b\}$ and $(a)_+=a\vee 0$. For $\beta>0$, $\lfloor\beta\rfloor$ denotes the largest integer strictly smaller than $\beta$. For any $k$ times differentiable function $g:\R^d\rightarrow \R$ and $x\in\R^d$, $g_{k,x}:\R^d\rightarrow \R$ is the degree $k$ Taylor polynomial of $g$ at $x$. We use $c$ and $C$ for absolute constants that may vary across appearances. For positive sequences $\{a_n\}$ and $\{b_n\}$, $a_n\lesssim b_n$ means $a_n\le Cb_n$ for all $n$, $a_n\gtrsim b_n$ if $b_n\lesssim a_n$, $a_n\asymp b_n$ if $a_n\lesssim b_n$ and $b_n\lesssim a_n$.

\section{A Unified Framework with Binary Sensitive Attributes}\label{sec:unify}

In this section, we develop a unified post-processing framework for fair classification with binary sensitive attributes ($K=2$), applicable to various fairness notions in both group-aware and group-blind scenarios. 
We first derive Bayes optimal $\alpha$-fair classifiers.
Then in Section~\ref{sec:unify_binary}, we propose a universal post-processing algorithm that guarantees both fairness and excess risk control. The framework is further extended to multi-class sensitive attributes ($K>2$) in Section A.2 of the supplement. 

\subsection{Bayes Optimal $\alpha$-fair Classifier}\label{sec:unify_bayes}

In this section, we study the Bayes optimal $\alpha$-fair classifier. We start with an equivalent characterization of the unfairness measures, which enables both a closed-form expression of the Bayes optimal classifier and accurate finite-sample approximations of the unfairness measures. 
Recall $\eta^{\rm aware}(X,A)=\Prob(Y=1|X,A)$, $\eta^{\rm blind}(X,A)=\Prob(Y=1|X)$ and $\rho_{a|y}(X)=\Prob(A=a|X,Y=y)$. Although $\eta^{\rm blind}$ does not take $A$ as input, we still write $\eta^{\rm blind}(X,A)$ for notational consistency.
We illustrate using EOO (Definition \ref{def:unfairness_eoo_binary}) as an example. 
\begin{Example}\label{exa:phi_eoo}
    If we denote $\phi^{\rm aware}_{{\rm EOO}}(x,a)=\big(\frac{\1(a=1)}{p_{1,1}}-\frac{\1(a=2)}{p_{1,2}}\big)\eta^{\rm aware}(x,a)$, $\phi^{\rm blind}_{{\rm EOO}}(x,a)=\big(\frac{\rho_{1|1}(x)}{p_{1,1}}-\frac{\rho_{2|1}(x)}{p_{1,2}}\big)\eta^{\rm blind}(x,a)$,
then for $G\in\{{\rm aware, blind}\}$ and any classifier $f^G$, 
\begin{equation}\label{eq:eg1}
\U_{\rm EOO}(f^{\rm G})=|(\E_{X|Y=1,A=1}-\E_{X|Y=1,A=2})f^G(X,A)|=|\E\phi^G_{\rm EOO}(X,A)f^G(X,A)|.
\end{equation}
\end{Example}
The derivation of \eqref{eq:eg1} is in Section C of the supplement. In the group-blind scenario, where $A$ is unavailable for prediction, it can be verified that $\big\{x\in\X:\frac{\rho_{1|1}(x)}{p_{1,1}}=\frac{\rho_{2|1}(x)}{p_{1,2}}\big\}$ is the classification boundary of the Bayes optimal classifier $h^*\in\argmin_{h\in\{1,2\}^{\X}}\Prob(h(X)=2|A=1,Y=1)+\Prob(h(X)=1|A=2,Y=1)$, which predicts $A$ from $X$ under the group-wise misclassification error conditioned on $Y=1$.
Then $\phi^{\rm blind}_\EOO(x)> 0$ (resp. $< 0$) if the Bayes optimal classifier predicts $h^*(x)=1$ (resp. $=2$). When predicting $A$ is difficult, i.e., $|\frac{\rho_{1|1}(x)}{p_{1,1}}-\frac{\rho_{2|1}(x)}{p_{1,2}}|$ is small and $x$ is near the classification boundary, $|\phi^{\rm blind}_{\rm EOO}(x)|$ is small. Consequently, $\sgn(\phi^{\rm blind}_\EOO)$ gives the Bayes optimal prediction of $A$, and $|\phi^{\rm blind}_\EOO|$ reflects prediction confidence. Similar interpretations carry over to the group-aware scenario, where $A$ is known. In the group-aware scenario, $\phi^{\rm aware}_\EOO>0$ (resp. $<0$) if $A=1$ (resp. $A=2$), and its magnitude satisfies $|\phi^{\rm aware}_{\rm EOO}|\gtrsim\eta^{\rm aware}$, reflecting greater prediction confidence.

We show in Section A.1 of the supplement that most of the commonly used unfairness measures, including demographic parity, equality of opportunity, overall accuracy equality, and predictive equality, share the following form.
\begin{Definition}[Unfairness Measure]\label{def:unfairness_characterzation}
    For binary sensitive attribute ($K=2$) and any randomized classifier $f^G$, we consider unfairness measure $\U(f^G)\in\R_+$ such that
     \begin{equation}\label{eq_unfairness_characterization}
    \U(f^G)=\bigg|\sum_{j\in[m]}\kappa_j\E_jf^G(X,A)\bigg|=\big|\E\phi^G(X,A)f^G(X,A)\big|, 
    \end{equation}
    where $\{\kappa_j\in\R:j\in[m]\}$ are coefficients, $\E_j$ are expectations conditioned on sensitive attributes, and $\phi^G:\R^d\times[2]\rightarrow\R$ is a bounded function determined by fairness notions. 
\end{Definition}
Equation \eqref{eq_unfairness_characterization} expresses the unfairness measure in two different forms. The first form can be estimated via sample averages and is essential for finite-sample, distribution-free fairness control, while the second reveals the optimal adjustment direction for unconstrained classifiers and is crucial for characterizing and estimating the Bayes optimal fair classifier. {Related works \citep{chen2024posthoc,zeng2024bayes,xian2024unified} also proposed general formulations but emphasize different aspects. \cite{chen2024posthoc} and \cite{xian2024unified} consider only the first form. Therefore, \cite{chen2024posthoc} only provides incomplete characterizations of the Bayes optimal classifier and relies on empirical risk minimization for estimation. \cite{zeng2024bayes} focuses on the second form, limiting the tightness of fairness control, see Section \ref{sec:numerical} for details.}

Since $\risk(f^G)=p_Y+\E\big(1-2\eta^G(X,A)\big)f^G(X,A)$ and $\U(f^G)$ in \eqref{eq_unfairness_characterization} are convex in $f^G$, Problem~\eqref{eq:bayes_obj} is convex in $f^G$.
Proposition \ref{prop:bayes_binary} gives an explicit expression for the Bayes optimal $\alpha$-fair classifier, which is a translation of the unconstrained Bayes optimal classifier. The case with multi-class sensitive attributes is studied in Section A of the supplement.

\begin{Proposition}[Bayes Optimal $\alpha$-fair Classifier]\label{prop:bayes_binary}
    For $K=2, G\in\{{\rm aware,blind}\}$, and unfairness measure $\U$ defined in Definition \ref{def:unfairness_characterzation}, the Bayes optimal $\alpha$-fair classifier $f^{*G}_\alpha\in[0,1]^{\R^d\times [2]}$ in Problem \eqref{eq:bayes_obj} satisfies, $P_{X,A}$-almost surely, 
    \begin{align*}
        f^{*G}_\alpha (X,A)=&\1\big(g^{*G}_\alpha (X,A)>0\big)+b^G(X,A)\1\big(g^{*G}_\alpha (X,A)=0\big), 
    \end{align*}
    for $g^{*G}_\alpha (X,A)=2\eta^G(X,A)-1-\lambda^{*G}_\alpha\phi^G(X,A)$, 
    \begin{equation}\label{eq:lambda_obj_binary}
        \lambda^{*G}_\alpha \in\argmin_{\lambda\in\R}\E\big(2\eta^G(X,A)-1-\lambda\phi^G(X,A)\big)_++\alpha|\lambda|, 
    \end{equation}
    and any function $b^G\in[0,1]^{\R^d\times[2]}$ such that $f^{*G}_\alpha$ satisfies the fairness constraint and  
    \begin{equation}\label{eq:lambda_binary_optimality}
        \lambda^{*G}_\alpha\E\phi^G(X,A)f^{*G}_\alpha(X,A)=|\lambda^{*G}_\alpha|\alpha. 
    \end{equation}
\end{Proposition}
\begin{Remark}\label{rem:lambda_upper_unify}  The optimal fair classifier can be viewed as a shifted version of the unconstrained Bayes rule, where the shift $\lambda^{*G}_\alpha\phi^G(X,A)$ represents the adjustment required to enforce fairness. The magnitude $|\lambda^{*G}_\alpha|$ quantifies how strongly the fairness constraint perturbs the decision boundary.
   Since the set of minimizers of Problem~\eqref{eq:lambda_obj_binary} is closed, we take $\lambda^{*G}_\alpha$ as the one with the smallest absolute value when multiple solutions exist. It follows from \eqref{eq:lambda_binary_optimality} that 
    \[|\lambda^{*G}_\alpha|\alpha=\E\lambda^{*G}_\alpha\phi^G(X,A)\1\big(2\eta^G(X,A)-1>\lambda^{*G}_\alpha\phi^G(X,A)\big)\le\E\big[\big(2\eta^G(X,A)-1\big)f^{*G}_\alpha(X,A)\big]\le 1,\]
    implying $|\lambda^{*G}_\alpha|\le\alpha^{-1}$. As shown in Section~\ref{sec:binary_eoo}, in the group-aware setting, $|\lambda^{*{\rm aware}}_\alpha|$ is upper bounded by a constant even as $\alpha\to 0$. On the contrary, in the group-blind scenario, for any $\alpha>0$, there exist distributions for which $|\lambda^{*{\rm blind}}_\alpha|\asymp\alpha^{-1}$. 
\end{Remark}

Bayes optimal fair classifiers have been studied for specific scenarios and fairness notions \citep{corbett2017algorithmic, menon2018cost, schreuder2021classification}. More recent works provide unified analyses \citep{chen2024posthoc,zeng2024bayes,xian2024unified}. {Specifically, \cite{chen2024posthoc} analyzed group-blind Bayes optimal classifiers under fairness constraints of the first form in \eqref{eq_unfairness_characterization}, but didn't explicitly characterize $\lambda_\alpha^{*G}$. In both group-aware and group-blind scenarios, \cite{zeng2024bayes} derived explicit characterizations of the Bayes optimal classifier under fairness constraints of the second form in \eqref{eq_unfairness_characterization}, but they mainly focus on binary sensitive attributes, with multiclass cases studied only under demographic parity. Moreover, they didn't provide characterizations like \eqref{eq:lambda_binary_optimality}, leaving the range of $\lambda_\alpha^{*G}$(see Remark \ref{rem:lambda_upper_unify}) unspecified. We will show that $\lambda_\alpha^{*G}$ is crucial to both the excess risk upper and lower bounds. For multiclass classification under fairness constraints of the first form in \eqref{eq_unfairness_characterization}, \cite{xian2024unified} characterized the Bayes optimal fair classifier pointwise rather than globally. Therefore, their results are not directly comparable to ours.}

From Proposition~\ref{prop:bayes_binary}, $g^{*G}_\alpha=0$ defines the classification boundary of the Bayes-optimal fair classifier $f^{*G}_\alpha$. On this boundary, the prediction $Y_{f^{*G}_\alpha}$ induced by $f^{*G}_\alpha$ is randomized. For simplicity, we assume the boundary has zero probability measure, i.e., $\Prob(g^{*G}_\alpha(X,A)=0)=0$. Proposition~\ref{prop:bayes_binary} shows that $g^{*G}_\alpha$ is a translation of the unconstrained Bayes-optimal classification boundary $2\eta^G-1$ by $\lambda^{*G}_\alpha\phi^G$, motivating classifiers of the form $\hat f^G_\alpha=\1(2\hat\eta^G-1>\hat\lambda^G\hat\phi^G)$. As shown in Section~\ref{sec:fairness}, given any $\hat\eta^G$ and $\hat\phi^G$, there always exists a $\hat\lambda^G$ ensuring $(\alpha,\delta)$-fairness when $\alpha$ is not too small. Moreover, as shown in Section~\ref{sec:excess}, if $\hat\eta^G$ and $\hat\phi^G$ accurately estimate $\eta^G$ and $\phi^G$, the resulting classifier $\hat f^G_\alpha$ achieves low prediction error.

\subsection{Post-processing Algorithm}\label{sec:unify_binary}

In this section, we propose a general post-processing algorithm for various fairness notions with guaranteed fairness and excess risk. 

We decompose the tolerance $\delta$ in $(\alpha,\delta)$-fairness as $\delta=\delta_{\rm init}+\delta_{\rm post}$, where $\delta_{\rm init}$ controls the probability of inaccurate initial estimators and $\delta_{\rm post}$ bounds the failure probability of the post-processing algorithm. 
Throughout the section, the initial estimators $\hat\eta^G$ and $\hat\phi^G$ are treated as fixed and independent of the training data $\D=\{(X_i,A_i,Y_i):i\in[n]\}$. Then our goal is to design a post-processing algorithm $\A^G$ that maps from $\D,\hat\eta^G,\hat\phi^G$ to a classifier $\hat f^G_\alpha=\A^G(\D;\hat\eta^G,\hat\phi^G)\in[0,1]^{\R^d\times[2]}$ and satisfies the $(\alpha,\delta_{\rm post})$-fairness constraint: 
$\Prob_{\D}\big(\U\big(\A^G(\D;\hat\eta^G,\hat\phi^G)\big)\le\alpha\big)\ge 1-\delta_{\rm post}$.
Here, $\U$ measures the unfairness of the classifier $\A^G(\D;\hat\eta^G,\hat\phi^G)$ on an independent test sample, and $\Prob_\D$ accounts for randomness in $\D$ and in the algorithm itself.
The role of $\delta_{\rm init}$ will be demonstrated in Section~\ref{sec:binary_eoo}.

As shown in Proposition~\ref{prop:bayes_binary}, the Bayes optimal $\alpha$-fair classifier $f^{*G}_\alpha$ depends on three components: $\eta^G,\phi^G$ and $\lambda^{*G}_\alpha$. Given estimators $\hat\eta^G$ and $\hat\phi^G$, it remains to estimate $\lambda^{*G}_\alpha$ from the data $\D$. This estimation is motivated by the following characterization of $\lambda^{*G}_\alpha$.
\begin{Lemma}[Characterization of $\lambda^{*G}_\alpha$]\label{lem:lambda_binary}
   Under the model set-up described above, suppose $\Prob\big(g^{*G}(X,A)=0\big)=0$. Denote $s^G=\sgn\big(\E\phi^G(X,A)\1\big(2\eta^G(X,A)>1\big)\big)$ with $\sgn(0)\in[-1,1]$, then $\lambda^{*G}_\alpha$ defined in \eqref{eq:lambda_obj_binary} satisfies $\lambda^{*G}_\alpha=s^G|\lambda^{*G}_\alpha|$ with 
    \[|\lambda^{*G}_\alpha|=\argmin_{\lambda_+\ge 0}\lambda_+\quad{\rm s.t.}\quad s^G\E\phi^G(X,A)\1\big(2\eta^G(X,A)-1>s^G\lambda_+\phi^G(X,A)\big)\le\alpha. \]
\end{Lemma}
Lemma~\ref{lem:lambda_binary} indicates that $\sgn(\lambda^{*G}_\alpha)$ can be identified as $s^G=\sgn\big(\E\phi^G(X,A)\1\big(2\eta^G(X,A)>1\big)\big)$ and $|\lambda^{*G}_\alpha|$ is chosen to exhaust the unfairness budget $\alpha$. Although Lemma~\ref{lem:lambda_binary} involves $\eta^G$ and $\phi^G$, the same intuition holds when replacing them with any estimators $\hat\eta^G$ and $\hat\phi^G$, see Section F of the supplement for details. 

In practice, the population unfairness in Lemma \ref{lem:lambda_binary} can be approximated by empirical unfairness, and the induced gap is quantified by the following lemma. 
Let $\{\hat \E_j:j\in[m]\}$ denote the conditional sample averages corresponding to $\{\E_j:j\in[m]\}$ computed from $\D$, and let $n_{(j)}$ be the sample size used for each $\hat \E_j$. Recalling from \eqref{eq_unfairness_characterization} that $\U(f)=|\sum_{j\in[m]}\kappa_j\E_jf(X,A)|$, it empirical approximation is $|\sum_{j\in[m]}\kappa_j\hat\E_jf(X,A)|$. 
Define  
\[\epsilon_\alpha=\sum_{j\in[m]}|\kappa_j|\bigg\{72\sqrt{\frac{2\log 4e^2}{n_{(j)}}}+\sqrt{\frac{1}{2n_{(j)}}\log\frac{2m}{\delta_{\rm post}}}\bigg\}. \]
Then, to ensure population unfairness below $\alpha$, it suffices to constrain the empirical unfairness below $\alpha-\epsilon_\alpha$. Importantly, $\epsilon_\alpha$ is fully distribution-free, enabling finite-sample and distribution-free fairness guarantees.

\begin{Lemma}\label{lem:unfairness_dev_binary}  
Under the model set-up described above, given any estimators $\hat\eta^G$ and $\hat\phi^G$, with probability at least $1-\delta_{\rm post}$ over the randomness of $\D$, we have
    $\sup_{\lambda\in\R}\big|\sum_{j\in[m]}\kappa_j(\hat\E_j-\E_{j})\1\big(2\hat\eta^G(X,A)-1>\lambda\hat\phi^G(X,A)\big)\big|\le\epsilon_\alpha$.
\end{Lemma}
Lemma~\ref{lem:unfairness_dev_binary} follows from the fact that, given $\hat\eta^G$ and $\hat\phi^G$, the function class $\big\{\1(2\hat\eta^G-1>\lambda\hat\phi^G):\lambda\in\R\big\}$ has VC dimension at most 2. Here we are not trying to find the tightest $\epsilon_\alpha$, the main message is that $\epsilon_\alpha$ is distribution-free and of order $O_P\big(\sqrt{\frac{\log({1}/{\delta_{\rm post}})}{n}}\big)$.

Motivated by Lemmas~\ref{lem:lambda_binary} and \ref{lem:unfairness_dev_binary}, we first estimate the sign $s^G$ by
$\hat s^G=\sgn\big(\sum_{j\in[m]}\kappa_j\hat\E_j\1\big(2\hat\eta^G(X,A)>1\big)\big)$,
and set $\hat\lambda^G=\hat s^G\hat\lambda_+^G$, where $\hat\lambda_+^G\ge 0$ is the smallest $\lambda_+$ satisfying
$\hat s^G\sum_{j\in[m]}\kappa_j\hat\E_j\1\big(2\hat\eta^G(X,A)-1>\hat s^G\lambda_+\hat\phi^G(X,A)\big)\le\alpha-\epsilon_\alpha$.
The resulting classifier is 
$\hat f^G_\alpha(x,a)=\1\big(2\hat\eta^G(x,a)-1>\hat\lambda^G\hat\phi^G(x,a)\big)$. 
We summarize the procedures in Algorithm \ref{alg:binary_unified}. Although derived from the randomized classifier (Definition \ref{def:randomized_classifier}), the constructed $\hat f_\alpha^G$ is not randomized.

\begin{algorithm}
\small
\caption{Post-Processing with Binary Sensitive Attribute}
\label{alg:binary_unified}
\begin{algorithmic}
\State{\bf Input:} Data $\D$, initial estimators $\hat\eta^G,\hat\phi^G$, the unfairness level $\alpha$, the tolerance $\delta_{\rm post}$, and the scenario $G\in\{{\rm aware, blind}\}$.
\State{\bf Step 1:} Set $\hat s^G=\sgn\big(\sum_{j\in[m]}\kappa_j\hat\E_j\1\big(2\hat\eta^G(X,A)>1\big)\big)$.
\State{\bf Step 2:} Solve
\[\hat\lambda_+^G=\argmin_{\lambda_+\ge 0}\lambda_+\quad{\rm s.t.}\quad\hat s^G\sum_{j\in[m]}\kappa_j\hat\E_j\1\big(2\hat\eta^G(X,A)-1>\hat s^G\lambda_+\hat\phi^G(X,A)\big)\le\alpha-\epsilon_\alpha.\]
\State{\bf Step 3:} Set $\hat\lambda^G=\hat s^G\hat\lambda_+^G$.
\State{\bf Output:} $\hat f^G_\alpha=\1\big(2\hat\eta^G-1>\hat\lambda^G\hat\phi^G\big)$.

\end{algorithmic}
\end{algorithm}

\begin{Remark}
    \begin{enumerate}
        \item {\cite{chen2024posthoc}, \cite{zeng2024bayes} and \cite{xian2024unified} also proposed post-processing algorithms. \cite{chen2024posthoc} doesn't fully exploit the structure of Bayes optimal fair classifiers, relying instead on empirical risk minimization under empirical fairness constraints. \cite{zeng2024bayes} controls fairness via the second formulation of unfairness measures in \eqref{eq_unfairness_characterization}, which limits the tightness of their fairness control, see Section \ref{sec:numerical}. \cite{xian2024unified} constructs classifiers pointwise, requiring linear programming for each prediction. All three achieve only $(\alpha+o(1),\delta)$-fairness and thus control fairness asymptotically. In addition, \cite{xian2024unified} depends on the multicalibration error of a pretrained model.}
        
        \item As noted earlier, Problem~\eqref{eq:bayes_obj} is generally nonconvex in the parameters of $f$. However, Algorithm~\ref{alg:binary_unified} reduces it to a one-dimensional problem in $\lambda_+$, regardless of the complexity of $\hat\eta^G$ and $\hat\phi^G$. As will become clear in Section \ref{sec:binary_eoo_aware}, in the group-aware scenario, the term $\hat s^G\sum_{j\in[m]}\kappa_j\hat\E_j\1\big(2\hat\eta^G(X,A)-1>\hat s^G\lambda_+\hat\phi^G(X,A)\big)$ in Step 2 of Algorithm \ref{alg:binary_unified} is typically non-increasing in $\lambda_+$, allowing efficient bisection search. In the group-blind scenario, the constraint in Step 2 may lose monotonicity, and grid search can be applied instead. In practice, we still adopt the bisection method, which performs well empirically.
    \end{enumerate}
\end{Remark}

\subsubsection{Fairness Guarantee}\label{sec:fairness}

To study the performance of the proposed algorithm, we begin by introducing some notation. Let $\epsilon_\phi$ represent the estimation error of the given initial estimator $\hat\phi^G$:
$\norm{\hat\phi^G-\phi^G}_\infty\le\epsilon_\phi.$
Assumption~\ref{ass:init} requires the initial estimators $2\hat\eta^G-1$ and $\hat\phi^G$ to be nowhere perfectly aligned. 
\begin{Assumption}[Initial Estimators]\label{ass:init}
    Given $\hat\eta^G$ and $\hat\phi^G$, we assume
    $\sup_{\lambda\in\R}\Prob(2\hat\eta^G(X,A)-1=\lambda\hat\phi^G(X,A))=0$.
\end{Assumption}
Assumption~\ref{ass:init} is mild. For example, when $X|A$ are continuous random vectors, one can always slightly perturb $\hat\eta^G$ and $\hat\phi^G$ to meet Assumption~\ref{ass:init} (see Section H of the supplement). 

The existence of $\hat\lambda_+^G$ in Algorithm~\ref{alg:binary_unified} relies on the monotonicity of $\hat s^G\E\phi^G(X,A)\1\big(2\hat\eta^G(X,A)-1>\hat s^G\lambda_+\phi^G(X,A)\big)$ in $\lambda_+$. In Algorithm~\ref{alg:binary_unified}, we replace the expectations $\E_j$ with sample averages $\hat\E_j$ and $\phi^G$ with $\hat\phi^G$. By Lemma~\ref{lem:unfairness_dev_binary}, the effect of $\hat\E_j$ is controlled by $\epsilon_\alpha$. If $\hat\phi^G\phi^G>0$, monotonicity of $\hat s^G\E\phi^G(X,A)\1\big(2\hat\eta^G(X,A)-1>\hat s^G\lambda_+\hat\phi^G(X,A)\big)$ is preserved, ensuring the existence of $\hat\lambda_+^G$. To handle $\hat\phi^G\phi^G\le 0$, we define the $\phi^G$-weighted margin 
$\tilde\epsilon_\phi^G=\E|\phi^G(X,A)|\1\big(\phi^G(X,A)\hat\phi^G(X,A)\le0\big)\le\E|\phi^G(X,A)|\1(|\phi^G(X,A)|\le\epsilon_\phi)\le\epsilon_\phi$.
Since $\epsilon_\phi$ is small, so is $\tilde\epsilon^G_\phi$. Moreover, in the group-aware case, $\sgn(\phi^{\rm aware})$ is known (as discussed after Example~\ref{exa:phi_eoo}), implying $\tilde\epsilon_\phi^{\rm aware}=0$ (see Section~\ref{sec:binary_eoo_aware} for example). Finally, we impose $\alpha\ge 2\epsilon_\alpha+\tilde\epsilon_\phi^G$ in Theorem~\ref{thm:fair_binary} to ensure that the effects of $\hat\E_j$ and $\hat\phi^G$ are small compared to $\alpha$. 

Given initial estimators $\hat\eta^G$ and $\hat\phi^G$, we show that the classifier $\hat f^G_\alpha$ achieves $(\alpha, \delta_{\rm post})$-fairness in a finite-sample and distribution-free manner, as long as $\alpha$ is not too small.
\begin{Theorem}[Fairness Guarantee]\label{thm:fair_binary} 
    Given $\hat\eta^G$ and $\hat\phi^G$ that satisfy Assumption~\ref{ass:init}, with probability at least $1-\delta_{\rm post}$, for any $\alpha\ge 2\epsilon_\alpha+\tilde\epsilon_\phi^G$, Algorithm~\ref{alg:binary_unified} has a unique output $\hat f_\alpha^G$ satisfying $\U(\hat f_\alpha^G)\le\alpha$.
\end{Theorem}

As shown in the minimax excess risk lower bounds (Theorems \ref{thm:lower_binary_eoo} and O.1 in the supplement), the constraint $\alpha\gtrsim\epsilon_\alpha+\tilde\epsilon_\phi^G$ is necessary for achieving vanishing excess risk. If $\alpha\lesssim\epsilon_\alpha+\tilde\epsilon_\phi^G$, any $(\alpha,\delta)$-fair algorithm incurs constant excess risk under certain distributions.

\subsubsection{Excess Risk Analysis}\label{sec:excess}

Beyond fairness, the constructed classifier should also achieve high predictive accuracy. To study the prediction performance of the proposed algorithm, we introduce a set of assumptions. The following margin condition characterizes the difficulty of the classification problem \citep{tsybakov2004optimal}, which ensures that most data points lie far from the classification boundary $g^{*G}_\alpha=0$ of the Bayes optimal $\alpha$-fair classifier $f^{*G}_\alpha=\1(g^{*G}_\alpha>0)$. 
\begin{Assumption}[Margin Assumption]\label{ass:margin}
    There exist $\gamma\ge 0$ and constant $c_1>0$ such that for any $\epsilon\ge 0$, we have
    $\Prob(|g^{*G}_\alpha (X,A)|\le\epsilon)\le c_1 \epsilon^{\gamma}$.
\end{Assumption}
Assumption \ref{ass:margin} implies $\Prob(g^{*G}_\alpha(X,A)=0)=0$. Recall $s^G=\sgn(\lambda_\alpha^{*G})$ and define the signed unfairness of the classifier $\1\big(2\eta^G-1>\lambda\phi^G\big)$ as 
\begin{equation}\label{eq:signed_unfairness}
    U(\lambda)=s^G\E\phi^G(X,A)\1\big(2\eta^G(X,A)-1>\lambda\phi^G(X,A)\big). 
\end{equation}
Assumption \ref{ass:ratio_poly_simple} then requires the unfairness difference 
$D(z)=|U(\lambda^{*G}_\alpha +z)-U(\lambda^{*G}_\alpha )|$
to grow at least polynomially in $z$ (see Remark~\ref{rem:poly_growth}). This condition enables control of $|\hat\lambda^G-\lambda^{*G}_\alpha |$ when $\U(\hat f)$ approaches $\alpha$. Similar assumptions have also been imposed in \cite{tong2013plug} for Neyman-Pearson classification, where explicit polynomial lower bounds are specified.

\begin{Assumption}[Polynomial Growth]\label{ass:ratio_poly_simple}
    Suppose $\frac{2\eta^G(X,A)-1}{\phi^G(X,A)}$ is a continuous random variable given $\phi^G(X,A)\ne 0$. For some constant $c_2>0$, any $z\in\R$, we have 
    $D(4z)\le c_2 D(z)$.
\end{Assumption}

\begin{Remark}\label{rem:poly_growth}
    \begin{enumerate}
        \item We assume the continuity of $\frac{2\eta^G(X,A)-1}{\phi^G(X,A)}$ for simplicity. Section K of the supplement provides the original versions of Assumptions \ref{ass:ratio_poly_simple} and \ref{ass:ratio_balance_simple} without this condition.
        \item The constant $4$ in Assumption \ref{ass:ratio_poly_simple} is arbitrary and can be replaced by any constant greater than 1. We choose 4 for convenience in proving Theorem~\ref{thm:upper_unify_binary}.
        \item Assumption \ref{ass:ratio_poly_simple} implies $D(z)\gtrsim |z|^{\log_4 c_2}$, meaning the unfairness difference $D(z)$ admits a polynomial lower bound (see Section J of the supplement).
        Since $\phi^G$ is bounded, the margin assumption (Assumption~\ref{ass:margin}) further gives 
        $D(z)\lesssim \Prob\big(|g^{*G}_\alpha (X,A)|<c|z|\big)\lesssim |z|^\gamma$,
        showing $D(z)$ is also bounded from above by some polynomial. 
    \end{enumerate}
\end{Remark}

We then introduce Assumption~\ref{ass:ratio_balance_simple} below. 
\begin{Assumption}\label{ass:ratio_balance_simple}
    There exist constants $c_3,c_4>0$ such that
    $(1+c_3^{-1})\big\{U(0)-U(\lambda^{*G}_\alpha)\big\}\le U(0)-U\big((1+c_4)\lambda^{*G}_\alpha\big)$.
\end{Assumption}

\begin{Remark} 
    When $\lambda^{*G}_\alpha=0$, Assumption~\ref{ass:ratio_balance_simple} holds trivially. If $\lambda^{*G}_\alpha\ne 0$, then $U(0)$ is the unfairness of the unconstrained Bayes optimal classifier $\1(2\eta>1)$, $U(\lambda^{*G}_\alpha)=\alpha$ and $U\big((1+c_4)\lambda^{*G}_\alpha\big)\le\alpha$. Note that the classifier $\1(2\eta^G-1>\lambda\phi^G)$ is a translation of $\1(2\eta^G>1)$ by $\phi^G$ with magnitude $|\lambda|$, and $U(0)-U(\lambda)$ is the unfairness reduction due to the translation. For $U\big((1+c_4)\lambda^{*G}_\alpha)\ge 0$, Assumption~\ref{ass:ratio_balance_simple} ensures that achieving a more stringent unfairness level $U\big((1+c_4)\lambda^{*G}_\alpha\big)$ with unfairness difference $U(0)-U\big((1+c_4)\lambda^{*G}_\alpha\big)$ comparable to $U(0)-U(\lambda^{*G}_\alpha)$, a translation with magnitude $(1+c_4)|\lambda^{*G}_\alpha|$ comparable to $|\lambda^{*G}_\alpha|$ is sufficient. Since $1\ge U(0)\ge U(\lambda^{*G}_\alpha )\ge U((1+c_4)\lambda^{*G}_\alpha )\ge -1$, Assumption~\ref{ass:ratio_balance_simple} holds trivially when $U(\lambda^{*G}_\alpha )-U((1+c_4)\lambda^{*G}_\alpha)$ exceeds a positive constant.
\end{Remark}

Let $c_\phi=\|\phi^G\|_\infty$ and $D_0=\U\big(\1(2\eta^G>1)\big)-\alpha$, then $D_0$ is the difference between the unfairness of the unconstrained Bayes optimal classifier $\1(2\eta^G>1)$ and the target level $\alpha$. If $D_0\le 0$, $\1(2\eta^G>1)$ is already $\alpha$-fair, so $f^{*G}_\alpha =\1(2\eta^G>1)$ and $\lambda^{*G}_\alpha =0$, otherwise, if $D_0>0$, $\1(2\eta^G>1)$ is not $\alpha$-fair and need to be adjusted by $\lambda^{*G}_\alpha \phi^G$. Let $\epsilon_\eta$ denote the estimation error of the given initial estimator $\hat\eta^G$,
$\norm{\hat\eta^G-\eta^G}_\infty\le\epsilon_\eta.$ 
Define the $\phi^G$-weighted margin of $2\eta^G-1$ as
$\tilde\epsilon_\eta^G=\E|\phi^G(X,A)|\1(|2\eta^G(X,A)-1|\le2\epsilon_\eta)$.
Similar to $\tilde\epsilon^G_\phi$ in Section~\ref{sec:fairness}, $\tilde\epsilon^G_\eta$ measures the effect of replacing $\eta^G$ by $\hat \eta^G$ in the unfairness measure and remains small if $2\eta^G-1$ is not overly concentrated around zero.

The following theorem controls the excess risk of $\hat f^G$ in the case where $D_0$ is not too close to 0, i.e., when $\1(2\eta^G>1)$ is sufficiently fair or unfair.

\begin{Theorem}[Excess Risk Upper Bound]\label{thm:upper_unify_binary}
    
    Given $\hat\eta^G,\hat\phi^G$, if Assumptions \ref{ass:init}, \ref{ass:margin}, \ref{ass:ratio_poly_simple} and \ref{ass:ratio_balance_simple} hold, then with probability at least $1-\delta_{\rm post}$, for any $\alpha$ with $\alpha\ge 2\epsilon_\alpha+\tilde\epsilon_\phi^G$ and such that the unfairness difference $D_0=\U\big(\1(2\eta^G>1)\big)-\alpha$ satisfies either
    $D_0\le -2\epsilon_\alpha-\tilde\epsilon_\eta^G$ or $D_0> \tilde\epsilon_\eta^G\vee c_3\big(2\epsilon_\alpha+c_\phi c_1(2\epsilon_\eta+(1+2c_4)|\lambda^{*G}_\alpha|\epsilon_\phi)^\gamma\big)$,
    we have
    \begin{equation}\label{eq:excess_risk_upper_binary_unify}
        \risk(\hat f^G_\alpha)-\risk(f^{*G}_\alpha)\lesssim|\lambda^{*G}_\alpha|\epsilon_\alpha+\epsilon_\eta^{1+\gamma}+(|\lambda^{*G}_\alpha|\epsilon_\phi)^{1+\gamma}. 
    \end{equation}

\end{Theorem}
The bound reveals that the statistical cost of enforcing fairness is amplified by the factor $|\lambda^{*G}_\alpha|$, which serves as a sensitivity parameter of the optimal decision rule to fairness constraints. In particular, the difficulty of fair classification is not only governed by the estimation error of nuisance quantities, but also by how aggressively the fairness constraint modifies the decision boundary.
\begin{Remark}
    If $\alpha\ge\U\big(\1(2\eta^G>1)\big)$, the unconstrained Bayes classifier $\1(2\eta^G>1)$ is already $\alpha$-fair and $\lambda^{*G}_\alpha=0$. Then the upper bound~\eqref{eq:excess_risk_upper_binary_unify} reduces to $O_P(\epsilon_\eta^{1+\gamma})$, matching the optimal rate for unconstrained classification up to logarithmic factors \citep{audibert2007fast}. 
    
    When the fairness constraint tightens so that $\alpha<\U\big(\1(2\eta^G>1)\big)$, we have $\lambda^{*G}_\alpha\ne 0$. As shown later, the upper bound~\eqref{eq:excess_risk_upper_binary_unify} is minimax optimal up to logarithmic factors. Comparing it with the unconstrained excess risk $O_P(\epsilon_\eta^{1+\gamma})$ reveals that enforcing fairness incurs a cost in excess risk with order $O_P\big(|\lambda^{*G}_\alpha|\epsilon_\alpha+(|\lambda^{*G}_\alpha|\epsilon_\phi)^{1+\gamma}\big)$, which typically grows as $\alpha$ decreases. Moreover, when $|\lambda^{*G}_\alpha|\gtrsim 1$, excess risk faster than $O_P(n^{-\frac{1}{2}})$ can not be attained, even for large $\gamma$ (i.e., when most data lie far from the boundary $g^{*G}_\alpha=0$).
\end{Remark}

\section{Applications to Equality of Opportunity}\label{sec:binary_eoo}
In this section, we apply the framework from Section~\ref{sec:unify} to EOO (see Definition \ref{def:unfairness_eoo_binary} and Example \ref{exa:phi_eoo}) with binary sensitive attributes under both group-aware and group-blind scenarios. An auxiliary dataset $\tilde\D = \{(\tilde X_i, \tilde A_i, \tilde Y_i): i \in [\tilde n]\}$, independently drawn from $P_{X, A, Y}$, is used to train the initial estimators $\hat\eta$ and $\hat\phi$, which are then refined on $\D$ following Algorithm~\ref{alg:binary_unified}. The combined dataset is denoted by $\D_{\rm all} = \D \cup \tilde\D$.
Sections \ref{sec:binary_eoo_aware} and \ref{sec:binary_eoo_blind} specify $\epsilon_\eta$ and $\epsilon_\phi$ in \eqref{eq:excess_risk_upper_binary_unify} under H\"older smoothness assumptions and provide explicit excess risk upper bounds \eqref{eq:excess_risk_upper_binary_unify} for group-aware and group-blind cases, respectively. Section~\ref{sec:binary_eoo_lower} establishes the corresponding minimax lower bounds. Comparing these bounds reveals the cost of group-blindness in excess risk. Throughout the section, $X$ is supported on $\X\subset[0,1]^d$.

\subsection{Group-Aware Excess Risk Upper Bound}\label{sec:binary_eoo_aware}

In this section, we apply the framework from Section~\ref{sec:unify} to EOO in the group-aware scenario. For simplicity, we omit the superscript ``aware" and write any group-aware function $f^{\rm aware}(X,A)$ as $f(X,A)$. 

Recall that $\eta(X,A)=\Prob(Y=1|X,A)$. From Example~\ref{exa:phi_eoo} and Proposition~\ref{prop:bayes_binary}, 
$\phi(x,a)=\frac{(3-2a)\eta(x,a)}{p_{1,a}}$, the Bayes optimal $\alpha$-fair classifier equals
$f^*_\alpha (x,a)=\1\big(g^*_\alpha (x,a)>0\big)$, $g^*_\alpha (x,a)=\big(2+\frac{(2a-3)\lambda^*_\alpha}{p_{1,a}}\big)\eta(x,a)-1$.
Recall $s=\sgn(\lambda^*_\alpha)$. As shown in Section M of the supplement, the group-aware $|\lambda^*_\alpha|$ is bounded by 1, and $f^*_\alpha$ can be equivalently expressed as a group-wise thresholding rule \citep{corbett2017algorithmic, menon2018cost, zeng2024bayes}, 
\begin{equation}\label{eq:lambda_upper_bound_eoo_aware}
    |\lambda^*_\alpha|\le p_{1,\frac{3-s}{2}},\quad f^*_\alpha(x,a)=\1\bigg(\eta(x,a)>\bigg(2+\frac{(2a-3)\lambda^*_\alpha}{p_{1,a}}\bigg)^{-1}\bigg). 
\end{equation}

To construct initial estimators, we make H\"older smoothness assumptions on $\eta(\cdot,a)$, $a\in[2]$.

\begin{Definition}[H\"older Class]\label{def:holder}
    Let $L>0$, the $(\beta,L)$-H\"older class, denoted $\HH(\beta,L)$, consists of all functions $g:[0,1]^d\rightarrow \R$ that are $\lfloor\beta\rfloor$ times differentiable and satisfy
    $\abs{g(x')-g_{\lfloor\beta\rfloor,x}(x')}\le L\norm{x-x'}_2^\beta$, $\forall x,x'\in[0,1]^d$,
    where $g_{\lfloor\beta\rfloor,x}$ is the degree $\lfloor\beta\rfloor$ Taylor polynomial of $g$ at $x$.
\end{Definition}
\begin{Assumption}[H\"older Smoothness]\label{ass:holder_aware}
    We assume $\eta(\cdot,1),\eta(\cdot,2)\in\HH(\beta_A,L_Y)$.
\end{Assumption}

We further impose a strong density assumption on $X|A$, originally introduced by \cite{audibert2007fast} and widely used in nonparametric classification \citep{cai2021transfer,kpotufe2018marginal}. 

\begin{Assumption}[Strong Density Assumption]\label{ass:density_aware}
    Recall ${\rm Leb}(\cdot)$ denote the Lebesgue measure on $\R^d$ and $B_2(c,r)$ is the $l_2$ ball in $\R^d$ centered at $c$ with radius $r$. Assume $X$ conditioned on $A$ has density $p_{X|A}$, and there exist constants $c_X,c_\mu,r_\mu>0$ such that 
    \[c_X\le p_{X|1}(x),p_{X|2}(x)\le c_X^{-1},\quad{\rm Leb}\big(\X \cap B_2(x,r)\big)\ge c_\mu{\rm Leb}\big(B_2(x,r)\big),\quad\forall 0<r\le r_\mu,\forall x\in\X. \]
\end{Assumption}

To ensure enough data for estimating $\eta$ and $\phi$, we assume the probabilities for observing each group are large enough.

\begin{Assumption}[Observability]\label{ass:observe}
    There is a constant $c_5>0$ such that $p_{1,1},p_{1,2}>c_5$.
\end{Assumption}

Under Assumptions~\ref{ass:holder_aware}, \ref{ass:density_aware} and \ref{ass:observe}, we get $\hat\eta(\cdot, a)$ for $\eta(\cdot, a)$ via local polynomial regression \citep{Tsybakov2009introduction,fan2018local}, see Section N in the supplement for details. 
Let $n_{1,a}=\sum_{i\in[n]}\1(Y_i=1,A_i=a)$, $\tilde n_{1,a}=\sum_{i\in[\tilde n]}\1(\tilde Y_i=1,\tilde A_i=a)$, $a\in[2]$. We estimate $p_{1,a}$ by $\hat p_{1,a}=\frac{\tilde n_{1,a}}{\tilde n}$ and $\phi$ by $\hat\phi(x,a)=\frac{(3-2a)\hat\eta(x,a)}{\hat p_{1,a}}$. The estimation errors satisfy $\epsilon_\eta\asymp\epsilon_\phi\asymp\big(\frac{d\log \tilde n+\log\frac{1}{\delta_{\rm init}}}{\tilde n}\big)^{\frac{\beta_Y}{2\beta_Y+d}}$. Without the loss of generality, we suppose $\hat\eta(x,a)>0$ for all $(x,a)\in[0,1]^d\times[2]$, otherwise, we replace $\hat\eta(x,a)$ by $\epsilon_\eta$ whenever $\hat\eta(x,a)=0$. Then it is clear that
$\tilde\epsilon_\phi=\E|\phi(X,A)|\1\big(\phi(X,A)\hat\phi(X,A)\le0\big)=0$. Moreover, since the sign of $\hat\phi(x,a)$ is determined by $a$, the term $\hat s(\hat\E_{X|Y=1,a=1}-\hat\E_{X|Y=1,a=2})\1(2\hat\eta-1>\hat s\lambda_+\hat\phi)$ in Step 2 of Algorithm \ref{alg:binary_unified} is non-increasing in $\lambda_+$, so a bisection search can be used for Algorithm \ref{alg:binary_unified}.

Following Lemma~\ref{lem:unfairness_dev_binary}, we define  
\begin{equation}\label{eq:epsilon_alpha_eoo}
    \epsilon_\alpha=72\sqrt{\frac{2\log 4e^2}{n_{1,1}}}+72\sqrt{\frac{2\log 4e^2}{n_{1,2}}}+\sqrt{\frac{1}{2n_{1,1}}\log \frac{4}{\delta_{\rm post}}}+\sqrt{\frac{1}{2n_{1,2}}\log \frac{4}{\delta_{\rm post}}}.
\end{equation}
Recall $\delta=\delta_{\rm init}+\delta_{\rm post}$. Let $\hat f_\alpha$ be the classifier from Algorithm~\ref{alg:binary_unified}, following the notations in Theorems~\ref{thm:fair_binary} and \ref{thm:upper_unify_binary}, we have the following excess risk control.

\begin{Corollary}[Group-aware Excess Risk Upper Bound]\label{cor:binary_aware_eoo}
    Suppose Assumptions~\ref{ass:init}, \ref{ass:margin}, \ref{ass:ratio_poly_simple}, \ref{ass:ratio_balance_simple}, \ref{ass:holder_aware}, \ref{ass:density_aware}, and \ref{ass:observe} hold. Then with probability at least $1-\delta$ over all samples $\D_{\rm all}$, for any $\alpha\ge2\epsilon_\alpha$ such that the unfairness difference $D_0=\U(\1(2\eta>1))-\alpha$ satisfies either
    $D_0\le -2\epsilon_\alpha-\tilde\epsilon_\eta$ or $D_0> \tilde\epsilon_\eta\vee c_3\big(2\epsilon_\alpha+\frac{c_1}{c_5}(2\epsilon_\eta+(1+2c_4)|\lambda^*_\alpha|\epsilon_\phi)^\gamma\big)$,
    where the constants $c_i$ are defined in Assumptions~\ref{ass:margin}, \ref{ass:ratio_balance_simple} and \ref{ass:observe}, 
    we have
     \begin{equation}\label{eq:excess_risk_upper_binary_eoo_aware}
        \begin{aligned}
            \risk(\hat f_\alpha)-\risk(f^*_\alpha)\lesssim|\lambda^*_\alpha|\sqrt{\frac{\log\frac{1}{\delta_{\rm post}}}{n}}+\bigg(\frac{d\log \tilde n+\log\frac{1}{\delta_{\rm init}}}{\tilde n}\bigg)^{\frac{\beta_Y(1+\gamma)}{2\beta_Y+d}}. 
        \end{aligned}
        \end{equation}

\end{Corollary}

\subsection{Group-Blind Excess Risk Upper Bound}\label{sec:binary_eoo_blind}

In this section, we study equality of opportunity with binary sensitive attributes in the group-blind scenario. For simplicity, we omit the superscript ``blind" and the argument $A$, and denote any group-blind function $f^{\rm blind}(X,A)$ simply as $f(X)$. 

Recall $\eta(X)=\Prob(Y=1|X)$ and $\rho_{a|y}(X)=\Prob(A=a|X,Y=y)$.
Example~\ref{exa:phi_eoo} and Proposition~\ref{prop:bayes_binary} imply
$\phi(x)=\big(\frac{\rho_{1|1}(x)}{p_{1,1}}-\frac{\rho_{2|1}(x)}{p_{1,2}}\big)\eta(x)$
and the Bayes optimal $\alpha$-fair classifier is 
\[f^*_\alpha(x)=\1\big(g^*_\alpha(x)>0\big),\quad g^*_\alpha(x)=\bigg\{2-\lambda^*_\alpha\bigg(\frac{\rho_{1|1}(x)}{p_{1,1}}-\frac{\rho_{2|1}(x)}{p_{1,2}}\bigg)\bigg\}\eta(x)-1. \]
Hence, the group-blind Bayes optimal $\alpha$-fair classifier is also a group-wise thresholding rule, but we need to infer the group $A$ and adjust the threshold according to the confidence $\big|\frac{\rho_{1|1}(x)}{p_{1,1}}-\frac{\rho_{2|1}(x)}{p_{1,2}}\big|$ of the prediction.

Similar to the group-aware scenario in Section~\ref{sec:binary_eoo_aware}, we make the following assumptions.

\begin{Assumption}[H\"older Smoothness]\label{ass:holder_blind}
    We assume $\eta$ and $\rho_{1|1}$ are both H\"older smooth with $\eta\in\HH(\beta_Y,L_Y)$ and $\rho_{1|1}\in\HH(\beta_A,L_A)$.
\end{Assumption}

\begin{Assumption}[Strong Density Assumption]\label{ass:density_blind}
    We assume $X$ has density $p_X$, and there exist constants $c_X, c_\mu, r_\mu>0$ such that 
    \[c_X\le p_X(x)\le c_X^{-1},\quad {\rm Leb}\big(\X\cap B_2(x,r)\big)\ge c_\mu{\rm Leb}\big(B_2(x,r)\big),\quad\forall 0<r\le r_\mu,\forall x\in\X. \]
\end{Assumption}

Then we use local polynomial regression to estimate $\eta$ and $\rho_{1|1}$ by $\hat\eta$ and $\hat\rho_{1|1}$, see Section N in the supplement. Let $n_{1,a}=\sum_{i\in[n]}\1(Y_i=1,A_i=a)$, $\tilde n_Y=\sum_{i\in[\tilde n]}\1(\tilde Y_i=1)$, $\tilde n_{1,a}=\sum_{i\in[\tilde n]}\1(\tilde Y_i=1,\tilde A_i=a)$, $a\in[2]$. We estimate $p_Y,p_{1,1}, p_{1,2}$ by $\hat p_Y=\frac{\tilde n_Y}{\tilde n},\hat p_{1,a}=\frac{\tilde n_{1,a}}{\tilde n}$, respectively, and estimate $\phi$ by $\hat\phi=\frac{\hat p_Y\hat\rho_{1|1}-\hat p_{1,1}}{\hat p_{1,1}\hat p_{1,2}}\hat\eta$. 
The corresponding estimation errors are
$\epsilon_\eta\asymp\big(\frac{d\log \tilde n+\log\frac{1}{\delta_{\rm init}}}{\tilde n}\big)^{\frac{\beta_Y}{2\beta_Y+d}}$, $\epsilon_\rho\asymp\big(\frac{d\log \tilde n+\log\frac{1}{\delta_{\rm init}}}{\tilde n}\big)^{\frac{\beta_A}{2\beta_A+d}}$, $\epsilon_\phi\asymp\epsilon_\eta+\epsilon_\rho$.
We use the same $\epsilon_\alpha$ as in Equation~\eqref{eq:epsilon_alpha_eoo} and set $\delta=\delta_{\rm init}+\delta_{\rm post}$. For the classifier $\hat f_\alpha$ from Algorithm~\ref{alg:binary_unified}, following the notations in Theorems~\ref{thm:fair_binary} and \ref{thm:upper_unify_binary}, we have the following excess risk control.

\begin{Corollary}[Group-blind Excess Risk Upper Bound]\label{cor:binary_blind_eoo}
    Suppose Assumptions~\ref{ass:init}, \ref{ass:margin}, \ref{ass:ratio_poly_simple}, \ref{ass:ratio_balance_simple}, \ref{ass:observe}, \ref{ass:holder_blind}, and \ref{ass:density_blind} hold. Then with probability at least $1-\delta$ over all the samples $\D_{\rm all}$, for any $\alpha\ge2\epsilon_\alpha+\tilde\epsilon_\phi$ such that the unfairness difference $D_0=\U(\1(2\eta>1))-\alpha$ satisfies either
    $D_0\le -2\epsilon_\alpha-\tilde\epsilon_\eta$ or $D_0> \tilde\epsilon_\eta\vee c_3\big(2\epsilon_\alpha+\frac{c_1}{c_5}(2\epsilon_\eta+(1+2c_4)|\lambda^*_\alpha|\epsilon_\phi)^\gamma\big)$,
    where the constants $c_i$ are defined in Assumptions~\ref{ass:margin}, \ref{ass:ratio_balance_simple} and \ref{ass:observe}, we have 
    \begin{equation}\label{eq:excess_risk_upper_binary_eoo_blind}
    \begin{aligned}
        \risk(\hat f_\alpha)-\risk(f^*_\alpha)\lesssim&|\lambda^*_\alpha|\sqrt{\frac{\log\frac{1}{\delta_{\rm post}}}{n}}+|\lambda^*_\alpha|^{1+\gamma}\bigg(\frac{d\log \tilde n+\log\frac{1}{\delta_{\rm init}}}{\tilde n}\bigg)^{\frac{\beta_A(1+\gamma)}{2\beta_A+d}}\\
        &+\big(1+|\lambda^*_\alpha|\big)^{1+\gamma}\bigg(\frac{d\log \tilde n+\log\frac{1}{\delta_{\rm init}}}{\tilde n}\bigg)^{\frac{\beta_Y(1+\gamma)}{2\beta_Y+d}}. 
    \end{aligned}
    \end{equation}

\end{Corollary}

\subsection{Minimax Excess Risk Lower Bound}\label{sec:binary_eoo_lower}

To assess the optimality of our post-processing algorithm and its excess risk upper bounds, we establish the minimax lower bounds for excess risk in this subsection. We begin by defining the parameter space under consideration.

\begin{Definition}[Group-Aware Parameter Space]\label{def:parameter_space_aware}
    Let the group-aware parameter space $\mathscr{P}^{\rm aware}$ be the set of all distributions $P_{X,A,Y}$ satisfying Assumptions~\ref{ass:margin}, \ref{ass:ratio_poly_simple}, \ref{ass:ratio_balance_simple}, \ref{ass:holder_aware}, \ref{ass:density_aware} and \ref{ass:observe}. 
\end{Definition}
\begin{Definition}[Group-Blind Parameter Space]\label{def:parameter_space_blind}
    Let the group-blind parameter space $\mathscr{P}^{\rm blind}$ be the set of all distributions $P_{X,A,Y}$ satisfying Assumptions~\ref{ass:margin}, \ref{ass:ratio_poly_simple}, \ref{ass:ratio_balance_simple}, 
    \ref{ass:observe}, 
    \ref{ass:holder_blind} and \ref{ass:density_blind}. 
\end{Definition}
To compare group-aware and group-blind excess risks, we consider the intersection parameter space $\mathscr{P}=\mathscr{P}^{\rm aware}\cap\mathscr{P}^{\rm blind}$. For quantities depending on a specific distribution $P$, we add the subscript $P$ to emphasize the distribution. For example, $f^{*{\rm aware}}_{\alpha,P}$ (resp. $f^{*{\rm blind}}_{\alpha,P}$) denotes the Bayes optimal $\alpha$-fair group-aware (resp. group-blind) classifier under distribution $P$. 

Recall that $\D_{\rm all}=\tilde\D\cup\D$ contains all samples, including $\tilde\D$ for training initial estimators and $\D$ for post-processing, with total sample size $N=\tilde n+n$. In fair classification, the algorithm must satisfy the following $(\alpha,\delta)$-fairness constraint.

\begin{Definition}[$(\alpha,\delta)$-Fair Algorithms]\label{def:fair_algorithm}
    For $G\in\{{\rm aware,blind}\}$, let algorithm $\A^G$ map the dataset $\D_{\rm all}\sim P_{X,A,Y}^{\otimes N}$ to $[0,1]^{\X\times[2]}$. The class of algorithms satisfying the $(\alpha,\delta)$-fairness constraint is 
    \[\mathscr{A}^G=\big\{\mathcal{A}^G:\Prob_{\D_{\rm all}\sim P^{\otimes N}}\big(\U_{{\rm EOO},P}\big(\mathcal{A}^G(\D_{\rm all})\big)\le\alpha\big)\ge 1-\delta, ~\forall P\in\mathscr{P}^G\big\}. \]
\end{Definition}
Note that $\mathscr{A}^G$ includes all post-processing algorithms where $\tilde\D$ is used for training initial estimators and $\D$ for calibration. Hence, the minimax lower bound over $\mathscr{A}^G$ also applies to all post-processing algorithms.

Under the set of models and algorithms defined above, the following theorem provides minimax lower bounds for the excess risks. Section O of the supplement provides minimax lower bounds for the expected excess risks.

\begin{Theorem}[Minimax Excess Risk Lower Bound]\label{thm:lower_binary_eoo}
    Suppose $\beta_Y\gamma\le d$, $\beta_A\gamma\le d$ and $\beta_Y\le\beta_A$. Consider the parameter space $\mathscr{P}=\mathscr{P}^{\rm aware}\cap\mathscr{P}^{\rm blind}$ and $(\alpha,\delta)$-fair algorithms, then for some constant $c\in(0,1)$, we have 
    \begin{equation}\label{eq:lower_aware}
        \inf_{\A^{\rm aware}\in\mathscr{A}^{\rm aware}}\sup_{P\in\mathscr{P}}\Prob_{\D_{\rm all}\sim P^{\otimes N}}\bigg(\risk_P\big(\A^{\rm aware}(\D_{\rm all})\big)-\risk_P(f^{*{\rm aware}}_{\alpha,P})\gtrsim N^{-\frac{\beta_Y(1+\gamma)}{2\beta_Y+d}}\bigg)\ge c-\delta,
    \end{equation}
    \begin{equation}\label{eq:lower_blind}
        \begin{aligned}
        &\inf_{\A^{\rm blind}\in\mathscr{A}^{\rm blind}}\sup_{P\in\mathscr{P}}\Prob_{\D_{\rm all}\sim P^{\otimes N}}\bigg(\risk_P\big(\A^{\rm blind}(\D_{\rm all})\big)-\risk_P(f^{*{\rm blind}}_{\alpha,P})\gtrsim\\
        &|\lambda^{*{\rm blind}}_{\alpha,P}|(N^{-\frac{1}{2}}\wedge\alpha)+\bigg(|\lambda^{*{\rm blind}}_{\alpha,P}|N^{-\frac{\beta_A}{2\beta_A+d}}\bigg)^{1+\gamma}+\bigg(\big(1+|\lambda^{*{\rm blind}}_{\alpha,P}|\big)N^{-\frac{\beta_Y}{2\beta_Y+d}}\bigg)^{1+\gamma}\bigg)\ge c-\delta.
    \end{aligned}
    \end{equation}
\end{Theorem}

This theorem suggests that the key distinction between the group-aware and group-blind settings lies in the behavior of $\lambda_\alpha^{*G}$ and the prediction error of sensitive attributions. In the group-aware case, $|\lambda_\alpha^{*\rm aware}|$ remains bounded (Equation~\eqref{eq:lambda_upper_bound_eoo_aware}), indicating that fairness constraints induce only a limited perturbation of the decision rule.
In contrast, in the group-blind setting, $|\lambda_\alpha^{*\rm blind}|$ can grow on the order of $\alpha^{-1}$ (Theorem O.1 in the supplement), reflecting a more complex decision boundary. Moreover, the group-blind lower bound \eqref{eq:lower_blind} contains an additional term $O_P(|\lambda^{*{\rm blind}}_\alpha|^{1+\gamma}N^{-\frac{\beta_A(1+\gamma)}{2\beta_A+d}})$ due to the prediction of sensitive attributes.
These divergences explain the fundamentally different excess risk behavior in the two settings and characterize \textbf{the cost of group-blindness}. See Remark O.2 in the supplement for more details.
\begin{Remark}
    The conditions $\beta_Y\gamma\le d$ and $\beta_A\gamma\le d$ are standard in minimax lower bound analyses for nonparametric classification, see, e.g., \citep{audibert2007fast, cai2021transfer}. The condition $\beta_Y\le\beta_A$ is only required when considering the parameter space $\mathscr{P}$. If we study the group-aware and group-blind lower bounds in $\mathscr{P}^{\rm aware}$ and $\mathscr{P}^{\rm blind}$, separately, then the same rates can be proved without assuming $\beta_Y\le\beta_A$.

\end{Remark}
\begin{Remark}\label{rem:minimax}
    The excess risk upper bounds~\eqref{eq:excess_risk_upper_binary_eoo_aware} and \eqref{eq:excess_risk_upper_binary_eoo_blind} contain polynomial factors in $d$. However, in nonparametric statistics, errors typically grow exponentially with the dimension $d$, often assuming $d\lesssim \log N$. Such a condition makes those polynomial factors only logarithmic. 
    When $\alpha\gtrsim N^{-\frac{1}{2}}$, the group-blind excess risk upper bound~\eqref{eq:excess_risk_upper_binary_eoo_blind} matches the minimax lower bound~\eqref{eq:lower_blind} up to logarithmic factors. 
    When $2\beta_Y\gamma\le d$, since $|\lambda^{*{\rm aware}}_\alpha|\le 1$ according to Equation~\eqref{eq:lambda_upper_bound_eoo_aware}, the group-aware upper bound~\eqref{eq:excess_risk_upper_binary_eoo_aware} also matches the minimax lower bound~\eqref{eq:lower_aware} up to logarithmic terms. 
    Therefore, Algorithm~\ref{alg:binary_unified} is minimax optimal up to logarithmic factors. 
\end{Remark}

\begin{Remark}[\textbf{Optimal Trade-off Between Excess Risk and Fairness }]\label{rem:tradeoff_blind}
    The optimal group-blind excess risks decrease with $\alpha$ (see Theorem O.1 in the supplement), revealing a trade-off between fairness and group-blind excess risk. See Section O of the supplement for further discussion.
\end{Remark}

\section{Numerical Studies}\label{sec:numerical}

In this section, we evaluate the proposed algorithm on synthetic and real data under the equality of opportunity constraint, and compare its performance with other state-of-the-art fair classification methods. 

\subsection{Simulation Results}

In the simulation studies, we define $P_{X,A,Y} = P_{X|Y,A}P_{Y,A}$ as follows. For $y\in\{0,1\}$, $a\in[2]$, we generate $(Y,A)$ according to $p_{0,1}=0.3$, $p_{0,2}=0.18$, $p_{1,1}=0.49$, $p_{1,2}=0.12$, then $X-\mu_{Y,A}\in\R^d$ condition on $Y,A$ follows distribution $F$ with $\mu_{0,1},\mu_{0,2},\mu_{1,2}\overset{\rm i.i.d.}{\sim}{\rm Unif}(0,1)^{\otimes d}$ and $\mu_{1,1}\sim{\rm Unif}(b,b+1)^{\otimes d}$, where $d,F,b$ will be specified later. 

We set $\alpha\in\{0.08,0.11,0.14,0.17,0.20\}$ and $\delta=0.05$. 
Then we generate $n$ training samples, $n$ calibration samples, and 5000 test samples from the specified distribution. 
The training samples are used to get initial estimators $\hat\eta^G$ and $\hat\phi^G$, which are then post-processed on the calibration data to produce $\hat f^G=\1(2\hat\eta^G-1>\hat\lambda\hat\phi^G)$. Finally, we evaluate the unfairness of $\hat f^G$ under equality of opportunity and the prediction error using the 5000 test samples. 
The simulation considers three settings for $(d,F,b,n)$ as follows: 
\begin{itemize}
    \item[(M1)] $d=5$, $F=N(0,I_d)$, $b=1$ and $n=1000$.
    \item[(M2)] $d=5$, $F=N(0,I_d)$, $b=0.5$ and $n=500$.
    \item[(M3)] $d=10$, $F=t_3^{\otimes d}$, $b=1$ and $n=1000$. 
\end{itemize}

For the initial estimators, we fit a multinomial logistic regression model to get $\hat P(Y,A|X)$ and plug it into $\eta^{\rm aware}(X,a)=\frac{\Prob(Y=y,A=a|X)}{\Prob(Y=0,A=a|X)+\Prob(Y=1,A=a|X)}$, $\eta^{\rm blind}(X)=\Prob(Y=y,A=1|X)+\Prob(Y=y,A=2|X)$ and $\rho_{a|y}(X)=\frac{\Prob(Y=y,A=a|X)}{\eta^{\rm blind}(X)}$ to construct the initial estimators of $\eta^G$ and $\phi^G$. 
The multinomial logistic model is well-specified for (M1) and (M2) but misspecified for (M3). 
For fairness calibration, we set $\epsilon_\alpha = \sqrt{\frac{\log\frac{1}{\delta}}{n_Y}}$, since the $\epsilon_\alpha$ in Section~\ref{sec:binary_eoo} may be too conservative in practice. 

In the group-blind scenario, we compare our methods with the fair plug-in rule (FPIR) of \cite{zeng2024bayes} and the modification with bias scores (MBS) of \cite{chen2024posthoc}. Since MBS is designed for the group-blind scenario, only FPIR is used for comparison in the group-aware case. All three methods share the same initial estimators as described above. 
Because $\eta^G$ and $\phi^G$ admit closed forms under the considered models, we can calculate the prediction error of the Bayes optimal fair classifier. 

For brevity, we report results only for (M1) and defer those for (M2) and (M3) to Section B of the supplement. Fixing the generated $\mu_{y,a}$'s, we repeat the process 100 times and report the averaged unfairness, the $95\%$ sample quantile of the 100 unfairness measures, and the averaged prediction errors for both group-blind and group-aware scenarios in Tables~\ref{tab:simulation_blind} and \ref{tab:simulation_aware}, respectively. The results show that our algorithm controls the unfairness approximately below $\alpha$ with probability at least 0.95, satisfying the $(\alpha,\delta)$-fairness constraint. In contrast, FPIR and MBS only control the average fairness and may still produce unfair classifiers for individual realizations of the observed data.

\begin{table}[t]
\small
    \centering
    \begin{tabular}{@{}lcccccc@{}}\hline
     & & \multicolumn{5}{c}{$\alpha$}\\\cline{3-7}
    Methods & & 0.08 & 0.11 & 0.14 & 0.17 & 0.20\\\hline
    
    Ours & $\bar\U_{\rm EOO}$ & 0.043(0.026) & 0.052(0.033) & 0.065(0.035) & 0.105(0.046) & 0.131(0.045)\\
     & $\U_{{\rm EOO},95}$ & 0.091 & 0.110 & 0.126 & 0.182 & 0.198\\
     & Error & 0.312(0.023) & 0.294(0.020) & 0.286(0.020) & 0.267(0.019) & 0.253(0.020)\\[6pt]
    FPIR & $\bar\U_{\rm EOO}$ & 0.097(0.058) & 0.129(0.067) & 0.159(0.062) & 0.180(0.061) & 0.213(0.062)\\
     & $\U_{{\rm EOO},95}$ & 0.202 & 0.246 & 0.253 & 0.271 & 0.308\\
     & Error & 0.272(0.028) & 0.256(0.028) & 0.243(0.027) & 0.234(0.026) & 0.221(0.023)\\[6pt]
    MBS & $\bar\U_{\rm EOO}$ & 0.082(0.046) & 0.113(0.043) & 0.132(0.045) & 0.176(0.048) & 0.200(0.045)\\
     & $\U_{{\rm EOO}, 95}$ & 0.157 & 0.175 & 0.196 & 0.253 & 0.274\\
     & Error & 0.278(0.022) & 0.263(0.019) & 0.254(0.020) & 0.235(0.018) & 0.224(0.017)\\[6pt]
     Bayes & Error & 0.263 & 0.249 & 0.235 & 0.223 & 0.217\\\hline
    \end{tabular}
    \caption{Unfairness measures and prediction errors of our method, FPIR, and MBS in the group-blind scenario under (M1). And the prediction errors of Bayes optimal fair classifiers. $\bar\U_{\rm EOO}$ is the average unfairness over 100 repetitions. $\U_{{\rm EOO},95}$ is the $95\%$ sample quantile of the unfairness measures produced by 100 repetitions. Error is the average prediction error. 
    }
    \label{tab:simulation_blind}
\end{table}

\begin{table}[t]
\small
    \centering
    \begin{tabular}{@{}lcccccc@{}}\hline
     & & \multicolumn{5}{c}{$\alpha$}\\\cline{3-7}
    Methods & & 0.08 & 0.11 & 0.14 & 0.17 & 0.20\\\hline
    Ours & $\bar\U_{\rm EOO}$ & 0.035(0.024) & 0.044(0.033) & 0.069(0.040) & 0.103(0.047) & 0.121(0.047)\\
     & $\U_{{\rm EOO},95}$ & 0.078 & 0.105 & 0.136 & 0.179 & 0.192\\
     & Error & 0.183(0.008) & 0.181(0.008) & 0.175(0.008) & 0.171(0.006) & 0.170(0.007)\\[6pt]
    FPIR & $\bar\U_{\rm EOO}$ & 0.112(0.079) & 0.128(0.082) & 0.143(0.087) & 0.168(0.089) & 0.193(0.084)\\
     & $\U_{{\rm EOO},95}$ & 0.252 & 0.264 & 0.297 & 0.295 & 0.329\\
     & Error & 0.179(0.017) & 0.175(0.012) & 0.171(0.009) & 0.169(0.011) & 0.167(0.008)\\[6pt]
     Bayes & Error & 0.161 & 0.160 & 0.157 & 0.156 & 0.155\\\hline
    \end{tabular}
    \caption{The unfairness measures and prediction errors of our method and FPIR, respectively in the group-aware scenario under (M1). The notation is the same as Table~\ref{tab:simulation_blind}.}
    \label{tab:simulation_aware}
\end{table}

Figure~\ref{fig:simulation} summarizes the trade-off between average prediction error and both average and $95\%$ quantile unfairness. 
In both group-aware and group-blind scenarios, the curve of FPIR lies to the right of ours, indicating a worse fairness-accuracy trade-off. This arise because FPIR evaluates empirical unfairness through $|\hat\E\hat\phi^G(X,A)\hat f^G(X,A)|$ rather than $|(\hat\E_{X|A=1,Y=1}-\hat\E_{X|A=2,Y=1})\hat f^G(X,A)|$, making its unfairness estimation error depend on the error of $\hat\phi$, which is larger than $\epsilon_\alpha$ in our method. This is verified by FPIR's larger variance and quantile of unfairness measures in Tables~\ref{tab:simulation_blind} and \ref{tab:simulation_aware}. 
MBS and our method show comparable performance, as both are derived from the Bayes optimal fair classifier. 

\begin{figure}[t]
    \centering
    \begin{subfigure}[b]{\textwidth}
        \includegraphics[width=\linewidth]{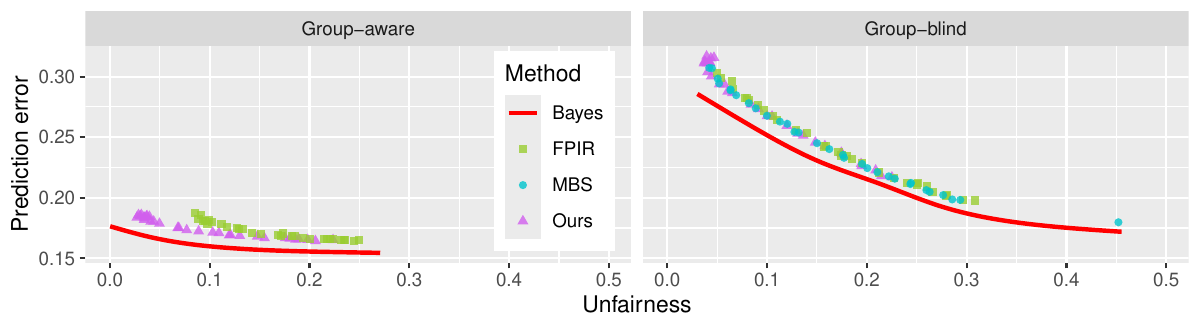}
        \caption{}
        \label{fig:simulation_average}
    \end{subfigure}
    \begin{subfigure}[b]{\textwidth}
        \includegraphics[width=\linewidth]{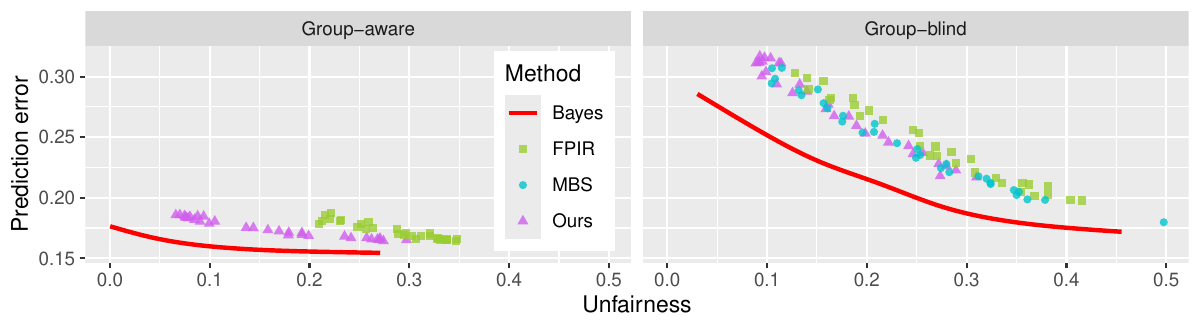}
        \caption{}
        \label{fig:simulation_quantile}
    \end{subfigure}
    \caption{\small{(a) Trade-off between prediction error and unfairness under (M1). The X-axis is the \textbf{average} unfairness $\bar\U_{\rm EOO}$ of classifiers over 100 repetitions, and the Y-axis is their average test prediction errors. The left and right panels correspond to the group-aware and group-blind scenarios, respectively. (b) Same as (a), except the X-axis is the \textbf{$95\%$ sample quantile} $\U_{{\rm EOO},95}$ of unfairness over 100 repetitions.}}
    \label{fig:simulation}
\end{figure}

Although the cost of group-blindness and the tradeoff in Remark~\ref{rem:tradeoff_blind} are minimax results, similar phenomena also appear in our simulations. Figure~\ref{fig:simulation}(b) shows that the group-aware and group-blind excess risks are comparable when the unfairness level is large. Theoretically, they are the excess risk of unconstrained classification in group-aware and group-blind scenarios.
As the unfairness level decreases, the group-blind excess risk grows sharply, indicating the tradeoff between group-blind excess risk and fairness, while the group-aware excess risk is relatively stable. 
For small unfairness levels, the group-blind excess risk becomes much larger, consistent with the theoretical cost of group-blindness arising from the error of predicting the sensitive attribute and the larger scale of $|\lambda^{*{\rm blind}}_\alpha|$ compared with $|\lambda^{*{\rm aware}}_\alpha|$, as verified in Figure~\ref{fig:lambda}.

\begin{figure}[t]
    \centering
    \includegraphics[width=0.6\linewidth]{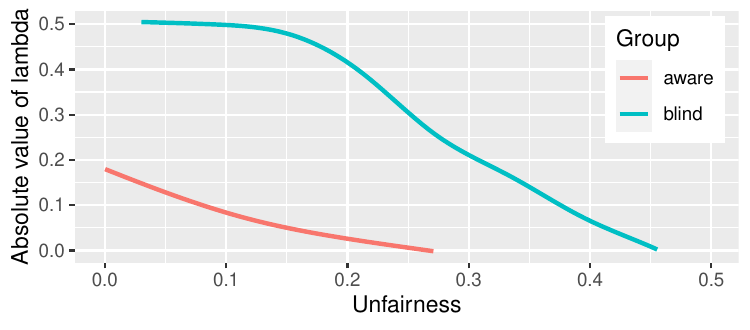}
    \caption{\small{The curve of $|\lambda^{*G}_\alpha|$ on $\alpha$ under (M1). The red line is for the group-aware scenario and the cyan line is for the group-blind scenario.}}
    \label{fig:lambda}
\end{figure}

\subsection{Real Data Analysis}

We apply the proposed algorithm to the Adult Census dataset \citep{asuncion2007uci}, which contains 48842 samples. The response is whether an individual's income exceeds \$50000. There are 14 non-sensitive covariates (e.g., age, marriage status, education), and gender is the sensitive attribute. We randomly split the dataset into 16000 training, 16000 calibration, and 16842 test samples. Initial estimators are trained on the training data using the same strategy as in the simulation study, the fair classifier is constructed utilizing the calibration samples, and the test data evaluates the prediction error and unfairness. We compare our method with FPIR and MBS and repeat the experiment 100 times. 

Given the large sample size, we set $\alpha\in\{0.04,0.06,0.08,0.10\}$ and $\delta=0.05$. Table~\ref{tab:real_blind} and \ref{tab:real_aware} report the prediction errors and unfairness measures for the group-blind and group-aware scenarios, respectively. The proposed method approximately controls the unfairness measures below $\alpha$ with probability $1-\delta$, whereas FPIR and MBS only control unfairness on average, which may lead to unfair decisions in individual realizations.

\begin{table}[t]
\small
    \centering
    \begin{tabular}{@{}lccccc@{}}\hline
     & & \multicolumn{4}{c}{$\alpha$}\\\cline{3-6}
    Methods & & 0.04 & 0.06 & 0.08 & 0.10\\\hline
    
    Ours & $\bar\U_{\rm EOO}$ & 0.028(0.020) & 0.037(0.024) & 0.052(0.033) & 0.069(0.032)\\
     & $\U_{{\rm EOO},95}$ & 0.059 & 0.077 & 0.109 & 0.120\\
     & Error & 0.151(0.002) & 0.150(0.002) & 0.150(0.002) & 0.150(0.002)\\[6pt]
    FPIR & $\bar\U_{\rm EOO}$ & 0.066(0.031) & 0.085(0.029) & 0.089(0.029) & 0.094(0.027)\\
     & $\U_{{\rm EOO},95}$ & 0.116 & 0.126 & 0.132 & 0.133\\
     & Error & 0.150(0.002) & 0.150(0.002) & 0.149(0.002) & 0.150(0.003)\\[6pt]
    MBS & $\bar\U_{\rm EOO}$ & 0.039(0.028) & 0.053(0.027) & 0.067(0.037) & 0.079(0.035)\\
     & $\U_{{\rm EOO}, 95}$ & 0.083 & 0.091 & 0.130 & 0.132\\
     & Error & 0.150(0.002) & 0.157(0.002) & 0.150(0.002) & 0.150(0.003)\\\hline
    \end{tabular}
    \caption{The unfairness measures and prediction errors in the group-blind scenario on the Adult Census dataset. The notation is the same as Table~\ref{tab:simulation_blind}.}
    \label{tab:real_blind}
\end{table}

\begin{table}[!b]
\small
    \centering
    \begin{tabular}{@{}lccccc@{}}\hline
     & & \multicolumn{4}{c}{$\alpha$}\\\cline{3-6}
    Methods & & 0.04 & 0.06 & 0.08 & 0.10\\\hline
    Ours & $\bar\U_{\rm EOO}$ & 0.027(0.020) & 0.038(0.024) & 0.053(0.033) & 0.069(0.032)\\
     & $\U_{{\rm EOO},95}$ & 0.060 & 0.076 & 0.112 & 0.124\\
     & Error & 0.151(0.002) & 0.151(0.002) & 0.150(0.002) & 0.150(0.002)\\[6pt]
    FPIR & $\bar\U_{\rm EOO}$ & 0.058(0.033) & 0.076(0.034) & 0.084(0.035) & 0.093(0.029)\\
     & $\U_{{\rm EOO},95}$ & 0.115 & 0.124 & 0.138 & 0.135\\
     & Error & 0.150(0.002) & 0.150(0.002) & 0.150(0.002) & 0.150(0.002)\\\hline
    \end{tabular}
    \caption{The unfairness measures and prediction errors of our method and FPIR, respectively in the group-aware scenario on the Adult Census dataset. The notation is the same as Table~\ref{tab:simulation_blind}.}
    \label{tab:real_aware}
\end{table}

Figure~\ref{fig:real} summarizes the trade-off between average prediction error and either the average unfairness or its 95\% sample quantile. Across all cases, FPIR fails to achieve small unfairness levels, and MBS tends to attain large unfairness levels. Therefore, our method achieves a better trade-off than FPIR and MBS.

\begin{figure}[t]
    \centering
    \begin{subfigure}[b]{\textwidth}
        \includegraphics[width=\linewidth]{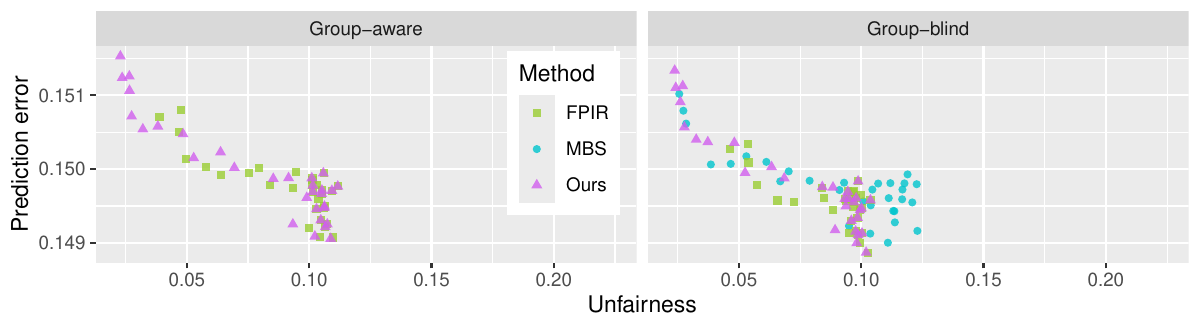}
        \caption{}
        \label{fig:real_average}
    \end{subfigure}
    \begin{subfigure}[b]{\textwidth}
        \includegraphics[width=\linewidth]{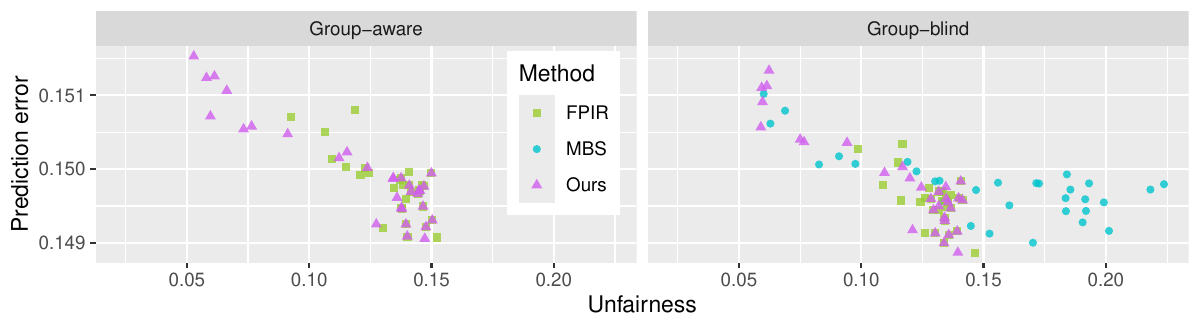}
        \caption{}
        \label{fig:real_quantile}
    \end{subfigure}
    \caption{The trade-off between prediction error and unfairness on the Adult Census dataset. The display is the same as Figure~\ref{fig:simulation}.}
    \label{fig:real}
\end{figure}

\section{Discussion}\label{sec:discussion}

In this work, we propose a comprehensive framework for fair classification with guaranteed fairness and excess risk across various fairness notions in both group-aware and group-blind scenarios. For binary sensitive attributes, we derive minimax lower bounds for excess risk, revealing the trade-off between excess risk and fairness, as well as the cost of group-blindness. In the following, we point out some interesting directions for future work.

For binary sensitive attributes, we study the excess risk when $\1(2\eta^G>1)$ is sufficiently fair or unfair. When $\U(\1(2\eta^G>1))$ approaches $\alpha$, additional assumptions, such as the detection condition \citep{tong2013plug}, are required to quantify the error of $\hat\lambda^G$. It would be interesting to derive a matching minimax lower bound, especially for the group-blind scenario.

For binary sensitive attributes, there is a gap $O_P(|\lambda^{*{\rm aware}}_\alpha|N^{-\frac{1}{2}})$ between the group-aware excess risk upper and lower bounds. We conjecture the upper bound is tight and new techniques may be required to prove the lower bound $O_P(|\lambda^{*{\rm aware}}_\alpha|N^{-\frac{1}{2}})$.

For multi-class sensitive attributes studied in Section A.2 of the supplement, we compare the prediction error of the proposed classifier $\hat f^G_{\hat\lambda_\alpha}$ with that of the Bayes optimal $\tilde\alpha$-fair classifier $f^{*G}_{\tilde\alpha}$ with a smaller unfairness level $\tilde\alpha\le\alpha$. Comparing directly with $f^{*G}_\alpha$ requires more complicated assumptions to control the error of $\hat\lambda^G_\alpha$. It would thus be interesting to define a more natural model space and characterize the minimax rate of $\risk(\hat f^G)-\risk(f^{*G}_\alpha)$.

%
%

%


\bibliography{reference}
\bibliographystyle{plainnat}

\newpage

\begin{center}
    {\Large Supplement to ``Finite-Sample and Distribution-Free Fair Classification: Optimal Trade-off Between Excess Risk and Fairness, and the Cost of Group-Blindness"}
\end{center}
\begin{abstract}
    In Section~\ref{sec:extension_supp}, we apply the results from Section~\ref{sec:unify} to other fairness notions, and propose a unified framework for fair classification with multi-class sensitive attributes. Section~\ref{sec:numerical_supp} contains the simulation results omitted in Section~\ref{sec:numerical}. Sections~\ref{sec:supp_exa_phi_eoo}, \ref{sec:supp_rem_poly_growth} and \ref{sec:supp_eq_lambda_upper_bound_eoo_aware} provide the derivation of Example~\ref{exa:phi_eoo}, Remark~\ref{rem:poly_growth} and Equation~\ref{eq:lambda_upper_bound_eoo_aware}, respectively. Section~\ref{sec:supp_modification} includes the modifications to $\hat\eta^G,\hat\phi^G$ for fulfilling Assumption~\ref{ass:init}. The rest of the supplement provides the omitted proofs of the theoretical results.
\end{abstract}

\appendix

\numberwithin{Lemma}{section}
\numberwithin{equation}{section}
\numberwithin{Definition}{section}
\numberwithin{Assumption}{section}
\numberwithin{Corollary}{section}
\numberwithin{algorithm}{section}
\numberwithin{Remark}{section}
\numberwithin{Theorem}{section}
\numberwithin{Proposition}{section}

\section{Extensions}\label{sec:extension_supp}

In this section, we build upon the previously discussed results by extending them in two directions. In Section~\ref{sec:binary_other}, we apply the framework proposed in Section~\ref{sec:unify} to more commonly used fairness notions. Then in Section~\ref{sec:unify_multi}, we extend the framework to multi-class sensitive attributes.

\subsection{Applications to More Fairness Notions}\label{sec:binary_other}

In this section, we apply the framework proposed in Section~\ref{sec:unify} to other widely used fairness notions defined in the following.

\begin{Definition}[Unfairness]\label{def:unfairness_binary}
    When $K=2$, for any randomized classifier $f$, the unfairness of $f$ in terms of
    \begin{enumerate}[1)]
        \item demographic parity (DP) is
        \[\U_{\rm DP}(f)=|\Prob(Y_f(X,A)=1|A=1)-\Prob(Y_f(X,A)=1|A=2)|,\]
        \item overall accuracy equality (OAE) is
        \begin{align*}
            \U_{\rm OAE}(f)=&|\Prob(Y_f(X,A)=1|A=1,Y=1)+\Prob(Y_f(X,A)=0|A=1,Y=0)\\
            &-\Prob(Y_f(X,A)=1|A=2,Y=1)-\Prob(Y_f(X,A)=0|A=2,Y=0)|,
        \end{align*}
        \item predictive equality (PE) is
        \[\U_{\rm PE}(f)=|\Prob(Y_f(X,A)=1|A=1,Y=0)-\Prob(Y_f(X,A)=1|A=2,Y=0)|.\]
    \end{enumerate}
\end{Definition}

Another commonly used fairness notion is equalized odds as defined in Definition~\ref{def:unfairness}. However, since the unfairness measure corresponding to equalized odds can not be reduced in this way to the absolute value of a linear combination of conditional expectations, we treat it as the multi-class sensitive attribute case and address it in Section~\ref{sec:unify_multi}.

Recall that $\rho_a(X)=\Prob(A=a|X)$, $p_a=\Prob(A=a)$. To apply Algorithm~\ref{alg:binary_unified} to unfairness measures in Definition~\ref{def:unfairness_binary}, we specify the corresponding $\phi^G_F$ for $G\in\{{\rm aware,blind}\}, F\in\{{\rm DP,OAE,PE}\}$ in Table~\ref{tab:phi_binary}. Here $G$ indicates the group-aware or group-blind scenarios and $F$ specifies the fairness notions: demographic parity (DP), overall accuracy equality (OAE), or predictive equality (PE).

\begin{table}
    \centering
    \begin{tabular}{@{}lcc@{}}\hline
        $F$ & $G={\rm aware}$ & $G={\rm blind}$ \\\hline
        \\[-1em]
        DP & $\frac{3-2a}{p_a}$ & $\frac{\rho_1(x)-p_1}{p_1p_2}$ \\
        \\[-1em]
        OAE & $\frac{3-2a}{p_{1,a}p_{0,a}}(p_a\eta^{G}(x,a)-p_{1,a})$ & $\sum_{a\in[2],y\in\{0,1\}}\frac{\rho_{a|y}(x)(\eta^G(x,a)+y-1)}{(3-2a)p_{y,a}}$\\
        \\[-1em]
        PE & $\frac{3-2a}{p_{0,a}}(1-\eta^{G}(x,a))$ & $\frac{(1-p_Y)\rho_{1|0}(x)-p_{0,1}}{p_{0,1}p_{0,2}}(1-\eta^G(x,a))$ \\[-1em]
        \\\hline
    \end{tabular}
    
    \caption{The form of $\phi^G_F(x,a)$ for $G\in\{{\rm aware,blind}\}, F\in\{{\rm DP,OAE,PE}\}$.}
    \label{tab:phi_binary}
\end{table}

We have $\U_F(f^G)=|\E\phi_F^G(X,A)f^G(X,A)|$ for $F\in\{{\rm DP,OAE,PE}\}$. Similar to Section~\ref{sec:binary_eoo_aware} and \ref{sec:binary_eoo_blind}, we impose the H\"older smoothness, observability assumptions, and strong density assumption~\ref{ass:density_aware}.

\begin{Assumption}[H\"older Smoothness]\label{ass:holder_other}
    We assume $\eta^G(\cdot,a)\in\HH(\beta_Y,L_Y)$, $\rho_1,\rho_{1|y}\in\HH(\beta_A,L_A)$ for all $G\in\{{\rm aware,blind}\}, a\in[2], y\in\{0,1\}$.
\end{Assumption}

\begin{Assumption}[Observability]\label{ass:observe_other}
    We assume the existence of constants $c_5>0$ such that $p_{y,a}>c_5$ for all $y\in\{0,1\},a\in[2]$.
\end{Assumption}

Following the same strategy as Section~\ref{sec:binary_eoo} to estimate $\eta^G$ and $\phi^G_F$ on $\tilde\D$, we denote $\epsilon_\eta\asymp\big(\frac{d\log \tilde n+\log\frac{1}{\delta_{\rm init}}}{\tilde n}\big)^{\frac{\beta_Y}{2\beta_Y+d}}$ and choose $\epsilon_{\phi,F}^G$ and $\epsilon_{\alpha,F}$ in Table~\ref{tab:epsilon_phi_binary} and Table~\ref{tab:epsilon_alpha_binary}, respectively. Here the value of $\epsilon_{\alpha,F}$ only depends on the fairness notions and remains the same across both group-aware and group-blind scenarios.

\begin{table}
    \centering
    \begin{tabular}{@{}lcc@{}}\hline
        $F$ & $G={\rm aware}$ & $G={\rm blind}$ \\\hline
        \\[-1em]
        DP & $\sqrt{\frac{\log\frac{1}{\delta_{\rm init}}}{\tilde n}}$ & $\big(\frac{d\log \tilde n+\log\frac{1}{\delta_{\rm init}}}{\tilde n}\big)^{\frac{\beta_A}{2\beta_A+d}}$\\
        \\[-1em]
        OAE, PE & $\big(\frac{d\log \tilde n+\log\frac{1}{\delta_{\rm init}}}{\tilde n}\big)^{\frac{\beta_Y}{2\beta_Y+d}}$ & $\big(\frac{d\log \tilde n+\log\frac{1}{\delta_{\rm init}}}{\tilde n}\big)^{\frac{\beta_Y}{2\beta_Y+d}}+\big(\frac{d\log \tilde n+\log\frac{1}{\delta_{\rm init}}}{\tilde n}\big)^{\frac{\beta_A}{2\beta_A+d}}$ \\[-1em]
        \\\hline
    \end{tabular}
    \caption{The order of $\epsilon_{\phi,F}^G$ for $G\in\{{\rm aware,blind}\}, F\in\{{\rm DP,OAE,PE}\}$.}
    \label{tab:epsilon_phi_binary}
\end{table}

\begin{table}
    \centering
    \begin{tabular}{@{}lcc@{}}\hline
        $F$ & $\epsilon_{\alpha,F}$\\\hline
        \\[-1em]
        DP & $\sum_{a\in[2]}\big(72\sqrt{\frac{2\log 4e^2}{n_a}}+\sqrt{\frac{1}{2n_a}\log \frac{4}{\delta_{\rm post}}}\big)$ \\
        \\[-1em]
        OAE  & $\sum_{a\in[2],y\in\{0,1\}}\big(72\sqrt{\frac{2\log 4e^2}{n_{y,a}}}+\sqrt{\frac{1}{2n_{y,a}}\log \frac{8}{\delta_{\rm post}}}\big)$\\
        \\[-1em]
        PE & $\sum_{a\in[2]}\big(72\sqrt{\frac{2\log 4e^2}{n_{0,a}}}+\sqrt{\frac{1}{2n_{0,a}}\log \frac{4}{\delta_{\rm post}}}\big)$\\[-1em]
        \\\hline
    \end{tabular}
    \caption{The value of $\epsilon_{\alpha,F}$ for $F\in\{{\rm DP,OAE,PE}\}$.}
    \label{tab:epsilon_alpha_binary}
\end{table}

Then we can guarantee the performance of the classifier $\hat f^G_{\alpha,F}$ produced by Algorithm~\ref{alg:binary_unified}.

\begin{Corollary}[Excess Risk Upper Bound]
    For $G\in\{{\rm aware,blind}\}, F\in\{{\rm DP,OAE,PE}\}$, suppose Assumptions~\ref{ass:init}, \ref{ass:margin}, \ref{ass:ratio_poly}, \ref{ass:ratio_balance}, \ref{ass:density_aware}, \ref{ass:holder_other} and \ref{ass:observe_other} hold, then with probability at least $1-\delta$, for any $\alpha\ge 2\epsilon_{\alpha,F}+\tilde\epsilon_{\phi,F}^G$ and such that the unfairness difference $D_0$ satisfies
        \[D_0\le -2\epsilon_{\alpha,F}-\tilde\epsilon_\eta^G\quad{\rm or}\quad D_0> \tilde\epsilon_\eta^G\vee c_3\big(2\epsilon_{\alpha,F}+\frac{c_1}{c_5}(2\epsilon_{\alpha,F}+(1+2c_4)|\lambda^{*G}_{\alpha,F}|\epsilon_{\phi,F}^G)^\gamma\big),\]
        we have
        \begin{equation}\label{eq:excess_risk_upper_binary_other}
        \begin{aligned}
            \risk(\hat f^G_{\alpha,F})-\risk(f^{*G}_{\alpha,F})\lesssim|\lambda^{*G}_{\alpha,F}|\epsilon_{\alpha,F}+\epsilon_\eta^{1+\gamma}+\big(|\lambda^{*G}_{\alpha,F}|\epsilon_{\phi,F}^G\big)^{1+\gamma}.
        \end{aligned}
        \end{equation}

\end{Corollary}

\subsection{A Unified Framework for Multi-Class Sensitive Attribute}\label{sec:unify_multi}

In this section, we provide a unified framework for the Bayes optimal $\alpha$-fair classifier, post-processing algorithm and excess risk analysis. 

\subsubsection{Bayes Optimal $\alpha$-Fair Classifier}\label{sec:bayes_multi}

First, we define the unfairness measures for multi-class sensitive attributes.

\begin{Definition}[Unfairness]\label{def:unfairness}
    For any randomized classifier $f$, the unfairness of $f$ in terms of
    \begin{enumerate}[1)]
        \item demographic parity is
        \[\U_{\rm DP}(f)=\max_{a\in[K]}|\Prob(Y_f(X,A)=1|A=a)-\Prob(Y_f(X,A)=1)|,\]
        \item equalized odds is
        \begin{align*}
            \U_{\rm EO}(f)=\max_{a\in[K]}\big\{&|\Prob(Y_f(X,A)=1|A=a,Y=1)-\Prob(Y_f(X,A)=1|Y=1)|\\
            &\vee|\Prob(Y_f(X,A)=0|A=a,Y=0)-\Prob(Y_f(X,A)=0|Y=0)|\big\},
        \end{align*}
        \item equality of opportunity is
        \[\U_{\rm EOO}(f)=\max_{a\in[K]}|\Prob(Y_f(X,A)=1|A=a,Y=1)-\Prob(Y_f(X,A)=1|Y=1)|,\]
        \item overall accuracy equality is
        \begin{align*}
            \U_{\rm OAE}(f)=\max_{a\in[K]}\big|&\Prob(Y_f(X,A)=1|A=a,Y=1)+\Prob(Y_f(X,A)=0|A=a,Y=0)\\
            &-\Prob(Y_f(X,A)=1|Y=1)-\Prob(Y_f(X,A)=0|Y=0)\big|,
        \end{align*}
        \item predictive equality is
        \[\U_{\rm PE}(f)=\max_{a\in[K]}|\Prob(Y_f(X,A)=1|A=a,Y=0)-\Prob(Y_f(X,A)=1|Y=0)|.\]
    \end{enumerate}
\end{Definition}

It is evident that these commonly used unfairness measures can be rewritten as
\[\U(f^G)=\|\E\Phi^G(X,A)f^G(X,A)\|,\]
for some bounded vector-valued function $\Phi^G=(\phi^G_k)_{k\in[\tilde K]}:\R^d\times[K]\rightarrow\R^{\tilde K}$, $\tilde K\in\mathbb{N}_+$ and norm $\|\cdot\|$ on $\R^{\tilde K}$, $G\in\{{\rm aware, blind}\}$. 
When $G={\rm blind}$, then $\Phi^G$ is only a function of $X$. Then we can characterize the solution of Problem~\eqref{eq:bayes_obj} as follows.

\begin{Proposition}[Bayes Optimal $\alpha$-fair Classifier]\label{prop:bayes}
    For $G\in\{{\rm aware,blind}\}$, the Bayes optimal classifier $f^{*G}_\alpha\in[0,1]^{\R^d\times [K]}$ of Problem \eqref{eq:bayes_obj} has the following form $P_{X,A}$-almost surely, with $P_{X,A}$ to be the joint distribution of $(X,A)$,
    \begin{align*}
        f^{*G}_\alpha (X,A)=&\1\big(g^{*G}_\alpha (X,A)>0\big)+b^G(X,A)\1\big(g^{*G}_\alpha (X,A)=0\big),
    \end{align*}
    for 
    \[g^{*G}_\alpha (X,A)=2\eta^G(X,A)-1-\lambda^{*G\top}_\alpha\Phi^G(X,A),\]
    \[\lambda^{*G}_\alpha \in\argmin_{\lambda\in\R^{\tilde K}}\E\big(2\eta^G(X,A)-1-\lambda^\top\Phi^G(X,A)\big)_++\alpha\|\lambda\|_*,\]
    and any $b^G\in[0,1]^{\R^d\times[K]}$ mapping from $\R^d\times [K]$ to $[0,1]$ such that $f^{*G}_\alpha$ satisfies the fairness constraint and 
    \begin{equation}\label{eq:lambda_optimality}
        \lambda^{*G\top}_\alpha\E\Phi^G(X,A)f^{*G}_\alpha(X,A)=\|\lambda^{*G}_\alpha\|_*\|\E\Phi^G(X,A)f^{*G}_\alpha(X,A)\|=\|\lambda^{*G}_\alpha\|_*\alpha.
    \end{equation}
    Here $\|\cdot\|_*$ is the dual norm of $\|\cdot\|$.
\end{Proposition}
\begin{Remark}\label{rem:lambda_upper_multi}
   Similar to binary sensitive attribute setting, $\|\lambda^{*G}_\alpha\|_*$ is always upper bounded. To see this, by Equation~\eqref{eq:lambda_optimality}, we know
    \[\|\lambda^{*G}_\alpha\|_*\alpha=\E\lambda^{*G\top}_\alpha\Phi^G(X,A)\1\big(2\eta^G(X,A)-1>\lambda^{*G\top}_\alpha\Phi^G(X,A)\big)\le\E\big(2\eta^G(X,A)-1\big)f^{*G}_\alpha(X,A)\le 1,\]
    therefore $\|\lambda^{*G}_\alpha\|_*\le\alpha^{-1}$. 
\end{Remark}

Similar to the case with binary sensitive attributes in Section~\ref{sec:unify}, the Bayes optimal $\alpha$-fair classifier in Proposition~\ref{prop:bayes} is a translation of the unconstrained Bayes optimal classifier $\1(2\eta^G>1)$ by $\lambda^{*G\top}_\alpha\Phi^G$. This motivates us to construct a fair classifier by post-processing.

\subsubsection{Post-processing Algorithm}

In this section, we propose a general algorithm for various fairness notions with fairness and excess risk guarantees. For the unfairness measures in Definition~\ref{def:unfairness}, the vector norms $\|\cdot\|$ and $\|\cdot\|_*$ in Proposition~\ref{prop:bayes} equal to the $\ell_\infty$ norm $\|\cdot\|_\infty$ and $\ell_1$ norm $\|\cdot\|_1$, respectively.

It is clear that for the same $\tilde K$ defined above, the unfairness measures in Definition~\ref{def:unfairness} can also be rewritten as
\[\U(f^G)=\max_{k\in[\tilde K]}\bigg|\sum_{j\in[m]}\kappa_j\E_{kj}f^G(X,A)\bigg|,\]
with $\{\kappa_j\in\R:j\in[m]\}$ are known coefficients and $\{\E_{kj}:k\in[\tilde K],j\in[m]\}$ are a set of conditional expectations given the sensitive attributes, depending on the fairness notions.

Similar to the binary sensitive attribute case in Section~\ref{sec:unify_binary}, we assume the initial estimators $\hat\eta^G$ and $\hat\Phi^G$ are given and independent of the dataset $\D$. We select $\hat\lambda$ based on $\D$ to post-process $\hat\eta^G$ and $\hat\Phi^G$. Denote
\[\hat f^G_\lambda(x,a)=\1\big(2\hat\eta^G(x,a)-1>\lambda^\top\hat\Phi^G(x,a)\big),\]
the excess risk of $\hat f^G_\lambda$ can be decomposed as
\begin{align*}
    &\risk(\hat f^G_\lambda)-\risk(f^{*G}_\alpha)\\
    =&\underbrace{\E|g^{*G}_\alpha(X,A)||f^{*G}_\alpha(X,A)-\hat f^G_\lambda(X,A)|}_{T_1}+\underbrace{\E\lambda^{*G\top}_\alpha\Phi^G(X,A)\big(f^{*G}_\alpha(X,A)-\hat f^G_\lambda(X,A)\big)}_{T_2}.
\end{align*}
For the binary sensitive attribute case in Section~\ref{sec:unify_binary}, $\lambda^{*G}_\alpha$ is a real number with two well-separated directions, i.e. positive or negative. Then as long as $\alpha$ is large enough, we are able to construct $\hat\lambda$ as in Algorithm~\ref{alg:binary_unified} such that it roughly maximizes $\sgn(\lambda^{*G}_\alpha) \E \phi^G(X,A) \hat f^G_\lambda(X,A)$, or equivalently minimizes $T_2$, and ensures $\U(\hat f^G_\lambda)=\max_{s\in\{1,-1\}}\E s\phi^G(X,A)\hat f^G_\lambda(X,A)\le\alpha$ simultaneously. However, for multi-class sensitive attributes, $\lambda^{*G}_\alpha$ has dimension $\tilde K>1$ and there are continuum directions $\{\mu\in\R^{\tilde K}:\|\mu\|_1=1\}$. Then it is not clear how to directly control $\frac{\lambda^{*G\top}_\alpha}{\|\lambda^{*G}_\alpha\|_1}\E\Phi^G(X,A)\hat f^G_\lambda(X,A)$ and $\U(\hat f^G_\lambda)=\sup_{\|\mu\|_1=1}\mu^\top\E\Phi^G(X,A)\hat f^G_\lambda(X,A)$ simultaneously. Consequently, the strategy in Algorithms~\ref{alg:binary_unified} can not be applied in this case. In the following, we propose an algorithm to select $\lambda$ using empirical risk minimization.

Denote the empirical unfairness as
\[\hat\U(f^G)=\max_{k\in[\tilde K]}\bigg|\sum_{j\in[m]}\kappa_j\hat\E_{kj}f^G(X,A)\bigg|\]
with $\{\hat\E_{kj}:k\in[\tilde K],j\in[m]\}$ to be the set of conditional sample averages corresponding to $\{\E_{kj}:k\in[\tilde K],j\in[m]\}$ based on data $\D$. We denote $n_{(kj)}$ to be the number of samples in $\D$ used to calculate the conditional sample average $\hat \E_{kj}$, and set
\[\epsilon_\alpha=\max_{k\in[\tilde K]}\sum_{j\in[m]}|\kappa_j|\bigg\{72\sqrt{\frac{(\tilde K+1)\log 4e^2}{n_{(kj)}}}+\sqrt{\frac{1}{2n_{(kj)}}\log\frac{2\tilde Km}{\delta_{\rm post}}}\bigg\}.\]
Then the following lemma ensures the possibility of distribution-free and finite-sample fairness control.
\begin{Lemma}\label{lem:unfairness_dev_multi}
    With probability at least $1-\delta_{\rm post}$ on $\D$,
\[\sup_{\lambda\in\R^{\tilde K}}|\hat \U(\hat f^G_\lambda)-\U(\hat f^G_\lambda)|\le\epsilon_\alpha.\]
\end{Lemma}
Following Lemma~\ref{lem:unfairness_dev_multi}, we estimate $\lambda^{*G}_\alpha$ by
\begin{equation}\label{eq:obj_multi}
    \begin{aligned}
    \hat\lambda^G_\alpha\in\argmin_{\lambda\in\R^{\tilde K}}\sum_{i=1}^n\1\big(Y_i\ne\hat f^G_\lambda(X_i,A_i)\big)\quad{\rm s.t.}\quad \hat\U(\hat f^G_\lambda)\le\alpha-\epsilon_\alpha,
\end{aligned}
\end{equation}
and set the classifier as $\hat f^G_{\hat\lambda_\alpha}$. We summarize the above procedures in Algorithm~\ref{alg:multi_unify}.

\begin{algorithm}
\caption{Post-processing with Multi-Class Sensitive Attribute}\label{alg:multi_unify}
\begin{algorithmic}
\State{\bf Input:} Data $\D$, the initial estimators $\hat\eta^G,\hat\Phi^G$, the unfairness level $\alpha$, the tolerance $\delta$, and the scenario $G\in\{{\rm aware,blind}\}$.
\State{\bf Output:} $\hat f^G_{\hat\lambda}$.
\State{\bf Step 1:} Solve
\[\hat\lambda^G_\alpha\in\argmin_{\lambda\in\R^{\tilde K}}\sum_{i=1}^n\1\big(Y_i\ne\hat f^G_\lambda(X_i,A_i)\big)\quad{\rm s.t.}\quad\hat\U(\hat f^G_\lambda)\le\alpha-\epsilon_\alpha.\]
\State{\bf Step 2:} Output $\hat f^G_{\hat\lambda_\alpha}(x,a)=\1\big(2\hat\eta^G(x,a)-1>\hat\lambda^{G\top}_\alpha\hat\Phi^G(x,a)\big)$.
\end{algorithmic}
\end{algorithm}

\subsubsection{Performance Guarantee}

To study the performance of the proposed algorithm, we denote $\epsilon_\eta$ and $\epsilon_\phi$ to be the estimation errors of $\hat\eta^G$ and $\hat\Phi^G$ respectively, 
\[\|\hat\eta^G-\eta^G\|_\infty\le\epsilon_\eta,\quad \max_{k\in[\tilde K]}\|\hat\phi_k^G-\phi_k^G\|_\infty\le\epsilon_\phi.\]
For $\tilde\epsilon_\alpha$ to be specified later, we denote $\tilde\alpha=\alpha-\tilde\epsilon_\alpha$ and denote the Bayes optimal $\tilde\alpha$-fair classifier as 
\[\lambda^{*G}_{\tilde\alpha}\in\argmin_{\lambda\in\R^{\tilde K}}\E\big(2\eta^G(X,A)-1-\lambda^\top\Phi^G(X,A)\big)_++\tilde\alpha\|\lambda\|_1,\]
\[f^{*G}_{\tilde\alpha}(x,a)=\1\big(g^{*G}_{\tilde\alpha}(x,a)>0),\quad g^{*G}_{\tilde\alpha}(x,a)=2\eta^G(x,a)-1-\lambda^{*G\top}_{\tilde\alpha}\Phi^G(x,a).\]
Then we denote the margins $\tilde\epsilon_\eta^G$ and $\tilde\epsilon_{g,\tilde\alpha}^G$ of $2\eta^G-1$ and $g^{*G}_{\tilde\alpha}$ as
\[\tilde\epsilon_\eta^G=\max_{k\in[\tilde K]}\E|\phi_k^G(X,A)|\1(|2\eta^G(X,A)-1|\le 2\epsilon_\eta),\]
\[\tilde\epsilon_{g,\tilde\alpha}^G=\max_{k\in[\tilde K]}\E|\phi^G_k(X,A)|\1\big(|g^{*G}_{\tilde\alpha}(X,A)|\le2\epsilon_\eta+\|\lambda^{*G}_{\tilde\alpha}\|_1\epsilon_\phi\big).\]
$\tilde\epsilon_\eta^G$ and $\tilde\epsilon_{g,\tilde\alpha}^G$ measure the mass of $2\eta^G(X,A)-1$ and $g^{*G}_{\tilde\alpha}(X,A)$ around 0, respectively. Since $\Phi^G$ is bounded and $\epsilon_\eta$, $\epsilon_\phi$ are typically small, we know $\tilde\epsilon_\eta^G$ and $\tilde\epsilon_{g,\tilde\alpha}^G$ are small as long as $2\eta^G-1$ and $g^{*G}_{\tilde\alpha}$ are not too concentrated around 0.

Similar to the binary sensitive attribute case, we denote $D_0=\U(\1(2\eta^G>1))-\alpha$ to be the difference between the unfairness of the unconstrained Bayes optimal classifier $\1(2\eta^G>1)$ and the specified unfairness level $\alpha$. If $D_0\le 0$, we know $f^{*G}_\alpha=\1(2\eta^G>1)$ and $\lambda^*_\alpha=0$. If $D_0>0$, $\1(2\eta^G>1)$ is not $\alpha$-fair and need to be adjusted by $\lambda^{*G\top}_\alpha\Phi^G$.

Now we specify the choice of $\tilde\epsilon_\alpha$ as follows.
\begin{enumerate}[1)]
    \item If $D_0\le-\tilde\epsilon_\eta^G-2\epsilon_\alpha$, we set $\tilde\epsilon_\alpha=0$.
    \item If $D_0>-\tilde\epsilon_\eta^G-2\epsilon_\alpha$, we choose $\tilde\epsilon_\alpha$ such that
    \begin{equation}\label{eq:tilde_epsilon_alpha}
        \tilde\epsilon_\alpha\ge2\epsilon_\alpha+\tilde\epsilon_{g,\tilde\alpha}^G.
    \end{equation}
\end{enumerate}
Therefore, when $D_0\le-\tilde\epsilon_\eta^G-2\epsilon_\alpha$, we have $\tilde\alpha=\alpha$ and $\lambda^*_{\tilde\alpha}=\lambda^*_\alpha=0$.
\begin{Remark}
    Now we give an example where Equation~\eqref{eq:tilde_epsilon_alpha} is satisfied. Suppose the density of $2\eta^G(X,A)-1-\lambda^\top\Phi^G(X,A)$ is upper bounded for all $\lambda\in\R^{\tilde K}$. Since $\phi^G_k$, $k\in[\tilde K]$ are bounded, we have $\tilde\epsilon_{g,\tilde\alpha}^G\le c\epsilon_\eta+c\|\lambda^{*G}_{\tilde\alpha}\|_1\epsilon_\phi$. Using the naive upper bound in Remark~\ref{rem:lambda_upper_multi}, we get $\|\lambda^*_{\tilde\alpha}\|_1\le(\alpha-\tilde\epsilon_\alpha)^{-1}$. When $\alpha\ge 4\epsilon_\alpha+2c\epsilon_\eta+2\sqrt{2c\epsilon_\phi}$, condition~\eqref{eq:tilde_epsilon_alpha} is satisfied for $\tilde\epsilon_\alpha=2\epsilon_\alpha+c\epsilon_\eta+\sqrt{2c\epsilon_\phi}$.
\end{Remark}

Similar to Section~\ref{sec:unify_binary}, we make the margin assumption \citep{tsybakov2004optimal}.

\begin{Assumption}[Margin Assumption]\label{ass:margin_multi}
    There exist $\tilde\gamma\ge0$ and constant $c_1>0$ such that for any $\epsilon\ge0$, we have
    \[\Prob\big(|g^{*G}_{\tilde\alpha}(X,A)|\le \epsilon\big)\le c_1\epsilon^{\tilde\gamma}.\]
\end{Assumption}

Note that the margin assumption~\ref{ass:margin_multi} is on $g^{*G}_{\tilde\alpha}$, not $g^{*G}_\alpha$. To clarify the rationale behind this choice, we first present the following theorem, which controls the fairness and excess risk.

\begin{Theorem}\label{thm:upper_multi}
    \begin{enumerate}[1)]
        \item With probability at least $1-\delta_{\rm post}$, for any $\alpha\ge\tilde\epsilon_\alpha$, we have Algorithm~\ref{alg:multi_unify} to be feasible and $\U(\hat f^G_{\hat\lambda_\alpha})\le\alpha$.
        \item Under Assumption~\ref{ass:margin_multi}, we have with probability at least $1-2\delta_{\rm post}$, for any $\alpha\ge\tilde\epsilon_\alpha$, 
        \begin{equation}\label{eq:excess_risk_upper_multi_unify}
            \risk(\hat f^G_{\hat\lambda_\alpha})-\risk(f^{*G}_{\tilde\alpha})\lesssim\|\lambda^{*G}_{\tilde\alpha}\|_1\tilde\epsilon_\alpha+\big(\epsilon_\eta+\|\lambda^{*G}_{\tilde\alpha}\|_1\epsilon_\phi\big)^{1+\tilde\gamma}+\bigg(\frac{\tilde K\log n+\log\frac{1}{\delta_{\rm post}}}{n}\bigg)^{\frac{1+\tilde\gamma}{2+\tilde\gamma}}.
        \end{equation}
    \end{enumerate}
\end{Theorem}

In Theorem~\ref{thm:upper_multi}, we compare the prediction error of $\hat f^G_{\hat\lambda_\alpha}$ with that of $f^{*G}_{\tilde\alpha}$, rather than $f^{*G}_\alpha$. So $\risk(\hat f^G_{\hat\lambda_\alpha})-\risk(f^{*G}_{\tilde\alpha})$ may not always be non-negative. To explain this, note that $\hat f^G_{\lambda^*_\alpha}$ may not satisfy the constraint $\hat \U(\hat f^G_{\lambda^*_\alpha})\le\alpha-\epsilon_\alpha$, therefore we have to select another $\lambda$ that meets the empirical fairness constraint $\hat \U(\hat f^G_\lambda)\le\alpha-\epsilon_\alpha$ and exhibits satisfactory prediction performance. $\lambda^*_{\tilde\alpha}$ turns out to satisfy both conditions. However, when comparing with $f^{*G}_\alpha$, it is imperative to control the distance between $\lambda^*_{\tilde\alpha}$ and $\lambda^*_\alpha$, necessitating additional assumptions such as those similar to the detection assumption \citep{tong2013plug} in the context of Neyman-Pearson classification. Consequently, to maintain the simplicity of our results with the fewest assumptions, we articulate the excess risk comparing with $f^{*G}_{\tilde\alpha}$. Then the margin assumption~\ref{ass:margin_multi} is also on $g^{*G}_{\tilde\alpha}$ instead of $g^{*G}_\alpha$.

Under the $\beta_Y$-H\"older smoothness assumption on $\eta^G$, if we estimate $\eta^G$ on an independent dataset $\tilde\D$ with sample size $\tilde n$, we know $\epsilon_\eta\asymp(\frac{d\log \tilde n}{\tilde n})^{\frac{\beta_Y}{2\beta_Y+d}}$. 
When $n\gtrsim \tilde n$, under the standard condition $\beta_Y\tilde\gamma\le d$ \citep{audibert2007fast}, with $d\gtrsim\tilde K$, we have $(\frac{\tilde K\log n}{n})^{\frac{1+\tilde\gamma}{2+\tilde\gamma}}\lesssim\epsilon_\eta^{1+\tilde\gamma}$. 
Then the excess risk~\eqref{eq:excess_risk_upper_multi_unify} becomes $O_P\big(\|\lambda^*_{\tilde\alpha}\|_1\tilde\epsilon_\alpha+\big(\epsilon_\eta+\|\lambda^*_{\tilde\alpha}\epsilon_\phi\|\big)^{1+\tilde\gamma}\big)$, sharing the same form as the excess risk~\eqref{eq:excess_risk_upper_binary_unify} with binary sensitive attributes.
If $\1(2\eta^G>1)$ is already $\tilde\alpha$-fair, we know $\lambda^*_{\tilde\alpha}=\lambda^*_\alpha=0$. 
Then the excess risk~\eqref{eq:excess_risk_upper_multi_unify} becomes $O_P(\epsilon_\eta^{1+\tilde\gamma})$, which is minimax optimal up to logarithmic factors \citep{audibert2007fast}. 
When $\1(2\eta^G>1)$ is not $\tilde\alpha$-fair, we know $\lambda^*_{\tilde\alpha}\ne 0$. Then we incur an additional cost $O_P\big(\|\lambda^{*G}_{\tilde\alpha}\|_1\tilde\epsilon_\alpha+(\|\lambda^{*G}_{\tilde\alpha}\|_1\epsilon_\phi)^{1+\tilde\gamma}\big)$ due to the fairness constraint.

\section{Supplementary Simulation Results}\label{sec:numerical_supp}

In this section, we present the omitted simulation results for (M2) and (M3). The results and interpretations are similar to those of (M1), therefore we report the results without explanations. 

\begin{table}
    \centering
    \small
    \begin{tabular}{@{}lcccccc@{}}\hline
     & & \multicolumn{5}{c}{$\alpha$}\\\cline{3-7}
    Methods & & 0.08 & 0.11 & 0.14 & 0.17 & 0.20\\\hline
    
    Ours & $\bar\U_{\rm EOO}$ & 0.055(0.042) & 0.055(0.039) & 0.061(0.044) & 0.086(0.053) & 0.100(0.057)\\
     & $\U_{{\rm EOO},95}$ & 0.137 & 0.134 & 0.137 & 0.173 & 0.184\\
     & Error & 0.318(0.027) & 0.319(0.026) & 0.308(0.024) & 0.297(0.018) & 0.295(0.018)\\[6pt]
    FPIR & $\bar\U_{\rm EOO}$ & 0.084(0.051) & 0.120(0.050) & 0.135(0.052) & 0.148(0.049) & 0.160(0.045)\\
     & $\U_{{\rm EOO},95}$ & 0.179 & 0.195 & 0.206 & 0.223 & 0.225\\
     & Error & 0.296(0.017) & 0.289(0.013) & 0.285(0.010) & 0.284(0.010) & 0.283(0.008)\\[6pt]
    MBS & $\bar\U_{\rm EOO}$ & 0.082(0.058) & 0.099(0.056) & 0.116(0.060) & 0.126(0.056) & 0.140(0.061)\\
     & $\U_{{\rm EOO}, 95}$ & 0.178 & 0.183 & 0.222 & 0.229 & 0.239\\
     & Error & 0.297(0.017) & 0.293(0.016) & 0.290(0.014) & 0.286(0.010) & 0.286(0.010)\\[6pt]
     Bayes & Error & 0.283 & 0.276 & 0.272 & 0.272 & 0.272\\\hline
    \end{tabular}
    \caption{The unfairness measures and prediction errors of our method, FPIR, and MBS, respectively in the group-blind scenario under (M2). And the prediction errors of Bayes optimal fair classifiers. $\bar\U_{\rm EOO}$ is the average unfairness over 100 repetitions. $\U_{{\rm EOO},95}$ is the $95\%$ sample quantile of the unfairness measures produced by 100 repetitions. Error is the average prediction error.}
    \label{tab:simulation_blind_m2}
\end{table}

\begin{table}
    \centering
    \small
    \begin{tabular}{@{}lcccccc@{}}\hline
     & & \multicolumn{5}{c}{$\alpha$}\\\cline{3-7}
    Methods & & 0.08 & 0.11 & 0.14 & 0.17 & 0.20\\\hline
    Ours & $\bar\U_{\rm EOO}$ & 0.051(0.041) & 0.047(0.035) & 0.063(0.049) & 0.074(0.050) & 0.088(0.062)\\
     & $\U_{{\rm EOO},95}$ & 0.121 & 0.114 & 0.157 & 0.171 & 0.198\\
     & Error & 0.266(0.010) & 0.267(0.009) & 0.264(0.009) & 0.262(0.009) & 0.262(0.007)\\[6pt]
    FPIR & $\bar\U_{\rm EOO}$ & 0.126(0.077) & 0.130(0.079) & 0.135(0.083) & 0.151(0.081) & 0.172(0.084)\\
     & $\U_{{\rm EOO},95}$ & 0.262 & 0.262 & 0.266 & 0.265 & 0.288\\
     & Error & 0.264(0.015) & 0.262(0.010) & 0.261(0.010) & 0.259(0.010) & 0.260(0.007)\\[6pt]
     Bayes & Error & 0.245 & 0.245 & 0.244 & 0.243 & 0.243\\\hline
    \end{tabular}
    \caption{The unfairness measures and prediction errors of our method and FPIR, respectively in the group-aware scenario under (M2). The notation is the same as Table~\ref{tab:simulation_blind_m2}.}
    \label{tab:simulation_aware_m2}
\end{table}

\begin{figure}
    \centering
    \begin{subfigure}[b]{\textwidth}
        \includegraphics[width=\linewidth]{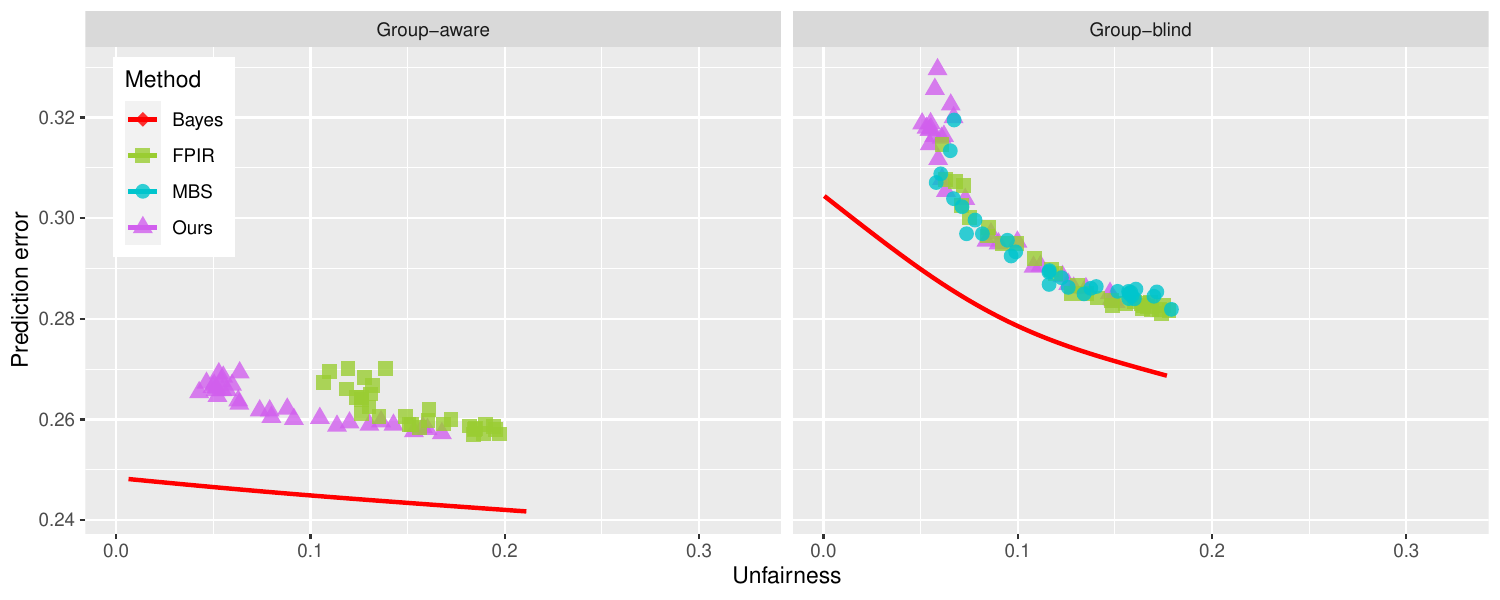}
        \caption{}
        \label{fig:simulation_average_m2}
    \end{subfigure}
    \begin{subfigure}[b]{\textwidth}
        \includegraphics[width=\linewidth]{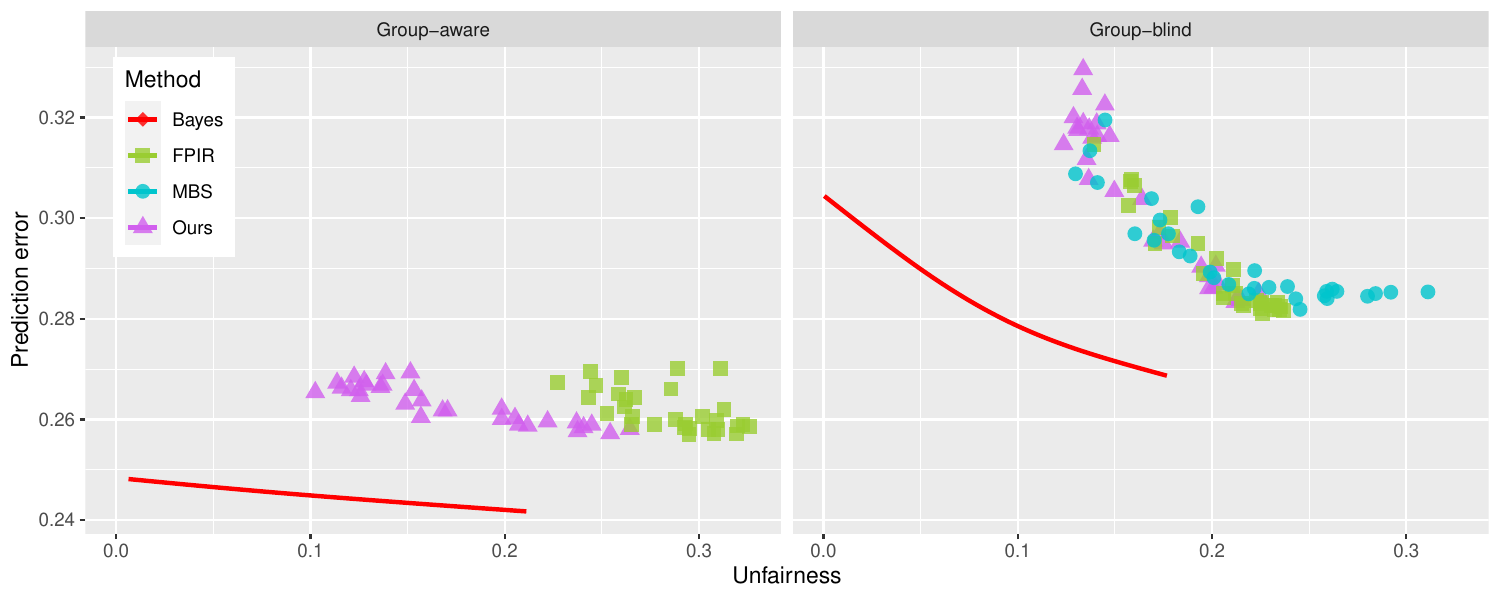}
        \caption{}
        \label{fig:simulation_quantile_m2}
    \end{subfigure}
    \caption{(a) The trade-off between prediction error and unfairness under (M2). The X-axis is the average unfairness measures $\bar\U_{\rm EOO}$ of the trained classifiers over 100 repetitions and the Y-axis is the average test prediction errors of these classifiers. The left and right panels correspond to the group-aware and group-blind scenarios, respectively. (b) As for (a) but the X-axis is the $95\%$ sample quantile $\U_{{\rm EOO},95}$ of the unfairness measures over 100 repetitions.}
    \label{fig:simulation_m2}
\end{figure}

\begin{figure}
    \centering
    \includegraphics[width=0.6\linewidth]{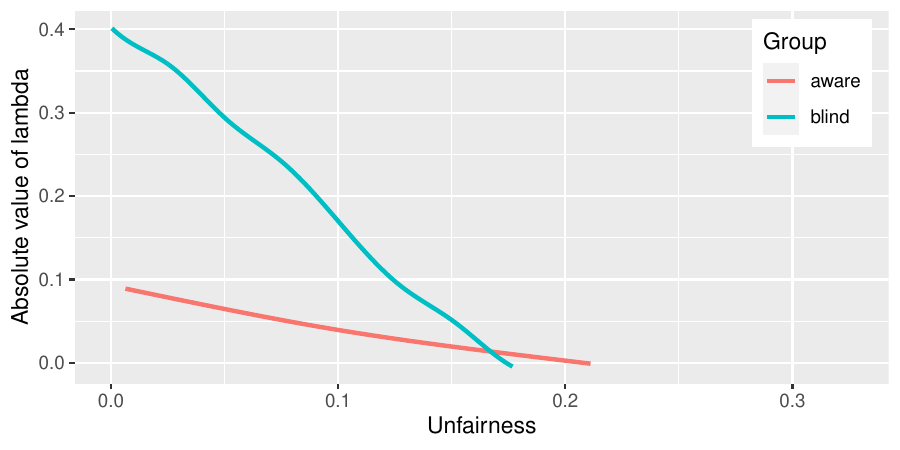}
    \caption{The curve of $|\lambda^{*G}_\alpha|$ on $\alpha$ under (M2). The red line is for the group-aware scenario and the cyan line is for the group-blind scenario.}
    \label{fig:lambda_m2}
\end{figure}

\begin{table}
    \centering
    \small
    \begin{tabular}{@{}lcccccc@{}}\hline
     & & \multicolumn{5}{c}{$\alpha$}\\\cline{3-7}
    Methods & & 0.08 & 0.11 & 0.14 & 0.17 & 0.20\\\hline
    
    Ours & $\bar\U_{\rm EOO}$ & 0.039(0.033) & 0.051(0.038) & 0.074(0.044) & 0.098(0.048) & 0.138(0.048)\\
     & $\U_{{\rm EOO},95}$ & 0.107 & 0.128 & 0.143 & 0.180 & 0.216\\
     & Error & 0.333(0.024) & 0.320(0.024) & 0.306(0.025) & 0.294(0.023) & 0.273(0.023)\\[6pt]
    FPIR & $\bar\U_{\rm EOO}$ & 0.101(0.061) & 0.124(0.054) & 0.157(0.066) & 0.178(0.056) & 0.209(0.069)\\
     & $\U_{{\rm EOO},95}$ & 0.211 & 0.210 & 0.271 & 0.271 & 0.313\\
     & Error & 0.293(0.030) & 0.279(0.027) & 0.266(0.029) & 0.258(0.022) & 0.245(0.026)\\[6pt]
    MBS & $\bar\U_{\rm EOO}$ & 0.088(0.043) & 0.109(0.048) & 0.140(0.053) & 0.172(0.047) & 0.211(0.048)\\
     & $\U_{{\rm EOO}, 95}$ & 0.148 & 0.181 & 0.214 & 0.247 & 0.282\\
     & Error & 0.299(0.023) & 0.285(0.023) & 0.273(0.025) & 0.260(0.021) & 0.243(0.018)\\[6pt]
     Bayes & Error & 0.265 & 0.251 & 0.238 & 0.226 & 0.214\\\hline
    \end{tabular}
    \caption{The unfairness measures and prediction errors of our method, FPIR, and MBS, respectively in the group-blind scenario under (M3). And the prediction errors of Bayes optimal fair classifiers. $\bar\U_{\rm EOO}$ is the average unfairness over 100 repetitions. $\U_{{\rm EOO},95}$ is the $95\%$ sample quantile of the unfairness measures produced by 100 repetitions. Error is the average prediction error.}
    \label{tab:simulation_blind_m3}
\end{table}

\begin{table}
    \centering
    \small
    \begin{tabular}{@{}lcccccc@{}}\hline
     & & \multicolumn{5}{c}{$\alpha$}\\\cline{3-7}
    Methods & & 0.08 & 0.11 & 0.14 & 0.17 & 0.20\\\hline
    Ours & $\bar\U_{\rm EOO}$ & 0.036(0.027) & 0.045(0.033) & 0.067(0.040) & 0.103(0.043) & 0.132(0.046)\\
     & $\U_{{\rm EOO},95}$ & 0.091 & 0.105 & 0.145 & 0.170 & 0.213\\
     & Error & 0.209(0.011) & 0.206(0.009) & 0.202(0.008) & 0.199(0.008) & 0.194(0.008)\\[6pt]
    FPIR & $\bar\U_{\rm EOO}$ & 0.103(0.084) & 0.112(0.075) & 0.136(0.091) & 0.147(0.078) & 0.190(0.102)\\
     & $\U_{{\rm EOO},95}$ & 0.269 & 0.249 & 0.305 & 0.295 & 0.360\\
     & Error & 0.206(0.019) & 0.201(0.016) & 0.197(0.012) & 0.197(0.010) & 0.193(0.010)\\[6pt]
     Bayes & Error & 0.156 & 0.152 & 0.150 & 0.147 & 0.146\\\hline
    \end{tabular}
    \caption{The unfairness measures and prediction errors of our method and FPIR, respectively in the group-aware scenario under (M3). The notation is the same as Table~\ref{tab:simulation_blind_m3}.}
    \label{tab:simulation_aware_m3}
\end{table}

\begin{figure}
    \centering
    \begin{subfigure}[b]{\textwidth}
        \includegraphics[width=\linewidth]{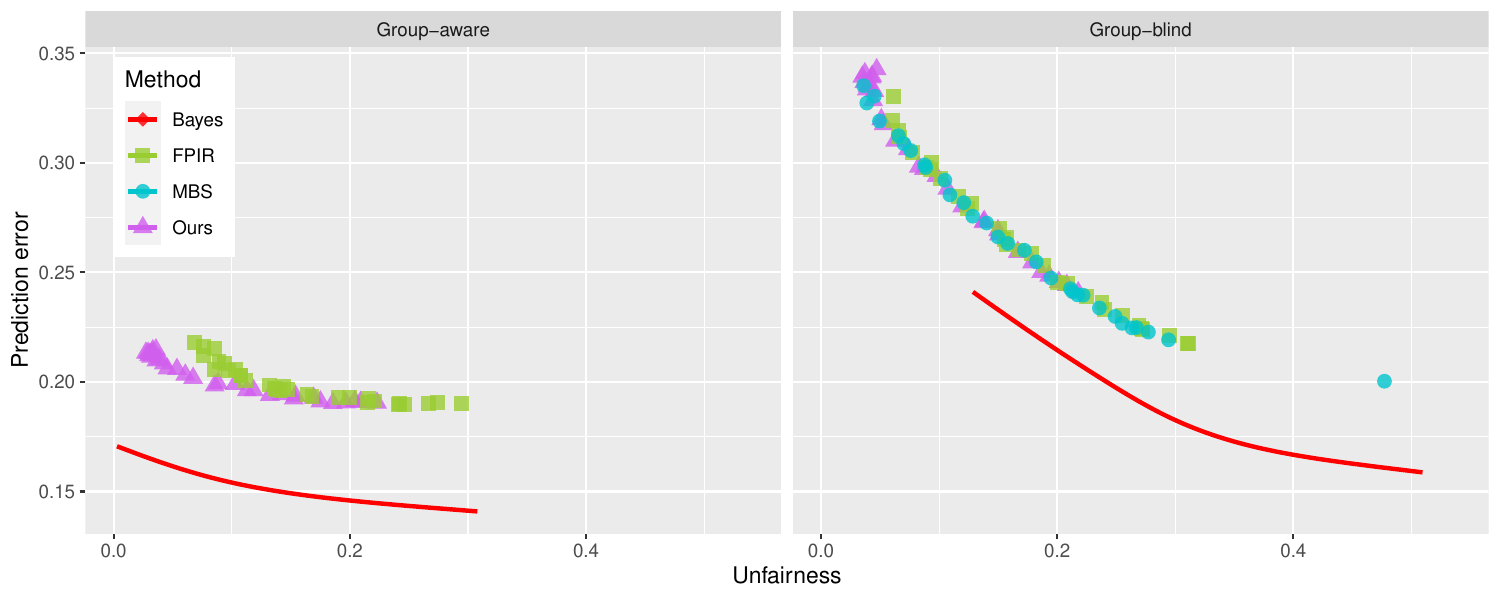}
        \caption{}
        \label{fig:simulation_average_m3}
    \end{subfigure}
    \begin{subfigure}[b]{\textwidth}
        \includegraphics[width=\linewidth]{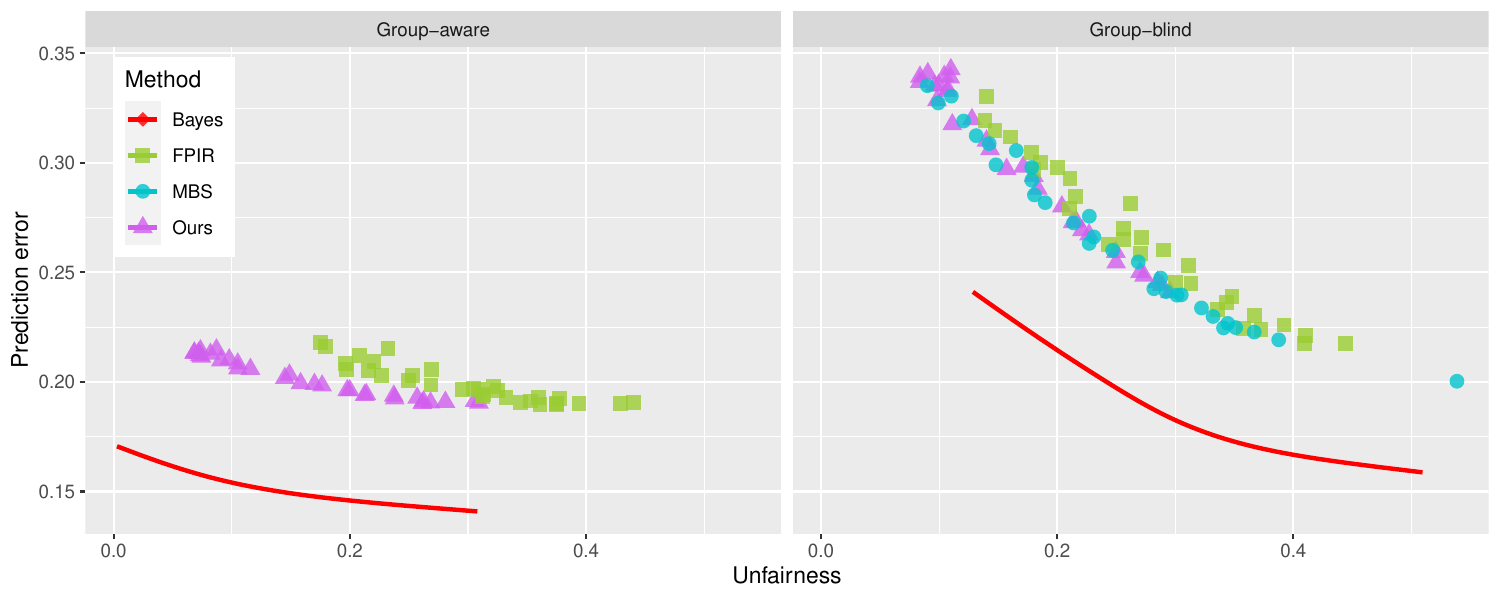}
        \caption{}
        \label{fig:simulation_quantile_m3}
    \end{subfigure}
    \caption{(a) The trade-off between prediction error and unfairness under (M3). The X-axis is the average unfairness measures $\bar\U_{\rm EOO}$ of the trained classifiers over 100 repetitions and the Y-axis is the average test prediction errors of these classifiers. The left and right panels correspond to the group-aware and group-blind scenarios, respectively. (b) As for (a) but the X-axis is the $95\%$ sample quantile $\U_{{\rm EOO},95}$ of the unfairness measures over 100 repetitions.}
    \label{fig:simulation_m3}
\end{figure}

\begin{figure}
    \centering
    \includegraphics[width=0.6\linewidth]{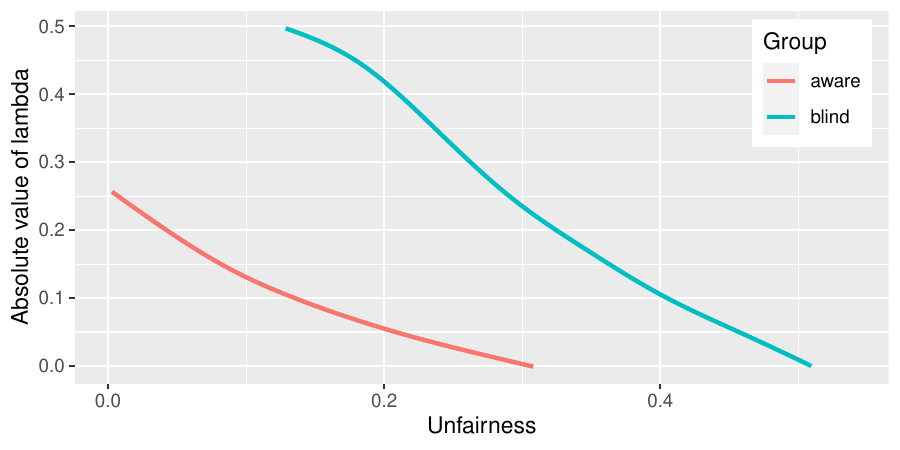}
    \caption{The curve of $|\lambda^{*G}_\alpha|$ on $\alpha$ under (M3). The red line is for the group-aware scenario and the cyan line is for the group-blind scenario.}
    \label{fig:lambda_m3}
\end{figure}

\section{Derivation of Example \ref{exa:phi_eoo}}\label{sec:supp_exa_phi_eoo}

In the group-aware scenario, we have
\begin{align*}
    &\Prob(Y_{f^{\rm aware}}(X,A)=1|Y=1,A=1)-\Prob(Y_{f^{\rm aware}}(X,A)=1|Y=1,A=2)\\
    =&\E(f^{\rm aware}(X,A)|Y=1,A=1)-\E(f^{\rm aware}(X,A)|Y=1,A=2)\\
    =&\frac{\E f^{\rm aware}(X,A)\1(Y=1,A=1)}{p_{1,1}}-\frac{\E f^{\rm aware}(X,A)\1(Y=1,A=2)}{p_{1,2}}\\
    =&\frac{\E\1(A=1)\eta^{\rm aware}(X,A)f^{\rm aware}(X,A)}{p_{1,1}}-\frac{\E\1(A=2)\eta^{\rm aware}(X,A)f^{\rm aware}(X,A)}{p_{1,2}}.
\end{align*}

In the group-blind scenario, we have
\begin{align*}
    &\Prob(Y_{f^{\rm blind}}(X,A)=1|Y=1,A=1)-\Prob(Y_{f^{\rm blind}}(X,A)=1|Y=1,A=2)\\
    =&\E (f^{\rm blind}(X,A)|Y=1,A=1)-\E(f^{\rm blind}(X,A)|Y=1,A=2)\\
    =&\frac{\E f^{\rm blind}(X,A)\1(Y=1,A=1)}{p_{1,1}}-\frac{\E f^{\rm blind}(X,A)\1(Y=1,A=2)}{p_{1,2}}\\
    =&\frac{\E\rho_{1|1}(X)\eta^{\rm blind}(X,A)f^{\rm blind}(X,A)}{p_{1,1}}-\frac{\E\rho_{2|1}(X)\eta^{\rm blind}(X,A)f^{\rm blind}(X,A)}{p_{1,2}}.
\end{align*}

\section{Proofs of Propositions~\ref{prop:bayes_binary} and \ref{prop:bayes}}
Since Proposition~\ref{prop:bayes_binary} is a special case of Proposition~\ref{prop:bayes}, we only state the proof for the latter.
\begin{proof}[Proof of Proposition \ref{prop:bayes}]
Since
    \begin{align*}
        \risk(f^G)=&\E\1(Y=1, Y_{f^G}=0)+\1(Y=0,Y_{f^G}=1)\\
        =&\E(1-f^G(X,A))\eta^G(X,A)+f^G(X,A)(1-\eta^G(X,A))\\
        =&p_Y+\E(1-2\eta^G(X,A))f^G(X,A),
    \end{align*}
    we can rewrite $f^{*G}_\alpha$ as the solution of the problem
    \begin{equation}\label{eq:bayes_new}
        \begin{aligned}
            f^{*G}_\alpha\in&\argmin_{f^G\in[0,1]^{\R^d\times[K]}}\E(1-2\eta^G(X,A))f^G(X,A),\quad{\rm s.t.}\quad \|\E\Phi^G(X,A)f^G(X,A)\|\le\alpha.
        \end{aligned}
    \end{equation}
    Considering the Lagrange function, Theorem 8.6.1 in \cite{luenberger1997optimization} and Corollary 3.3 in \cite{sion1958general} imply
    \begin{equation}\label{eq:bayes_minimax}
        \begin{aligned}
        &\sup_{\lambda\in\R^{\tilde K}}\E\big(1-2\eta^G(X,A)+\lambda^\top\Phi^G(X,A)\big)f^{*G}_\alpha(X,A)-\alpha\|\lambda\|_*\\
        =&\sup_{\nu\ge0}\sup_{\mu\in\R^{\tilde K},\|\mu\|_*\le 1}\E\big(1-2\eta^G(X,A)+\nu\mu^\top\Phi^G(X,A)\big)f^{*G}_\alpha(X,A)-\nu\alpha\\
        =&\sup_{\nu\ge0}\E(1-2\eta^G(X,A))f^{*G}_\alpha(X,A)+\nu(\|\E\Phi^G(X,A)f^{*G}_\alpha(X,A)\|-\alpha)\\
        =&\inf_{f^G\in[0,1]^{\R^d\times[K]}}\sup_{\nu\ge0}\E(1-2\eta^G(X,A))f^G(X,A)+\nu(\|\E\Phi^G(X,A)f^G(X,A)\|-\alpha)\\
        \overset{\text{\cite{luenberger1997optimization}}}{=}&\sup_{\nu\ge 0}\inf_{f^G\in[0,1]^{\R^d\times[K]}}\E(1-2\eta^G(X,A))f^G(X,A)+\nu(\|\E\Phi^G(X,A)f^G(X,A)\|-\alpha)\\
        =&\sup_{\nu\ge 0}\inf_{f^G\in[0,1]^{\R^d\times[K]}}\sup_{\mu\in\R^{\tilde K},\|\mu\|_*\le 1}\E\big(1-2\eta^G(X,A)+\nu\mu^\top\Phi^G(X,A)\big)f^G(X,A)-\nu\alpha\\
        \overset{\text{\cite{sion1958general}}}{=}&\sup_{\nu\ge 0}\sup_{\mu\in\R^{\tilde K},\|\mu\|_*\le 1}\inf_{f^G\in[0,1]^{\R^d\times[K]}}\E\big(1-2\eta^G(X,A)+\nu\mu^\top\Phi^G(X,A)\big)f^G(X,A)-\nu\alpha\\
        =&\sup_{\lambda\in\R^{\tilde K}}\inf_{f^G\in[0,1]^{\R^d\times[K]}}\E\big(1-2\eta^G(X,A)+\lambda^\top\Phi^G(X,A)\big)f^G(X,A)-\alpha\|\lambda\|_*\\
        =&-\inf_{\lambda\in\R^{\tilde K}}\big\{\E\big(2\eta^G(X,A)-1-\lambda^\top\Phi^G(X,A)\big)_++\alpha\|\lambda\|_*\big\}.
    \end{aligned}
    \end{equation}
    Denote
    \[\mu^*_\alpha\in\argmax_{\lambda\in\R^{\tilde K}}\E\big(1-2\eta^G(X,A)+\lambda^\top\Phi^G(X,A)\big)f^{*G}_\alpha(X,A)-\alpha\|\lambda\|_*,\]
    \[\lambda^*_\alpha\in\argmin_{\lambda\in\R^{\tilde K}}\E\big(2\eta^G(X,A)-1-\lambda^\top\Phi^G(X,A)\big)_++\alpha\|\lambda\|_*,\]
    \begin{align*}
        h^{*G}_\alpha(X,A)=&\1\big(2\eta^G(X,A)-1-\lambda^{*\top}_\alpha\Phi^G(X,A)>0\big)\\
        &+\tilde b^G(X,A)\1\big(2\eta^G(X,A)-1-\lambda^{*\top}_\alpha\Phi^G(X,A)=0\big),
    \end{align*}
    for some $\tilde b^G\in[0,1]^{\R^d\times[K]}$ and
    \[g^G(f^G,\lambda)=\E\big(1-2\eta^G(X,A)+\lambda^\top\Phi^G(X,A)\big)f^G(X,A)-\alpha\|\lambda\|_*,\]
    we know
    \[g^G(h_\alpha^{*G},\lambda^*_\alpha)\le g^G(f^{*G}_\alpha,\lambda^*_\alpha)\le g^G(f^{*G}_\alpha,\mu^*_\alpha).\]
    Together with Equation \eqref{eq:bayes_minimax} gives
    \begin{equation}\label{eq:saddle}
        g^G(h^{*G}_\alpha,\lambda^*_\alpha)= g^G(f^{*G}_\alpha,\lambda^*_\alpha)=g^G(f^{*G}_\alpha,\mu^*_\alpha).
    \end{equation}
    The first equality in~\eqref{eq:saddle} implies that $f^{*G}_\alpha$ must have the following form $P_{X,A}$-almost surely, 
    \begin{align*}
        f^{*G}_\alpha(X,A)=&\1\big(2\eta^G(X,A)-1-\lambda^{*\top}_\alpha\Phi^G(X,A)>0\big)\\
        &+b^G(X,A)\1\big(2\eta^G(X,A)-1-\lambda^{*\top}_\alpha\Phi^G(X,A)=0\big).
    \end{align*}
    The second equality in~\eqref{eq:saddle} implies
    \[\lambda^*_\alpha\in\argmax_{\lambda\in\R^{\tilde K}}\E\lambda^\top\Phi^G(X,A)f^{*G}_\alpha(X,A)-\alpha\|\lambda\|_*.\]
    We denote $\nu^*=\|\lambda^*_\alpha\|$, if $\nu^*=0$, then it holds trivially that
    \[\lambda^{*\top}_\alpha\E\Phi^G(X,A)f^{*G}_\alpha(X,A)=\|\lambda^*_\alpha\|_*\|\E\Phi^G(X,A)f^{*G}_\alpha(X,A)\|=\alpha\|\lambda^*_\alpha\|_*.\]
    If $\nu^*\ne 0$, we let $\mu^*=\frac{\lambda^*_\alpha}{\nu^*}$, then it is straightforward that
    \[(\nu^*,\mu^*)\in\argmax_{(\nu,\mu):\nu\ge0,\mu\in\R^{\tilde K},\|\mu\|_*\le 1}Q(\nu,\mu),\quad Q(\nu,\mu)=\E\nu\mu^\top\Phi^G(X,A)f^{*G}(X,A)-\alpha\nu.\]
    Since $\mu^*\in\argmax_{\mu\in\R^{\tilde K},\|\mu\|_*\le 1}Q(\nu^*,\mu)$, we know
    \[\mu^{*\top}_\alpha\E\Phi^G(X,A)f^{*G}_\alpha(X,A)=\|\E\Phi^G(X,A)f^{*G}_\alpha(X,A)\|.\]
    Similarly, since $\nu^*\in\argmax_{\nu\ge 0}Q(\nu,\mu^*)$ and $\|\E\Phi^G(X,A)f^{*G}_\alpha(X,A)\|\le\alpha$, $\nu^*>0$ implies
    \[\|\E\Phi^G(X,A)f^{*G}_\alpha(X,A)\|=\alpha.\]
    
    In conclusion, we have 
    \[\lambda^{*\top}_\alpha\E\Phi^G(X,A)f^{*G}_\alpha(X,A)=\|\lambda^*_\alpha\|_*\|\E\Phi^G(X,A)f^{*G}_\alpha(X,A)\|=\alpha\|\lambda^*_\alpha\|_*.\]

    Moreover, since the first three lines in Equation~\eqref{eq:bayes_minimax} holds for any classifier $f^G$, we know any minimizer of
    \[\argmin_{f^G\in[0,1]^{\R^d\times[K]}}\sup_{\lambda\in\R^{\tilde K}}g^G(f^G,\lambda)\]
    is a Bayes optimal classifier. For any function $b^G\in[0,1]^{\R^d\times[K]}$ such that
    \[\|\E\Phi^G(X,A)h^G(X,A)\|\le\alpha,\quad\lambda^{*\top}_\alpha\E\Phi^G(X,A)h^G(X,A)=\|\lambda^*_\alpha\|_*\|\E\Phi^G(X,A)h^G(X,A)\|=\alpha\|\lambda^*_\alpha\|_*,\]
    with
    \[h^G=\1\big(2\eta^G-1-\lambda^{*\top}_\alpha\Phi^G>0\big)
        +b^G\1\big(2\eta^G-1-\lambda^{*\top}_\alpha\Phi^G=0\big),\]
    we have
    \begin{align*}
        \sup_{\lambda\in\R^{\tilde K}}g^G(h^G,\lambda)=&\E\big(1-2\eta^G(X,A)\big)h^G(X,A)+\sup_{\lambda\in\R^{\tilde K}}\E\lambda^\top\Phi^G(X,A)h^G(X,A)-\alpha\|\lambda\|_*\\
        =&\E\big(1-2\eta^G(X,A)\big)h^G(X,A)\\
        =&\E\big(1-2\eta^G(X,A)\big)h^G(X,A)+\lambda^{*\top}_\alpha\E\Phi^G(X,A)h^G(X,A)-\alpha\|\lambda^*_\alpha\|_*\\
        =&g^G(h^G,\lambda^*_\alpha)\\
        =&g^G(f^{*G}_\alpha,\lambda^*_\alpha)\\
        =&\inf_{f^G\in[0,1]^{\R^d\times[K]}}\sup_{\lambda\in\R^{\tilde K}}g^G(f^G,\lambda),
    \end{align*}
    where the last equality comes from the fact that $(f^{*G}_\alpha,\lambda^*)$ is a saddle point due to Equation~\eqref{eq:saddle}. Therefore $h^G$ is a Bayes optimal classifier.
\end{proof}

\section{Proof of Lemma~\ref{lem:lambda_binary}}

\begin{proof}[Proof of Lemma~\ref{lem:lambda_binary}]
    Denote $s_\lambda=\sgn(\lambda^*_\alpha)$ with $\sgn(0)\in[-1,1]$, we separate the proof into two cases depending on whether $\lambda^*_\alpha=0$.

    \begin{enumerate}[1)]
        \item If $\lambda^*_\alpha=0$, we have $|\E\phi^G(X,A)\1\big(2\eta^G(X,A)>1\big)|\le\alpha$. Then $|\lambda^*_\alpha|=0$ is the smallest non-negative real number $\lambda_+$ such that
        \[s\E\phi^G(X,A)\1\big(2\eta^G(X,A)-1>s\lambda_+\phi^G(X,A)\big)\le\alpha.\]
        \item If $\lambda^*_\alpha\ne 0$, we know
        \[s_\lambda\E\phi^G(X,A)\1\big(2\eta^G(X,A)-1>s_\lambda|\lambda^*_\alpha|\phi^G(X,A)\big)=\alpha.\]
        Due to the non-increasing property of
        \[\lambda_+\rightarrow s_\lambda\E\phi^G(X,A)\1\big(2\eta^G(X,A)-1>s_\lambda \lambda_+\phi^G(X,A)\big),\]
        we get
        \[s_\lambda\E\phi^G(X,A)\1\big(2\eta^G(X,A)>1\big)\ge s_\lambda\E\phi^G(X,A)\1\big(2\eta^G(X,A)-1>s_\lambda |\lambda^*_\alpha|\phi^G(X,A)\big)\ge\alpha,\]
        which implies $s=s_\lambda$.
        
        In the following, we prove that $|\lambda^*_\alpha|$ is the smallest non-negative real number $\lambda_+$ such that
        \[s\E\phi^G(X,A)\1\big(2\eta^G(X,A)-1>s\lambda_+\phi^G(X,A)\big)\le\alpha.\]
        To this end, suppose there exists $0\le\lambda_+<|\lambda^*_\alpha|$ such that the above inequality is satisfied. It follows from the monotonicity that
        \[s\E\phi^G(X,A)\1\big(2\eta^G(X,A)-1>s\lambda_+\phi^G(X,A)\big)=\alpha.\]
        Denote $\Delta=|\lambda^*_\alpha|-\lambda_+$ and $g^{*G}_\alpha=2\eta^G-1-\lambda^*_\alpha\phi^G$, we have
        \begin{align*}
            0=& s\E\phi^G(X,A)\1\big(2\eta^G(X,A)-1>s|\lambda^*_\alpha|\phi^G(X,A)\big)\\
            &-s\E\phi^G(X,A)\1\big(2\eta^G(X,A)-1>s\lambda_+\phi^G(X,A)\big)\\
            =&s\E\phi^G(X,A)\1\big(-s\Delta\phi^G(X,A)\ge g^{*G}_\alpha(X,A)>0\big)\\
            &-s\E\phi^G(X,A)\1\big(-s\Delta\phi^G(X,A)<g^{*G}_\alpha(X,A)\le0\big)\\
            =&-\E|\phi^G(X,A)|\1\big(\Delta|\phi^G(X,A)|\ge g^{*G}_\alpha(X,A)>0, s\phi^G(X,A)<0\big)\\
            &-\E|\phi^G(X,A)|\1\big(-\Delta|\phi^G(X,A)|<g^{*G}_\alpha(X,A)\le 0, s\phi^G(X,A)>0\big),
        \end{align*}
        which implies
        \[\Prob\big(\Delta|\phi^G(X,A)|\ge g^{*G}_\alpha(X,A)>0, s\phi^G(X,A)<0\big)=0,\]
        \[\Prob\big(-\Delta|\phi^G(X,A)|<g^{*G}_\alpha(X,A)\le 0, s\phi^G(X,A)>0\big)=0.\]
        Then we check Problem~\eqref{eq:lambda_obj_binary} at $s\lambda_+$.
    \begin{align*}
        &\bigg(\E\big(2\eta^G(X,A)-1-\lambda^*_\alpha\phi^G(X,A)\big)_++\alpha|\lambda^*|\bigg)-\bigg(\E\big(2\eta^G(X,A)-1-s\lambda_+\phi^G(X,A)\big)_++\alpha\lambda_+\bigg)\\
        =&\E\big(2\eta^G(X,A)-1\big)\1\big(2\eta^G(X,A)-1>\lambda^*_\alpha\phi^G(X,A)\big)\\
        &-\E\big(2\eta^G(X,A)-1\big)\1\big(2\eta^G(X,A)-1>s\lambda_+\phi^G(X,A)\big)\\
        =&\E\big(2\eta^G(X,A)-1\big)\1\big(-s\Delta\phi^G(X,A)\ge g^{*G}_\alpha(X,A)>0\big)\\
        &-\E\big(2\eta^G(X,A)-1\big)\1\big(-s\Delta\phi^G(X,A)< g^{*G}_\alpha(X,A)\le 0\big)\\
        =&\E\big(2\eta^G(X,A)-1\big)\1(\Delta|\phi^G(X,A)|\ge g^{*G}_\alpha(X,A)>0, s\phi^G(X,A)<0\big)\\
        &-\E\big(2\eta^G(X,A)-1\big)\1(\Delta|\phi^G(X,A)|< g^{*G}_\alpha(X,A)\le0, s\phi^G(X,A)>0\big)\\
        =&0.
    \end{align*}
    So $s\lambda_+$ is also a minimizer of Problem~\eqref{eq:lambda_obj_binary} with $\lambda_+<|\lambda^*_\alpha|$ which contradicts the definition of $\lambda^*_\alpha$. Then we conclude the result that $|\lambda^*_\alpha|$ is the smallest non-negative real number $\lambda_+$ such that
        \[s\E\phi^G(X,A)\1\big(2\eta^G(X,A)-1>s\lambda_+\phi^G(X,A)\big)\le\alpha.\]
    \end{enumerate}
    Combining pieces proves the lemma.
\end{proof}

\section{Further Intuition about Algorithm \ref{alg:binary_unified}}
Although Lemma~\ref{lem:lambda_binary} involves $\eta^G$, it can be shown that the intuition remains effective even if we replace $\eta^G$ with any estimator $\hat\eta^G$. Denote $\tilde s^G=\sgn\big(\E\phi^G(X,A)\1\big(2\hat\eta^G(X,A)>1\big)\big)$, then we have the following lemma stating that there always exists $\tilde\lambda_+\ge 0$ such that the unfairness of $\1(2\hat\eta^G(X,A)-1>\tilde s^G\tilde\lambda_+\phi^G)$ is below $\alpha$.
\begin{Lemma}\label{lem:lambda_etahat_binary}
    Under the model set-up described above, suppose $\sup_{\lambda\in\R}\Prob(2\hat\eta^G(X,A)-1=\lambda\phi^G(X,A))=0$. If we define $\tilde\lambda_+$ to be
    \[\tilde\lambda_+=\argmin_{\lambda_+\ge 0}\lambda_+\quad{\rm s.t.}\quad \tilde s^G\E\phi^G(X,A)\1\big(2\hat\eta^G(X,A)-1>\tilde s^G\lambda_+\phi^G(X,A)\big)\le\alpha,\]
    then $\tilde\lambda_+$ is well-defined and $\U(\1(2\hat\eta^G-1>\tilde s^G\tilde\lambda_+\phi^G))\le\alpha$.
\end{Lemma}

Lemma~\ref{lem:lambda_etahat_binary} is due to the monotonicity of $\tilde s^G\E\phi^G(X,A)\1\big(2\hat\eta^G(X,A)-1>\tilde s^G\lambda_+\phi^G(X,A)\big)$ with respect to $\lambda_+$. It implies that for any $\hat\eta^G$, the fairness constraint can always be satisfied by shifting $\1(2\hat\eta^G-1>0)$ to $\1(2\hat\eta^G-1>\lambda\phi^G)$ for some $\lambda\in\R$. 

\begin{proof}[Proof of Lemma~\ref{lem:lambda_etahat_binary}]

    At first we show the existence of $\tilde\lambda_+$. Since 
    \begin{align*}
        &\tilde s^G\E\phi^G(X,A)\1\big(2\hat\eta^G(X,A)-1>\tilde s^G\lambda_+\phi^G(X,A)\big)\\
        =&\E|\phi^G(X,A)|\1\big(2\hat\eta^G(X,A)-1>\lambda_+|\phi^G(X,A)|,\tilde s^G\phi^G(X,A)>0\big)\\
        &-\E|\phi^G(X,A)|\1\big(2\hat\eta^G(X,A)-1>-\lambda_+|\phi^G(X,A)|,\tilde s^G\phi^G(X,A)<0\big),
    \end{align*}
    which is non-positive when $\lambda_+$ increases to infinity. Therefore $\tilde\lambda_+$ is always well defined.

    Then we verify the unfairness control in two cases.

    \textbf{Case (1)}: If $\tilde\lambda_+=0$, it follows from the definition of $\tilde s^G$ and $\tilde\lambda_+$ that
    \[\tilde s^G\E\phi^G(X,A)\1\big(2\hat\eta^G(X,A)>1\big)\in[0,\alpha],\]
    which implies
    \[\U(\1(2\hat\eta^G>1))\le\alpha.\]
    
    \textbf{Case (2)}: If $\tilde\lambda_+>0$, since 
    \[\sup_{\lambda\in\R}\Prob(2\hat\eta^G(X,A)-1=\lambda\phi^G(X,A))=0,\]
    we know
    \[\tilde s^G\E\phi^G(X,A)\1\big(2\hat\eta^G(X,A)-1>\tilde s^G\lambda_+\phi^G(X,A)\big)=\alpha.\]
    Therefore $\U(\1(2\hat\eta^G-1>\tilde s^G\tilde\lambda_+\phi^G))=\alpha$.
\end{proof}

\section{Proofs of Lemmas~\ref{lem:unfairness_dev_binary} and \ref{lem:unfairness_dev_multi}}

Lemmas~\ref{lem:unfairness_dev_binary} and \ref{lem:unfairness_dev_multi} follow directly from the following Lemma~\ref{lem:vc}, which can be obtained from Theorem 12.1 and Theorem 13.7 in \cite{Boucheron2013concentration}.
\begin{Lemma}[Empirical Process]\label{lem:vc}
    Suppose $Z_1,\ldots,Z_n\overset{\rm i.i.d.}{\sim}Z\in\mathcal{Z}$ and $\mathcal{C}$ is a class of subsets of $\mathcal{Z}$ with finite VC dimension $v$, then with probability at least $1-\delta$, we have
    \[\sup_{C\in\mathcal{C}}\abs{\frac{1}{n}\sum_{i=1}^n \1(Z_i\in C)-\Prob(Z\in\mathcal{C})}\le 72\sqrt{\frac{v\log 4e^2}{n}}+\sqrt{\frac{1}{2n}\log \frac{2}{\delta}}.\]
\end{Lemma}

\section{Modification to $\hat\eta^G,\hat\phi^G$ for fulfilling Assumption~\ref{ass:init}}\label{sec:supp_modification}

Without loss of generality, we assume $\Prob_{X,A}(\hat\phi^G(X,A)=0)=0$, otherwise, we can replace $\hat\phi^G$ by $\tilde\phi^G=\hat\phi^G+\epsilon_\phi\1(\hat\phi^G=0)$. Then $\Prob_{X,A}(\tilde\phi^G(X,A)=0)=0$ and $\|\tilde\phi^G-\phi^G\|_\infty\le 2\epsilon_\phi$. Similarly, we assume $\Prob_{X,A}(2\hat\eta^G(X,A)=1)=0$, otherwise, we replace $\hat\eta^G$ by $\tilde\eta^G=\hat\eta^G+\epsilon_\eta\1(2\hat\eta^G=1)$, then $\Prob_{X,A}(2\tilde\eta^G(X,A)=1)=0$ and $\|\tilde\eta^G-\eta^G\|_\infty\le 2\epsilon_\eta$.

If $X|A$ is continuous, we know $\Prob_{X,A}\big(2\hat\eta^G(X,A)-1=\lambda\hat\phi^G(X,A)\big)>0$ if and only if ${\rm Leb}(S_\lambda)>0$ with $S_\lambda=\big\{x:2\hat\eta^G(x,a)-1=\lambda\hat\phi^G(x,a),a\in[2]\big\},\lambda\ne0$. Since the CDF of $\frac{2\hat\eta^G(X,A)-1}{\hat\phi^G(X,A)}$ conditioned on $\hat\eta^G,\hat\phi^G$ has at most countably many discontinuous points, there are only countably many $\lambda$'s such that ${\rm Leb}(S_\lambda)>0$. For any such $\lambda\in\R$, on the set $S_\lambda$, we replace $\hat\eta^G(x,a)$ by $\tilde\eta^G(x,a)=\hat\eta^G(x,a)+(\frac{1}{2}-\hat\eta^G(x,a))\epsilon_\eta|\sin(x_1)|$ with $x_1$ to be the first coordinate of $x$. Then it is straightforward to verify that ${\rm Leb}(S_\lambda\cap \tilde S_{\tilde\lambda})={\rm Leb}\{\epsilon_\eta|\sin(x_1)|=1-\frac{\tilde\lambda}{\lambda}\}=0$ for all $\tilde\lambda\in\R$ where $\tilde S_\lambda=\big\{x:2\tilde\eta^G(x,a)-1=\lambda\hat\phi^G(x,a),a\in[2]\big\}$. Moreover $\sup_{x\in S_\lambda,a\in[2]}|\tilde\eta^G(x,a)-\eta^G(x,a)|\le\frac{3}{2}\epsilon_\eta$. Therefore Assumption~\ref{ass:init} is met after the modification.

\section{Proof of Theorem~\ref{thm:fair_binary}}

\begin{proof}[Proof of Theorem~\ref{thm:fair_binary}]

    Throughout the proof, the expectations are taken with respect to a new test sample $(X,A,Y)$ conditioned on the dataset $\D$.
    
    \textbf{Existence of $\hat\lambda^G$}:

    Denote the event $E$ as
    \[E=\bigg\{\sup_{\lambda\in\R}\bigg|\sum_{j\in[m]}\kappa_j(\hat\E_j-\E_{j})\1\big(2\hat\eta^G(X,A)-1>\lambda\hat\phi^G(X,A)\big)\bigg|\le\epsilon_\alpha\bigg\},\]
    Lemma \ref{lem:unfairness_dev_binary} implies $\Prob(E^c)\le\delta_{\rm post}$. Under event $E$, we have
    \begin{align*}
        &\hat s^G\sum_{j\in[m]}\kappa_j\hat\E_j\1\big(2\hat\eta^G(X,A)-1>\hat s^G\lambda_+\hat\phi^G(X,A)\big)\\
        \le&\hat s^G\sum_{j\in[m]}\kappa_j\E_j\1\big(2\hat\eta^G(X,A)-1>\hat s^G\lambda_+\hat\phi^G(X,A)\big)+\epsilon_\alpha\\
        =&\hat s^G\E\phi^G(X,A)\1\big(2\hat\eta^G(X,A)-1>\hat s^G\lambda_+\hat\phi^G(X,A)\big)+\epsilon_\alpha\\
        \le&\hat s^G\E\phi^G(X,A)\1\big(2\hat\eta^G(X,A)-1>\hat s^G\lambda_+\hat\phi^G(X,A),\phi^G(X,A)\hat\phi^G(X,A)>0\big)\\
        &+\E|\phi^G(X,A)|\1\big(\phi^G(X,A)\hat\phi^G(X,A)\le 0\big)+\epsilon_\alpha.\\
        =&\underbrace{\E|\phi^G(X,A)|\1\big(2\hat\eta^G(X,A)-1>\lambda_+|\hat\phi^G(X,A)|,\phi^G(X,A)\hat\phi^G(X,A)>0,\hat s^G\phi^G(X,A)>0\big)}_{T_1}\\
        &-\underbrace{\E|\phi^G(X,A)|\1\big(2\hat\eta^G(X,A)-1>-\lambda_+|\hat\phi^G(X,A)|,\phi^G(X,A)\hat\phi^G(X,A)>0,\hat s^G\phi^G(X,A)<0\big)}_{T_2}\\
        &+\tilde\epsilon_\phi^G+\epsilon_\alpha
    \end{align*}
    It is not hard to see that $T_1-T_2$ are non-increasing in $\lambda_+$ and $\lim_{\lambda_+\rightarrow+\infty} T_1-T_2\le 0$. Since $\alpha\ge 2\epsilon_\alpha+\tilde\epsilon_\phi^G$, for $\lambda_+$ large enough, we have
    \[\hat s^G\sum_{j\in[m]}\kappa_j\hat\E_j\1\big(2\hat\eta^G(X,A)-1>\hat s^G\lambda_+\hat\phi^G(X,A)\big)\le\tilde\epsilon_\phi^G+\epsilon_\alpha\le\alpha-\epsilon_\alpha,\]
    which implies $\hat\lambda_+^G$ is well defined.

    \textbf{Fairness constraint}:

    We prove this part by considering two cases separately.

    \textbf{Case (1)}: If
    \[\hat s^G\sum_{j\in[m]}\kappa_j\hat\E_j\1\big(2\hat\eta^G(X,A)-1>\hat\lambda^G\hat\phi^G(X,A)\big)\le\alpha-\epsilon_\alpha,\]
    we know that under event $E$,
    \begin{align*}
        &\hat s^G\E\phi^G(X,A)\1\big(2\hat\eta^G(X,A)-1>\hat\lambda^G\hat\phi^G(X,A)\big)\\
        \le&\hat s^G\sum_{j\in[m]}\kappa_j\hat\E_j\1\big(2\hat\eta^G(X,A)-1>\hat\lambda^G\hat\phi^G(X,A)\big)+\epsilon_\alpha\\
        \le&\alpha.
    \end{align*}
    If $\hat\lambda^G=0$, we get
    \begin{align*}
        &-\hat s^G\E\phi^G(X,A)\1\big(2\hat\eta^G(X,A)>1\big)\\
        \le&-\hat s^G\sum_{j\in[m]}\kappa_j\hat\E_j\1\big(2\hat\eta^G(X,A)>1\big)+\epsilon_\alpha\\
        \le&\epsilon_\alpha\\
        \le&\alpha.
    \end{align*}
    If $\hat\lambda^G_+>0$, for any $0<\Delta\le\hat\lambda_+^G$, it follows from the definition of $\hat\lambda_+^G$ that
    \begin{equation}\label{eq:fair_limit}
        \begin{aligned}
        &-\hat s^G\E\phi^G(X,A)\1\big(2\hat\eta^G(X,A)-1>\hat s^G(\hat\lambda_+^G-\Delta)\hat\phi^G(X,A)\big)\\
        \le&-\hat s^G\sum_{j\in[m]}\kappa_j\hat\E_j\1\big(2\hat\eta^G(X,A)-1>\hat s^G(\hat\lambda_+^G-\Delta)\hat\phi^G(X,A)\big)+\epsilon_\alpha\\
        \le&-\alpha+2\epsilon_\alpha\\
        \le&\alpha.
    \end{aligned}
    \end{equation}
    Setting $\Delta\rightarrow 0+$, since $\phi^G$ is bounded, then the continuity in Assumption~\ref{ass:init} implies
    \begin{align*}
        -\hat s^G\E\phi^G(X,A)\1\big(2\hat \eta^G(X,A)-1>\hat\lambda^G\hat\phi^G(X,A)\big)\le\alpha.
    \end{align*}
    Therefore, under the event $E$, we have
    \[\U(\hat f^G_\alpha)=|\E\phi^G(X,A)\1\big(2\hat \eta^G(X,A)-1>\hat\lambda^G\hat\phi^G(X,A)\big)|\le\alpha.\]

    \textbf{Case (2)}: If the empirical unfairness measure in Step 2 of Algorithm~\ref{alg:binary_unified} jumps at $\hat\lambda^G_+$ such that
    \[\hat s^G\sum_{j\in[m]}\kappa_j\hat\E_j\1\big(2\hat\eta^G(X,A)-1>\hat\lambda^G\hat\phi^G(X,A)\big)>\alpha-\epsilon_\alpha,\]
    then there exists a sequence $\{\lambda_{+t}:i\in\N_+\}$ such that $\lambda_{+t}\searrow\hat\lambda^G_+$ and 
    \[\hat s^G\sum_{j\in[m]}\kappa_j\hat\E_j\1\big(2\hat\eta^G(X,A)-1>\hat s^G\lambda_{+t}\hat\phi^G(X,A)\big)\le\alpha-\epsilon_\alpha,\quad\forall t\in\N_+.\]
    Then
    \begin{align*}
        &\hat s^G\E\phi^G(X,A)\1\big(2\hat\eta^G(X,A)-1>\hat s^G\lambda_{+t}\hat\phi^G(X,A)\big)\\
        \le&\hat s^G\sum_{j\in[m]}\kappa_j\hat\E_j\1\big(2\hat\eta^G(X,A)-1>\hat s^G\lambda_{+t}\hat\phi^G(X,A)\big)+\epsilon_\alpha\\
        \le&\alpha.
    \end{align*}
    Setting $t\rightarrow\infty $, since $\phi^G$ is bounded, then the continuity in Assumption~\ref{ass:init} implies that
    \[\hat s^G\E\phi^G(X,A)\1\big(2\hat\eta^G(X,A)-1>\hat\lambda^G\hat\phi^G(X,A)\big)\le\alpha.\]
    Using the same proof with \textbf{Case (1)}, we can also show that under the event $E$,
    \[\U(\hat f^G_\alpha)\le\alpha.\]

\end{proof}

\section{Derivations in Remark~\ref{rem:poly_growth}}\label{sec:supp_rem_poly_growth}

Assumption~\ref{ass:ratio_poly} supposes the unfairness difference $D$ satisfies that for any $z\in\R$,
    \[D(4z)\le c_2 D(z).\]
    It is not hard to see that
    \begin{align*}
        &\forall |z'|\ge |z|>0,~zz'>0,~1\ge\frac{D(z)}{D(z')}\ge\bigg(\frac{z}{z'}\bigg)^{\log_4 c_2}\\
        \Longrightarrow&\forall z\in\R, ~D(4z)\le c_2 D(z)\\
        \Longrightarrow&\forall |z'|\ge |z|>0,~zz'>0,~1\ge\frac{D(z)}{D(z')}\ge\bigg(\frac{z}{4z'}\bigg)^{\log_4 c_2}.
    \end{align*}
    If we fix $z'$ such that $|z'|$ is a large constant, under Assumption~\ref{ass:ratio_poly}, for any $z\in\R$ with $zz'\ge0$, $|z|\le|z'|$, we see
    \[D(z)\gtrsim |z|^{\log_4 c_2}.\]

\section{Original Versions of Assumptions \ref{ass:ratio_poly_simple} and \ref{ass:ratio_balance_simple}}\label{sec:original_assumption}

\begin{Assumption}[Polynomial Growth]\label{ass:ratio_poly}
    For some constant $c_2>0$, any $z>0$ and $j\in\{-1,1\}$, we have 
    \[\E\left[|\phi^G(X,A)|\1\bigg(0<\frac{jg^{*G}_\alpha(X,A)}{s^G\phi^G(X,A)}<4z\bigg)\right]\le c_2\E\left[|\phi^G(X,A)|\1\bigg(0<\frac{jg^{*G}_\alpha(X,A)}{s^G\phi^G(X,A)}<z\bigg)\right].\]
\end{Assumption}

\begin{Assumption}\label{ass:ratio_balance}
    There exist constants $c_3,c_4>0$ such that
    \[\E\left[|\phi^G(X,A)|\1\bigg(0>\frac{g^{*G}_\alpha(X,A)}{s^G\phi^G(X,A)}\ge-|\lambda^{*G}_\alpha|\bigg)\right]\le c_3\E\left[|\phi^G(X,A)|\1\bigg(0<\frac{g^{*G}_\alpha(X,A)}{s^G\phi^G(X,A)}\le c_4|\lambda^{*G}_\alpha|\bigg)\right].
    \]
\end{Assumption}
It is clear that if $\frac{2\eta^G(X,A)-1}{\phi^G(X,A)}$ is a continuous random variable given $\phi^G(X,A)\ne 0$, then for $j\in\{-1,1\}$, 
\[\E\left[|\phi^G(X,A)|\1\bigg(0<\frac{jg^{*G}_\alpha(X,A)}{s^G\phi^G(X,A)}<z\bigg)\right]=|U(\lambda^{*G}_\alpha+js^Gz)-U(\lambda^{*G}_\alpha)|,\]
which can be interpreted as the change of signed unfairness measures around $\lambda^*_\alpha$. Consequently, Assumptions \ref{ass:ratio_poly_simple} and \ref{ass:ratio_balance_simple} are equivalent to Assumptions \ref{ass:ratio_poly} and \ref{ass:ratio_balance}. In the following proof, we consider the original Assumptions \ref{ass:ratio_poly} and \ref{ass:ratio_balance}.

\section{Proof of Theorem \ref{thm:upper_unify_binary}}

\begin{proof}[Proof of Theorem \ref{thm:upper_unify_binary}]
    Throughout the proof, the expectations are taken with respect to a new test sample $(X,A,Y)$ conditioned on the dataset $\D$.
    
    The excess risk of $\hat f^G$ can be expressed as
    \begin{align*}
        &\risk(\hat f^G_\alpha)-\risk(f^{*G}_\alpha)\\
        =&\E(2\eta^G(X,A)-1)(f^{*G}_\alpha(X,A)-\hat f^G_\alpha(X,A))\\
        =&\underbrace{\E\abs{2\eta^G(X,A)-1-\lambda^{*G}_\alpha \phi^G(X,A)}\abs{f^{*G}_\alpha (X,A)-\hat f^G_\alpha(X,A)}}_{T_1}+\underbrace{\E\lambda^{*G}_\alpha \phi^G(X,A)(f^{*G}_\alpha (X,A)-\hat f^G_\alpha(X,A))}_{T_2}.
    \end{align*}
    While $T_1$ corresponds to the error of $\hat f_\alpha^G$ for $f_\alpha^{*G}$, $T_2$ measures the conservativeness of $\hat f_\alpha^G$ in fairness control.
    Then the rest of the proof focuses on the control of $T_1$ and $T_2$. 
    
    \textbf{At first, we summarize a sketch of the proof.} Denote $U(0)=\U(\1(2\eta^G>1))$, we consider two cases (1) $\alpha\ge U(0)+\tilde\epsilon_\eta^G+2\epsilon_\alpha$ and (2) $\alpha<U(0)-\tilde\epsilon_\eta^G\vee c_3\big(2\epsilon_\alpha+c_\phi c_1(2\epsilon_\eta+(1+2c_4)|\lambda^{*G}_\alpha|\epsilon_\phi)^\gamma\big)$. Case (1) happens if the unconstrained Bayes optimal classifier $\1(2\eta^G>1)$ is sufficiently fair and case (2) happens if $\1(2\eta^G>1)$ is sufficiently unfair. 
    
    In case (1), the unconstrained Bayes optimal classifier is already $\alpha$-fair. Therefore, $\lambda_\alpha^{*G}=0$ and thus $T_2=0$. $T_1$ then reduces to the excess risk of standard nonparametric classification.

    In case (2), $\lambda_\alpha^{*G}\ne 0$, then we need to study both $T_1$ and $T_2$. The bound for $T_2$ is mainly due to the deviation between empirical and population unfairness measures, and can be analyzed straightforwardly. For $T_1$, recall $\hat f_\alpha^G$ consists of three components: $\hat\eta^G$, $\hat\phi^G$, and $\hat\lambda^G$. Since the errors of $\hat\eta^G$ and $\hat\phi^G$ have been assumed, we only need to analyze the error of $\hat\lambda^G$. To this end, we further consider two cases (a) $s^G\hat\lambda^G\ge s^G\lambda^{*G}_\alpha $, and (b) $s^G\hat\lambda^G<s^G\lambda^{*G}_\alpha $. Then Assumptions \ref{ass:ratio_poly} and \ref{ass:ratio_balance} allow us to control $|\hat\lambda^G-\lambda_\alpha^{*G}|$.
    
    \textbf{Then we provide the formal proof.}
    Denote
    \[E=\bigg\{\sup_{\lambda\in\R}\bigg|\sum_{j\in[m]}\kappa_j(\hat\E_j-\E_{j})\1\big(2\hat\eta^G(X,A)-1>\lambda\hat\phi^G(X,A)\big)\bigg|\le\epsilon_\alpha\bigg\}.\]
    
    \textbf{Case (1)}: If $\alpha\ge U(0)+\tilde\epsilon_\eta^G+2\epsilon_\alpha$, we know $\lambda^{*G}_\alpha =0$. Since
    \begin{align*}
        &|\E\phi^G(X,A)\1(2\eta^G(X,A)>1)-\E\phi^G(X,A)\1(2\hat\eta^G(X,A)>1)|\\
        \le&\E|\phi^G(X,A)|\1(|2\eta^G(X,A)-1|\le 2\|\hat\eta^G-\eta^G\|_\infty)\\
        \le&\tilde\epsilon_\eta^G,
    \end{align*}
    we have under $E$,
    \begin{align*}
        &\bigg|\sum_{j\in[m]}\kappa_j\hat \E_j\1(2\hat\eta^G(X,A)>1)\bigg|\\
        \le&\bigg|\sum_{j\in[m]}\kappa_j\E_{j}\1(2\hat\eta^G(X,A)>1)\bigg|+\epsilon_\alpha\\
        \le&U(0)+\tilde\epsilon_\eta^G+\epsilon_\alpha\\
        \le& \alpha-\epsilon_\alpha,
    \end{align*}
    which implies $\hat\lambda^G=0$. Then the excess risk can be controlled as
    \begin{align*}
        &\risk(\hat f^G_\alpha)-\risk(f^{*G}_\alpha)\\
        =&\E|2\eta^G(X,A)-1||\1(2\eta^G(X,A)>1)-\1(2\hat\eta^G(X,A)>1)|\\
        =&\E|2\eta^G(X,A)-1|\1\big(0<2\eta^G(X,A)-1\le2(\eta^G(X,A)-\hat\eta^G(X,A))\big)\\
        &+\E|2\eta^G(X,A)-1|\1\big(0\ge2\eta^G(X,A)-1>2(\eta^G(X,A)-\hat\eta^G(X,A))\big)\\
        \le&\E|2\eta^G(X,A)-1|\1\big(|2\eta^G(X,A)-1|\le2\|\hat\eta^G-\eta^G\|_\infty\big)\\
        \lesssim&\epsilon_\eta^{1+\gamma}.
    \end{align*}

    \textbf{Case (2)}: If $\alpha<U(0)-\tilde\epsilon_\eta^G\vee c_3\big(2\epsilon_\alpha+c_\phi c_1(2\epsilon_\eta+(1+2c_4)|\lambda^{*G}_\alpha|\epsilon_\phi)^\gamma\big)$, we have
    \[\E\lambda^{*G}_\alpha \phi^G(X,A)f^{*G}_\alpha (X,A)=|\lambda^{*G}_\alpha|\alpha.\]
    Since $s^G=\sgn\big(\sum_{j\in[m]}\kappa_j\E_j\1(2\eta^G(X,A)>1)\big)$, under $E$, it happens
    \begin{align*}
        &s^G\sum_{j\in[m]}\kappa_j\hat\E_j\1(2\hat\eta^G(X,A)>1)\\
        \ge& s^G\sum_{j\in[m]}\kappa_j\E_{j}\1(2\hat\eta^G(X,A)>1)-\epsilon_\alpha\\
        \ge&s^G\E\phi^G(X,A)\1(2\eta^G(X,A)>1)-\tilde\epsilon_\eta^G-\epsilon_\alpha\\
        =&U(0)-\tilde\epsilon_\eta^G-\epsilon_\alpha\\
        >&0,
    \end{align*}
    therefore $\hat s^G=s^G$. Then it must happen that $\hat\lambda_+^G>0$, otherwise if $\hat\lambda_+^G=0$,
    \[\alpha\ge s^G\E\phi^G(X,A)\1(2\hat\eta^G(X,A)>1)\ge s^G\E\phi^G(X,A)\1(2\eta^G(X,A)>1)-\tilde\epsilon_\eta^G
        >\alpha,\]
    which is impossible. Therefore, by Equation \eqref{eq:fair_limit},
    \begin{align*}
        &\E\lambda^{*G}_\alpha \phi^G(X,A)\hat f^G_\alpha(X,A)\\
        =&\abs{\lambda^{*G}_\alpha }s^G\E\phi^G(X,A)\1\big(2\hat \eta^G(X,A)-1>\hat\lambda^G\hat\phi^G(X,A)\big)\\
        =&\abs{\lambda^{*G}_\alpha }\lim_{\Delta\rightarrow 0+}s^G\E\phi^G(X,A)\1(2\hat\eta^G(X,A)-1>\hat s^G(\hat\lambda_+^G-\Delta)\hat\phi^G(X,A))\\
        \ge&\abs{\lambda^{*G}_\alpha }(\alpha-2\epsilon_\alpha).
    \end{align*}
    Then we can control $T_2$ as
    \[T_2\le 2\abs{\lambda^{*G}_\alpha }\epsilon_\alpha.\]
    
    To analyze $T_1$, we should bound $|\hat\lambda^G-\lambda^{*G}_\alpha |$. Now we consider the following two cases \textbf{(a)} $s^G\hat\lambda^G\ge s^G\lambda^{*G}_\alpha $ and \textbf{(b)} $s^G\hat\lambda^G<s^G\lambda^{*G}_\alpha $. Denote
    \[s^G_\phi(x,a)=\sgn(\phi^G(x,a)),\quad\hat\epsilon_g=2\epsilon_\eta+|\hat\lambda^G|\epsilon_\phi,\quad \epsilon_g=2\epsilon_\eta+|\lambda^{*G}_\alpha |\epsilon_\phi,\]
    and $s^G_\phi(x,a)=1$ when $\phi^G(x,a)=0$.
    
    \textbf{(a)}. If $s^G\hat\lambda^G\ge s^G\lambda^{*G}_\alpha $, we denote $\Delta=s^G\hat\lambda^G-s^G\lambda^{*G}_\alpha $. Then we know
    \begin{align*}
        &s^G\E\phi^G(X,A)\1(g^{*G}_\alpha (X,A)>0)\\
        =&\alpha\\
        \le&s^G\E\phi^G(X,A)\1\big(2\hat\eta^G(X,A)-1>\hat\lambda^G\hat\phi^G(X,A)\big)+2\epsilon_\alpha\\
        =&s^G\E\phi^G(X,A)\1\big(g^{*G}_\alpha (X,A)>2(\eta^G(X,A)-\hat\eta^G(X,A))\\
        &\qquad+(\hat\lambda^G-\lambda^{*G}_\alpha )\phi^G(X,A)+\hat\lambda^G(\hat\phi^G(X,A)-\phi^G(X,A))\big)+2\epsilon_\alpha\\
        \le&s^G\E\phi^G(X,A)\1\big(g^{*G}_\alpha (X,A)>s_\phi^G(X,A)s^G\big(\Delta|\phi^G(X,A)|-\hat\epsilon_g\big)\big)+2\epsilon_\alpha.
    \end{align*}
    Therefore we have
    \begin{equation}\label{eq:Delta_lambda_upper}
    \begin{aligned}
        0\le&2\epsilon_\alpha+s^G\E\phi^G(X,A)\1\big(g^{*G}_\alpha (X,A)>s_\phi^G(X,A)s^G(\Delta|\phi^G(X,A)|-\hat\epsilon_g)\big)\\
        &-s^G\E\phi^G(X,A)\1(g^{*G}_\alpha (X,A)>0)\\
        =&2\epsilon_\alpha+\E|\phi^G(X,A)|\1\big(g^{*G}_\alpha (X,A)>\Delta|\phi^G(X,A)|-\hat\epsilon_g,s^G\phi^G(X,A)>0\big)\\
        &-\E|\phi^G(X,A)|\1\big(g^{*G}_\alpha (X,A)>-\Delta|\phi^G(X,A)|+\hat\epsilon_g,s^G\phi^G(X,A)<0\big)\\
        &-\E|\phi^G(X,A)|\1(g^{*G}_\alpha (X,A)>0,s^G\phi^G(X,A)>0)+\E|\phi^G(X,A)|\1(g^{*G}_\alpha (X,A)>0,s^G\phi^G(X,A)<0)\\
        =&2\epsilon_\alpha+\E|\phi^G(X,A)|\1\big(0\ge g^{*G}_\alpha (X,A)>\Delta|\phi^G(X,A)|-\hat\epsilon_g),s^G\phi^G(X,A)>0\big)\\
        &-\E|\phi^G(X,A)|\1\big(0<g^{*G}_\alpha (X,A)\le\Delta|\phi^G(X,A)|-\hat\epsilon_g,s^G\phi^G(X,A)>0\big)\\
        &+\E|\phi^G(X,A)|\1\big(0<g^{*G}_\alpha (X,A)\le-\Delta|\phi^G(X,A)|+\hat\epsilon_g,s^G\phi^G(X,A)<0\big)\\
        &-\E|\phi^G(X,A)|\1\big(0\ge g^{*G}_\alpha (X,A)>-\Delta|\phi^G(X,A)|+\hat\epsilon_g,s^G\phi^G(X,A)<0\big)\\
        \le&2\epsilon_\alpha+\E|\phi^G(X,A)|\1\big(0>s_\phi^G(X,A)s^Gg^{*G}_\alpha (X,A)\ge\Delta|\phi^G(X,A)|-\hat\epsilon_g\big)\\
        &-\E|\phi^G(X,A)|\1\big(0< s_\phi^G(X,A)s^Gg^{*G}_\alpha (X,A)<\Delta|\phi^G(X,A)|-\hat\epsilon_g\big)\\
        \le&2\epsilon_\alpha+\E|\phi^G(X,A)|\1\big(0> s_\phi^G(X,A)s^Gg^{*G}_\alpha (X,A)\ge-\hat\epsilon_g,\Delta|\phi^G(X,A)|<\hat\epsilon_g\big)\\
        &-\E|\phi^G(X,A)|\1\big(0<s_\phi^G(X,A)s^Gg^{*G}_\alpha (X,A)<\frac{1}{2}\Delta|\phi^G(X,A)|,\Delta|\phi^G(X,A)|>2\hat\epsilon_g\big).
    \end{aligned}
    \end{equation}
    Now we can control $T_1$ as
    \begin{equation}\label{eq:excess_a}
        \begin{aligned}
        T_1=&\E|g^{*G}_\alpha (X,A)|\1\big(0<g^{*G}_\alpha (X,A)\le(\hat\lambda^G-\lambda^{*G}_\alpha )\phi^G(X,A)\\
        &\qquad+\hat\lambda^G(\hat\phi^G(X,A)-\phi^G(X,A))+2(\eta^G(X,A)-\hat\eta^G(X,A))\big)\\
        &+\E|g^{*G}_\alpha (X,A)|\1\big(0\ge g^{*G}_\alpha (X,A)>(\hat\lambda^G-\lambda^{*G}_\alpha )\phi^G(X,A)\\
        &\qquad+\hat\lambda^G(\hat\phi^G(X,A)-\phi^G(X,A))+2(\eta^G(X,A)-\hat\eta^G(X,A))\big)\\
        \le&\E|g^{*G}_\alpha (X,A)|\1\big(0<g^{*G}_\alpha (X,A)\le\Delta|\phi^G(X,A)|+\hat\epsilon_g,s^G\phi^G(X,A)\ge0\big)\\
        &+\E|g^{*G}_\alpha (X,A)|\1\big(0<g^{*G}_\alpha (X,A)\le-\Delta|\phi^G(X,A)|+\hat\epsilon_g,s^G\phi^G(X,A)<0\big)\\
        &+\E|g^{*G}_\alpha (X,A)|\1\big(0\ge g^{*G}_\alpha (X,A)>\Delta|\phi^G(X,A)|-\hat\epsilon_g,s^G\phi^G(X,A)\ge 0\big)\\
        &+\E|g^{*G}_\alpha (X,A)|\1\big(0\ge g^{*G}_\alpha (X,A)>-\Delta|\phi^G(X,A)|-\hat\epsilon_g,s^G\phi^G(X,A)<0\big)\\
        \le&\E|g^{*G}_\alpha (X,A)|\1\big(0< s_\phi^G(X,A)s^Gg^{*G}_\alpha (X,A)\le\Delta|\phi^G(X,A)|+\hat\epsilon_g\big)\\
        &+\E|g^{*G}_\alpha (X,A)|\1\big(0> s_\phi^G(X,A)s^Gg^{*G}_\alpha (X,A)\ge-\hat\epsilon_g\big)\\
        \le&\E|g^{*G}_\alpha (X,A)|\1\big(0<s_\phi^G(X,A)s^Gg^{*G}_\alpha (X,A)<2\Delta|\phi^G(X,A)|\big)\\
        &+\E|g^{*G}_\alpha (X,A)|\1\big(0<s_\phi^G(X,A)s^Gg^{*G}_\alpha (X,A)\le2\hat\epsilon_g\big)+c\hat\epsilon_g^{1+\gamma}\\
        \overset{\text{Assumption~\ref{ass:ratio_poly}}}{\lesssim}&\E\Delta|\phi^G(X,A)|\1\big(0<s_\phi^G(X,A)s^Gg^{*G}_\alpha (X,A)<\frac{1}{2}\Delta|\phi^G(X,A)|\big)+\hat\epsilon_g^{1+\gamma}\\
        =&\E\Delta|\phi^G(X,A)|\1\big(0<s_\phi^G(X,A)s^Gg^{*G}_\alpha (X,A)<\frac{1}{2}\Delta|\phi^G(X,A)|,\Delta|\phi^G(X,A)|\le2\hat\epsilon_g\big)\\
        &+\E\Delta|\phi^G(X,A)|\1\big(0<s_\phi^G(X,A)s^Gg^{*G}_\alpha (X,A)<\frac{1}{2}\Delta|\phi^G(X,A)|,\Delta|\phi^G(X,A)|>2\hat\epsilon_g\big)+\hat\epsilon_g^{1+\gamma}\\
        \le&\E\Delta|\phi^G(X,A)|\1\big(0<s^G_\phi(X,A)s^Gg^{*G}_\alpha (X,A)<\frac{1}{2}\Delta|\phi^G(X,A)|,\Delta|\phi^G(X,A)|>2\hat\epsilon_g\big)+c\hat\epsilon_g^{1+\gamma}\\
        \overset{\text{Equation~\eqref{eq:Delta_lambda_upper}}}{\le}&\E\Delta|\phi^G(X,A)|\1\big(0>s_\phi^G(X,A)s^Gg^{*G}_\alpha (X,A)\ge-\hat\epsilon_g,\Delta|\phi^G(X,A)|<\hat\epsilon_g\big)+2\Delta\epsilon_\alpha+c\hat\epsilon_g^{1+\gamma}\\
        \lesssim&\hat\epsilon_g^{1+\gamma}+\Delta\epsilon_\alpha.
    \end{aligned}
    \end{equation}

    Now we argue that if $\alpha<U(0)-\tilde\epsilon_\eta^G\vee c_3\big(2\epsilon_\alpha+c_\phi c_1(2\epsilon_\eta+(1+2c_4)|\lambda^*_\alpha|\epsilon_\phi)^\gamma\big)$, it must happen $\Delta\le 2c_4|\lambda^*_\alpha |$. To show this, we start from the case $\alpha< U(0)-\tilde\epsilon_\eta^G\vee c_3(2\epsilon_\alpha+c_\phi c_1\hat\epsilon_g^\gamma)$. Since
    \begin{align*}
        &s^G\E\phi^G(X,A)\1(2\eta^G(X,A)>1)\\
        >&\alpha+2c_3\epsilon_\alpha+c_\phi c_1c_3\hat\epsilon_g^\gamma\\
        =&s^G\E\phi^G(X,A)\1(g^{*G}_\alpha (X,A)>0)+2c_3\epsilon_\alpha+c_\phi c_1c_3\hat\epsilon_g^\gamma,
    \end{align*}
    we have
    \begin{equation}\label{eq:abs_lambda}
        \begin{aligned}
        0<&s^G\E\phi^G(X,A)\1(g^{*G}_\alpha (X,A)>-\lambda^{*G}_\alpha \phi^G(X,A))-s^G\E\phi^G(X,A)\1(g^{*G}_\alpha (X,A)>0)-2c_3\epsilon_\alpha-c_\phi c_1c_3\hat\epsilon_g^\gamma\\
        =&s^G\E\phi^G(X,A)\1\big(0\ge g^{*G}_\alpha(X,A)>-\lambda^*_\alpha\phi^G(X,A)\big)\\
        &-s^G\E\phi^G(X,A)\1\big(0<g^{*G}_\alpha(X,A)\le-\lambda^*_\alpha\phi^G(X,A)\big)-2c_3\epsilon_\alpha-c_\phi c_1c_3\hat\epsilon_g^\gamma\\
        =&\E|\phi^G(X,A)|\1\big(0>g^{*G}_\alpha(X,A)>-|\lambda^*_\alpha||\phi^G(X,A)|,s^G\phi^G(X,A)>0\big)\\
        &-\E|\phi^G(X,A)|\1\big(0>g^{*G}_\alpha(X,A)>|\lambda^*_\alpha||\phi^G(X,A)|,s^G\phi^G(X,A)<0\big)\\
        &-\E|\phi^G(X,A)|\1\big(0<g^{*G}_\alpha(X,A)\le-|\lambda^*_\alpha||\phi^G(X,A)|,s^G\phi^G(X,A)>0\big)\\
        &+\E|\phi^G(X,A)|\1\big(0<g^{*G}_\alpha(X,A)\le|\lambda^*_\alpha||\phi^G(X,A)|,s^G\phi^G(X,A)<0\big)-2c_3\epsilon_\alpha-c_\phi c_1c_3\hat\epsilon_g^\gamma\\
        \le&\E|\phi^G(X,A)|\1\big(0>s_\phi^G(X,A)s^Gg^{*G}_\alpha(X,A)\ge-|\lambda^*_\alpha||\phi^G(X,A)|\big)-2c_3\epsilon_\alpha-c_\phi c_1c_3\hat\epsilon_g^\gamma.
    \end{aligned}
    \end{equation}

    It follows from Equation \eqref{eq:Delta_lambda_upper} that
    \begin{align*}
        0\le&2\epsilon_\alpha+\E|\phi^G(X,A)|\1\big(0> s_\phi^G(X,A)s^Gg^{*G}_\alpha (X,A)\ge-\hat\epsilon_g\big)\\
        &-\E|\phi^G(X,A)|\1\big(0<s_\phi^G(X,A)s^Gg^{*G}_\alpha (X,A)<\Delta|\phi^G(X,A)|-\hat\epsilon_g\big)\\
        \le&2\epsilon_\alpha+\E|\phi^G(X,A)|\1\big(0> s_\phi^G(X,A)s^Gg^{*G}_\alpha (X,A)\ge-\hat\epsilon_g\big)\\
        &-\E|\phi^G(X,A)|\1\big(0<s_\phi^G(X,A)s^Gg^{*G}_\alpha (X,A)<\frac{1}{2}\Delta|\phi^G(X,A)|\big)\\
        &+\E|\phi^G(X,A)|\1\big(0<s_\phi^G(X,A)s^Gg^{*G}_\alpha (X,A)<\hat\epsilon_g\big)\\
        \le&2\epsilon_\alpha+\E|\phi^G(X,A)|\1\big(|g^{*G}_\alpha (X,A)|\le\hat\epsilon_g\big)-\E|\phi^G(X,A)|\1\bigg(0<\frac{g^{*G}_\alpha (X,A)}{s^G\phi^G(X,A)}<\frac{\Delta}{2}\bigg).
    \end{align*}
    Note that $c_4|\lambda^*_\alpha |\ge\frac{\Delta}{2}$ if
    \[\E|\phi^G(X,A)|\1\bigg(0<\frac{g^{*G}_\alpha(X,A)}{s^G\phi^G(X,A)}\le c_4|\lambda^*_\alpha|\bigg)>\E|\phi^G(X,A)|\1\bigg(0<\frac{g^{*G}(X,A)}{s^G\phi^G(X,A)}<\frac{\Delta}{2}\bigg),\]
    then it suffices to show
    \[\E|\phi^G(X,A)|\1\bigg(0<\frac{g^{*G}_\alpha (X,A)}{s^G\phi^G(X,A)}\le c_4|\lambda^*_\alpha |\bigg)>2\epsilon_\alpha+\E|\phi^G(X,A)|\1\big(|g^{*G}_\alpha (X,A)|\le\hat\epsilon_g\big).\]
    By Assumption \ref{ass:ratio_balance} and Equation \eqref{eq:abs_lambda}, we get
    \begin{align*}
        &\E|\phi^G(X,A)|\1\bigg(0<\frac{g^{*G}_\alpha (X,A)}{s^G\phi^G(X,A)}\le c_4|\lambda^*_\alpha |\bigg)\\
        \overset{\text{Assumption~\ref{ass:ratio_balance}}}{\ge}&\frac{1}{c_3}\E|\phi^G(X,A)|\1\bigg(0>\frac{g^{*G}_\alpha (X,A)}{s^G\phi^G(X,A)}\ge-|\lambda^*_\alpha |\bigg)\\
        \overset{\text{Equation~\eqref{eq:abs_lambda}}}{>}&2\epsilon_\alpha+c_\phi c_1\hat\epsilon_g^\gamma\\
        \ge&2\epsilon_\alpha+\E|\phi^G(X,A)|\1\big(|g^{*G}_\alpha (X,A)|\le\hat\epsilon_g\big),
    \end{align*}
    which means $c_4|\lambda^*_\alpha |\ge\frac{\Delta}{2}$. Since
    \[\hat\epsilon_g=2\epsilon_\eta+(|\lambda^*_\alpha|+\Delta)\epsilon_\phi\le2\epsilon_\eta+(1+2c_4)|\lambda^*_\alpha|\epsilon_\phi,\]
    and $|\lambda^*_\alpha |$ is non-increasing with $\alpha$, we know $2c_4|\lambda^*_\alpha |\ge\Delta$ also holds if
    \[\alpha\le U(0)-\tilde\epsilon_\eta^G\vee c_3\big(2\epsilon_\alpha+c_\phi c_1(2\epsilon_\eta+(1+2c_4)|\lambda^*_\alpha|\epsilon_\phi)^\gamma\big).\]
    
    Then the excess risk can be upper bounded as
    \[\risk(\hat f^G)-\risk(f^{*G}_\alpha)\lesssim \abs{\lambda^*_\alpha }\epsilon_\alpha+\hat\epsilon_g^{1+\gamma}+\Delta\epsilon_\alpha\lesssim|\lambda^*_\alpha |\epsilon_\alpha+\epsilon_\eta^{1+\gamma}+\big(|\lambda^*_\alpha|\epsilon_\phi\big)^{1+\gamma}.\]

    \textbf{(b)}. If $s^G\hat\lambda^G<s^G\lambda^{*G}_\alpha $, we know $|\hat\lambda^G|<|\lambda^{*G}_\alpha |$ then $\hat\epsilon_g<\epsilon_g$, we denote $\Delta=s^G\lambda^{*G}_\alpha -s^G\hat\lambda^G$. Similarly, we have
    \begin{align*}
        &s^G\E\phi^G(X,A)\1(g^{*G}_\alpha (X,A)>0)\\
        =&\alpha\\
        \ge&s^G\E\phi^G(X,A)\1\big(2\hat\eta^G(X,A)-1>\hat\lambda^G\hat\phi^G(X,A)\big)\\
        \ge&s^G\E\phi^G(X,A)\1\big(g^{*G}_\alpha (X,A)>s_\phi^G(X,A)s^G\big(\hat\epsilon_g-\Delta|\phi^G(X,A)|\big)\big),
    \end{align*}
    then
    \begin{equation}\label{eq:Delta_lambda_lower}
        \begin{aligned}
            0\le&s^G\E\phi^G(X,A)\1(g^{*G}_\alpha (X,A)>0)-s^G\E\phi^G(X,A)\1\big(g^{*G}_\alpha (X,A)>s_\phi^G(X,A)s^G\big(\hat\epsilon_g-\Delta|\phi^G(X,A)|\big)\big)\\
            =&\E|\phi^G(X,A)|\1(g^{*G}_\alpha (X,A)>0,s^G\phi^G(X,A)>0)-\E|\phi^G(X,A)|\1(g^{*G}_\alpha (X,A)>0,s^G\phi^G(X,A)<0)\\
            &-\E|\phi^G(X,A)|\1\big(g^{*G}_\alpha (X,A)>\hat\epsilon_g-\Delta|\phi^G(X,A)|,s^G\phi^G(X,A)>0\big)\\
            &+\E|\phi^G(X,A)|\1\big(g^{*G}_\alpha (X,A)>-\hat\epsilon_g+\Delta|\phi^G(X,A)|,s^G\phi^G(X,A)<0\big)\\
            =&\E|\phi^G(X,A)|\1\big(0<g^{*G}_\alpha (X,A)\le\hat\epsilon_g-\Delta|\phi^G(X,A)|,s^G\phi^G(X,A)>0\big)\\
            &-\E|\phi^G(X,A)|\1\big(0\ge g^{*G}_\alpha (X,A)>\hat\epsilon_g-\Delta|\phi^G(X,A)|,s^G\phi^G(X,A)>0\big)\\
            &+\E|\phi^G(X,A)|\1\big(0\ge g^{*G}_\alpha (X,A)>-\hat\epsilon_g+\Delta|\phi^G(X,A)|,s^G\phi^G(X,A)<0\big)\\
            &-\E|\phi^G(X,A)|\1\big(0<g^{*G}_\alpha (X,A)\le-\hat\epsilon_g+\Delta|\phi^G(X,A)|,s^G\phi^G(X,A)<0\big)\\
            \le&\E|\phi^G(X,A)|\1\big(0<s_\phi^G(X,A)s^Gg^{*G}_\alpha (X,A)\le\hat\epsilon_g-\Delta|\phi^G(X,A)|\big)\\
            &-\E|\phi^G(X,A)|\1\big(0>s_\phi^G(X,A)s^Gg^{*G}_\alpha (X,A)>\hat\epsilon_g-\Delta|\phi^G(X,A)|\big)\\
            \le&\E|\phi^G(X,A)|\1\big(0<s_\phi^G(X,A)s^Gg^{*G}_\alpha (X,A)\le\hat\epsilon_g,\Delta|\phi^G(X,A)|<\hat\epsilon_g\big)\\
            &-\E|\phi^G(X,A)|\1\big(0>s_\phi^G(X,A)s^Gg^{*G}_\alpha (X,A)>-\frac{1}{2}\Delta|\phi^G(X,A)|,\Delta|\phi^G(X,A)|>2\hat\epsilon_g\big).
        \end{aligned}
    \end{equation}
    Now we can control $T_1$ as
    \begin{equation}\label{eq:excess_b}
        \begin{aligned}
        T_1=&\E|g^{*G}_\alpha (X,A)|\1\big(0<g^{*G}_\alpha (X,A)\le(\hat\lambda-\lambda^*_\alpha )\phi^G(X,A)\\
        &\qquad+\hat\lambda(\hat\phi^G(X,A)-\phi^G(X,A))+2(\eta^G(X,A)-\hat\eta^G(X,A))\big)\\
        &+\E|g^{*G}_\alpha (X,A)|\1\big(0\ge g^{*G}_\alpha (X,A)>(\hat\lambda-\lambda^*_\alpha )\phi^G(X,A)\\
        &\qquad+\hat\lambda(\hat\phi^G(X,A)-\phi^G(X,A))+2(\eta^G(X,A)-\hat\eta^G(X,A))\big)\\
        \le&\E|g^{*G}_\alpha (X,A)|\1\big(0<g^{*G}_\alpha (X,A)\le-\Delta|\phi^G(X,A)|+\hat\epsilon_g,s^G\phi^G(X,A)\ge0\big)\\
        &+\E|g^{*G}_\alpha (X,A)|\1\big(0<g^{*G}_\alpha (X,A)\le\Delta|\phi^G(X,A)|+\hat\epsilon_g,s^G\phi^G(X,A)<0\big)\\
        &+\E|g^{*G}_\alpha (X,A)|\1\big(0>g^{*G}_\alpha (X,A)>-\Delta|\phi^G(X,A)|-\hat\epsilon_g,s^G\phi^G(X,A)\ge0\big)\\
        &+\E|g^{*G}_\alpha (X,A)|\1\big(0>g^{*G}_\alpha (X,A)>\Delta|\phi^G(X,A)|-\hat\epsilon_g,s^G\phi^G(X,A)<0\big)\\
        \le&\E|g^{*G}_\alpha (X,A)|\1\big(0>s_\phi^G(X,A)s^Gg^{*G}_\alpha (X,A)\ge-\Delta|\phi^G(X,A)|-\hat\epsilon_g\big)\\
        &+\E|g^{*G}_\alpha (X,A)|\1\big(0<s_\phi^G(X,A)s^Gg^{*G}_\alpha (X,A)\le\hat\epsilon_g\big)\\
        \le&\E|g^{*G}_\alpha (X,A)|\1\big(0>s_\phi^G(X,A)s^Gg^{*G}_\alpha (X,A)>-2\Delta|\phi^G(X,A)|\big)\\
        &+\E|g^{*G}_\alpha (X,A)|\1\big(0>s_\phi^G(X,A)s^Gg^{*G}_\alpha (X,A)\ge-2\hat\epsilon_g\big)+c\hat\epsilon_g^{1+\gamma}\\
        \overset{\text{Assumption~\ref{ass:ratio_poly}}}{\lesssim}&\E\Delta|\phi^G(X,A)|\1\big(0>s_\phi^G(X,A)s^Gg^{*G}_\alpha (X,A)>-\frac{1}{2}\Delta|\phi^G(X,A)|\big)+\hat\epsilon_g^{1+\gamma}\\
        =&\E\Delta|\phi^G(X,A)|\1\big(0>s_\phi^G(X,A)s^Gg^{*G}_\alpha (X,A)>-\frac{1}{2}\Delta|\phi^G(X,A)|,\Delta|\phi^G(X,A)|\le2\hat\epsilon_g\big)\\
        &+\E\Delta|\phi^G(X,A)|\1\big(0>s_\phi^G(X,A)s^Gg^{*G}_\alpha (X,A)>-\frac{1}{2}\Delta|\phi^G(X,A)|,\Delta|\phi^G(X,A)|>2\hat\epsilon_g\big)+\hat\epsilon_g^{1+\gamma}\\
        \le&\E\Delta|\phi^G(X,A)|\1\big(0>s_\phi^G(X,A)s^Gg^{*G}_\alpha (X,A)>-\frac{1}{2}\Delta|\phi^G(X,A)|,\Delta|\phi^G(X,A)|>2\hat\epsilon_g\big)+c\hat\epsilon_g^{1+\gamma}\\
        \overset{\text{Equation~\eqref{eq:Delta_lambda_lower}}}{\le}&\E\Delta|\phi^G(X,A)|\1\big(0<s_\phi^G(X,A)s^Gg^{*G}_\alpha (X,A)\le\hat\epsilon_g,\Delta|\phi^G(X,A)|<\hat\epsilon_g\big)+c\hat\epsilon_g^{1+\gamma}\\
        \lesssim&\hat\epsilon_g^{1+\gamma}\\
        \le&\epsilon_g^{1+\gamma}.
    \end{aligned}
    \end{equation}
    Then the excess risk can be upper bounded as
    \[\risk(\hat f^G_\alpha)-\risk(f^{*G}_\alpha )\lesssim \abs{\lambda^{*G}_\alpha }\epsilon_\alpha+\epsilon_\eta^{1+\gamma}+\big(|\lambda^{*G}_\alpha|\epsilon_\phi\big)^{1+\gamma}.\]

\end{proof}

\section{Derivation of Equation~\eqref{eq:lambda_upper_bound_eoo_aware}}\label{sec:supp_eq_lambda_upper_bound_eoo_aware}
When $\lambda^*_\alpha=0$, we have $g^*_\alpha(x,a)=2\eta(x,a)-1$. If $\lambda^*_\alpha\ne 0$, we denote $s=\sgn(\lambda^*_\alpha)$, then Equation \eqref{eq:lambda_binary_optimality} implies
\begin{align*}
    \alpha=&-\E\frac{(2A-3)s}{p_{1,A}}\eta(X,A)\1\bigg(\bigg(2+\frac{(2A-3)\lambda^*_\alpha}{p_{1,A}}\bigg)\eta(X,A)>1\bigg)\\
    =&\E\frac{1}{p_{1,A}}\eta(X,A)\1\bigg(\bigg(2-\frac{|\lambda^*_\alpha|}{p_{1,A}}\bigg)\eta(X,A)>1,(2A-3)s<0\bigg)\\
    &-\E\frac{1}{p_{1,A}}\eta(X,A)\1\bigg(\bigg(2+\frac{|\lambda^*_\alpha|}{p_{1,A}}\bigg)\eta(X,A)>1,(2A-3)s>0\bigg),
\end{align*}
it follows that
\[2-\frac{|\lambda^*_\alpha|}{p_{1,\frac{3-s}{2}}}\ge 1.\]
So we have
\[|\lambda^*_\alpha|\le p_{1,\frac{3-s}{2}},\quad\min_{a\in[2]}\bigg\{2+\frac{(2a-3)\lambda^*_\alpha}{p_{1,a}}\bigg\}\ge 1,\]
and $f^*_\alpha$ can be equivalently expressed as a group-wise thresholding rule
\[f^*_\alpha(x,a)=\1\bigg(\eta(x,a)>\bigg(2+\frac{(2a-3)\lambda^*_\alpha}{p_{1,a}}\bigg)^{-1}\bigg).\] 

\section{Local Polynomial Regressions in Section \ref{sec:binary_eoo}}

\subsection{Group-Aware Scenario}
Under Assumptions~\ref{ass:holder_aware}, \ref{ass:density_aware} and \ref{ass:observe}, we can apply local polynomial regression \citep{Tsybakov2009introduction,fan2018local} to estimate $\eta(\cdot, a)$. Denote $t=(t_j)_{j\in [d]}\in\N^d$, $|t|=\sum_{j\in[d]}t_j$. For $x=(x_j)_{j\in[d]}\in\R^d$, we denote $x^t=\prod_{j\in[d]}x_j^{t_j}$ and denote $V_Y(\cdot):\R^d\rightarrow\R^{\lfloor\beta_Y\rfloor+d\choose d}$ to be a vector-valued function indexed by $t$ with $|t|\le\lfloor\beta_Y\rfloor$ and satisfies $\big(V_Y(x)\big)_t=x^t$. For $h_Y>0,x\in[0,1]^d$ and a kernel $\KK:\R^d\rightarrow\R_+$, denote $\hat\theta_{Y,a}(x)\in\R^{\lfloor\beta_Y\rfloor+d\choose d},a\in[2]$ to be
\[\hat\theta_{Y,a}(x)=\argmin_{\theta\in\R^{\lfloor\beta_Y\rfloor+d\choose d}}\sum_{i\in[\tilde n]}\bigg(\tilde Y_i-V_Y^\top\bigg(\frac{\tilde X_i-x}{h_Y}\bigg)\theta\bigg)^2\KK\bigg(\frac{\tilde X_i-x}{h_Y}\bigg)\1(\tilde A_i=a),\]
then the local polynomial estimators $\hat\eta(\cdot,a)$ are
\[\hat\eta(\cdot,a)=V_Y^\top(0)\hat\theta_{Y,a}(x).\]
If we choose the kernel $\KK$ such that $\KK\in\HH(1,L_K)$ and there exist constants $k_l,k_u>0$ such that $k_l\1(\|x\|_2\le k_l)\le\KK(x)\le k_u\1(\|x\|_2\le 1)$ for any $x\in\R^d$, then we can control the estimation error of the local polynomial estimators as follows. The proof of Lemma~\ref{lem:local_polynomial_aware} is similar to that of Theorem 3.2 in \cite{audibert2007fast} and Theorem 1.8 in \cite{Tsybakov2009introduction}, so it is omitted.

\begin{Lemma}[Initial Estimators]\label{lem:local_polynomial_aware}
    Choose $h_Y\asymp\big(\frac{d\log \tilde n+\log\frac{1}{\delta_{\rm init}}}{\tilde n}\big)^{\frac{1}{2\beta_Y+d}}$. Under Assumptions~\ref{ass:holder_aware}, \ref{ass:density_aware} and \ref{ass:observe}, with probability at least $1-\frac{\delta_{\rm init}}{2}$, we have
    \[\max_{a\in[2]}\|\hat\eta(\cdot,a)-\eta(\cdot,a)\|_\infty\lesssim\bigg(\frac{d\log \tilde n+\log\frac{1}{\delta_{\rm init}}}{\tilde n}\bigg)^{\frac{\beta_Y}{2\beta_Y+d}}.\]
\end{Lemma}

\subsection{Group-Blind Scenario}

We use local polynomial regression to estimate $\eta$ and $\rho_{1|1}$. Recall that $V_Y(\cdot):\R^d\rightarrow\R^{\lfloor\beta_Y\rfloor+d\choose d}$ is a vector-valued function indexed by $t$ with $|t|\le\lfloor\beta_Y\rfloor$ and satisfies $\big(V_Y(x)\big)_t=x^t$. 
Similarly, suppose $V_A(\cdot):\R^d\rightarrow\R^{\lfloor\beta_A\rfloor+d\choose d}$ is indexed by $t$ with $|t|\le\lfloor\beta_A\rfloor$ and satisfies $\big(V_A(x)\big)_t=x^t$. 
For $h_Y, h_A>0$, $x\in[0,1]^d$ and the same kernel $\KK$ as in Section~\ref{sec:binary_eoo_aware}, denote $\hat\theta_Y(x)\in\R^{\lfloor\beta_Y\rfloor+d\choose d}$ and $\hat\theta_A(x)\in\R^{\lfloor\beta_A\rfloor+d\choose d}$ to be
\[\hat\theta_Y(x)=\argmin_{\theta\in\R^{\lfloor\beta_Y\rfloor+d\choose d}}\sum_{i\in[\tilde n]}\bigg(\tilde Y_i-V_Y^\top\bigg(\frac{\tilde X_i-x}{h_Y}\bigg)\theta\bigg)^2\KK\bigg(\frac{\tilde X_i-x}{h_Y}\bigg),\]
\[\hat\theta_A(x)=\argmin_{\theta\in\R^{\lfloor\beta_A\rfloor+d\choose d}}\sum_{\tilde Y_i=1,i\in[\tilde n]}\bigg(2-\tilde A_i-V_A^\top\bigg(\frac{\tilde X_i-x}{h_A}\bigg)\theta\bigg)^2\KK\bigg(\frac{\tilde X_i-x}{h_A}\bigg),\]
then the local polynomial estimators are 
\[\hat\eta(x)=V_Y^\top(0)\hat\theta_Y(x), \quad\hat\rho_{1|1}(x)=V_A^\top(0)\hat\theta_A(x).\]
Denote $n_{1,a}=\sum_{i\in[n]}\1(Y_i=1,A_i=a)$, $\tilde n_Y=\sum_{i\in[\tilde n]}\1(\tilde Y_i=1)$, $\tilde n_{1,a}=\sum_{i\in[\tilde n]}\1(\tilde Y_i=1,\tilde A_i=a)$, $a\in[2]$. Then we can control the estimation error of the local polynomial estimators as follows. The proof of Lemma~\ref{lem:local_polynomial_blind} is similar to that of Lemma~\ref{lem:local_polynomial_aware}, so is omitted.

\begin{Lemma}[Initial Estimators]\label{lem:local_polynomial_blind}
    Choose $h_Y\asymp\big(\frac{d\log \tilde n+\log\frac{1}{\delta_{\rm init}}}{\tilde n}\big)^{\frac{1}{2\beta_Y+d}}$, $h_A\asymp\big(\frac{d\log \tilde n+\log\frac{1}{\delta_{\rm init}}}{\tilde n}\big)^{\frac{1}{2\beta_A+d}}$. Under Assumptions~\ref{ass:observe}, \ref{ass:holder_blind} and \ref{ass:density_blind}, with probability at least $1-\frac{\delta_{\rm init}}{2}$, we have
    \[\|\hat\eta-\eta\|_\infty\lesssim\bigg(\frac{d\log \tilde n+\log\frac{1}{\delta_{\rm init}}}{\tilde n}\bigg)^{\frac{\beta_Y}{2\beta_Y+d}},\quad\|\hat\rho_{1|1}-\rho_{1|1}\|_\infty\lesssim\bigg(\frac{d\log \tilde n+\log\frac{1}{\delta_{\rm init}}}{\tilde n}\bigg)^{\frac{\beta_A}{2\beta_A+d}}.\]
\end{Lemma}

\section{Minimax Expected Excess Risk Lower Bound}

\begin{Theorem}[Minimax Expected Excess Risk Lower Bound]\label{thm:lower_expected_binary_eoo}
    Under the assumptions in Theorem~\ref{thm:lower_binary_eoo}, for any $\alpha>0$, there exist models $P$ such that $|\lambda^{*{\rm blind}}_{\alpha,P}|\asymp \alpha^{-1}$, and we have the minimax lower bounds for the expected excess risks, 
    \begin{equation}\label{eq:lower_aware_expected}
        \inf_{\A^{\rm aware}\in\mathscr{A}^{\rm aware}}\sup_{P\in\mathscr{P}}\bigg\{\E_{\D_{\rm all}\sim P^{\otimes N}}\risk_P\big(\A^{\rm aware}(\D_{\rm all})\big)-\risk_P(f^{*{\rm aware}}_{\alpha,P})\bigg\}\gtrsim N^{-\frac{\beta_Y(1+\gamma)}{2\beta_Y+d}}(c-\delta),
    \end{equation}
    if $\alpha\lesssim N^{-\frac{\beta_Y\gamma}{(2\beta_Y+d)(1+\gamma)}}$, then
    \begin{equation}\label{eq:lower_blind_expected}
        \begin{aligned}
            &\inf_{\A^{\rm blind}\in\mathscr{A}^{\rm blind}}\sup_{P\in\mathscr{P}}\bigg\{\E_{\D_{\rm all}\sim P^{\otimes N}}\risk_P\big(\A^{\rm blind}(\D_{\rm all})\big)-\risk_P(f^{*{\rm blind}}_{\alpha,P})\bigg\}\\
            \gtrsim& \bigg[\bigg\{\alpha^{-1}N^{-\frac{1}{2}}+\bigg(\alpha^{-1}N^{-\frac{\beta_A}{2\beta_A+d}}\bigg)^{1+\gamma}+N^{-\frac{\beta_Y(1+\gamma)}{2\beta_Y+d}}\bigg\}\wedge 1\bigg](c-\delta),
        \end{aligned}
    \end{equation}
    if $\alpha\gtrsim N^{-\frac{\beta_Y\gamma}{(2\beta_Y+d)(1+\gamma)}}$, then
    \begin{equation}\label{eq:lower_blind_expected_all}
        \begin{aligned}
            &\inf_{\A^{\rm blind}\in\mathscr{A}^{\rm blind}}\sup_{P\in\mathscr{P}}\bigg\{\E_{\D_{\rm all}\sim P^{\otimes N}}\risk_P\big(\A^{\rm blind}(\D_{\rm all})\big)-\risk_P(f^{*{\rm blind}}_{\alpha,P})\bigg\}\\
            \gtrsim& \bigg[\bigg\{\alpha^{-1}N^{-\frac{1}{2}}+\bigg(\alpha^{-1}N^{-\frac{\beta_A}{2\beta_A+d}}\bigg)^{1+\gamma}+\big(\alpha^{-1}N^{-\frac{\beta_Y}{2\beta_Y+d}}\big)^{1+\gamma}\bigg\}\wedge 1\bigg](c-\delta).
        \end{aligned}
    \end{equation}
\end{Theorem}
\cite{xian2024unified} also identified an $O(\alpha^{-1})$ dependence in the excess risk analysis for their algorithm. While they showed that the $O(\alpha^{-1})$ dependence is necessary for their proposed algorithm, our result reveals that it is inevitable for any algorithm.

\begin{Remark}
The upper bounds for expected excess risks follow directly from Equations~\eqref{eq:excess_risk_upper_binary_eoo_aware} and \eqref{eq:excess_risk_upper_binary_eoo_blind} by choosing proper $\delta$. Similar to Remark~\ref{rem:minimax}, recall from Remark~\ref{rem:lambda_upper_unify} that $|\lambda^{*{\rm blind}}_\alpha|\le\alpha^{-1}$, then we know the expected group-blind excess risk of Algorithm~\ref{alg:binary_unified} is minimax optimal up to logarithmic factors. Since Equation~\eqref{eq:lambda_upper_bound_eoo_aware} implies $|\lambda^{*{\rm aware}}_\alpha|\le 1$, when $2\beta_Y\gamma\le d$, the expected group-aware excess risk of Algorithm~\ref{alg:binary_unified} is also minimax optimal up to logarithmic factors.
\end{Remark}

\begin{Remark}[\textbf{Cost of Group-blindness}]\label{rem:cost_blind}
    By comparing the group-aware excess risk upper bound~\eqref{eq:excess_risk_upper_binary_eoo_aware} to the group-blind lower bound~\eqref{eq:lower_blind}, we observe two sources of cost of group-blindness: 
    
    On the one hand, the group-blind lower bound~\eqref{eq:lower_blind} contains an extra term $O_P(|\lambda^{*{\rm blind}}_\alpha|^{1+\gamma}N^{-\frac{\beta_A(1+\gamma)}{2\beta_A+d}})$. Recall that $|\lambda^{*{\rm blind}}_\alpha|$ is the magnitude of translation from $\1(2\eta^{\rm blind}>1)$ to $f^{*{\rm blind}}_\alpha$, and $O_P(N^{-\frac{\beta_A}{2\beta_A+d}})$ is the error of estimating the prediction function $\rho_{1|1}$ of $A$ given $X$ and $Y=1$.

    On the other hand, as we have argued in Equation~\eqref{eq:lambda_upper_bound_eoo_aware} and Theorem~\ref{thm:lower_expected_binary_eoo}, the group-aware $|\lambda^{*{\rm aware}}_\alpha|$ is always less than 1 but the group-blind $|\lambda^{*{\rm blind}}_\alpha|$ can be as large as $O(\alpha^{-1})$. The latter happens when $(X,Y)$ contains little information about $A$. Specifically, recall from the discussion of Example~\ref{exa:phi_eoo} that $|\phi^{\rm blind}|$ roughly characterizes the confidence of predicting $A$ given $X$ and $Y=1$, i.e., the amount of information of $A$ contained in $X$ and $Y$. When predicting $A$ is relatively hard such that $|\phi^{\rm blind}|\asymp \alpha$, suppose, for example, $|\E\phi^{\rm blind}(X)\1(2\eta^{\rm blind}(X)>1)|=2\alpha$, then it may require $\lambda^{*{\rm blind}}_\alpha\asymp\alpha^{-1}$ to adjust $\1(2\eta^{\rm blind}>1)$ such that $|\E\phi^{\rm blind}(X)\1(2\eta^{\rm blind}(X)-1>\lambda^{*{\rm blind}}_\alpha\phi^{\rm blind}(X))|=\alpha$. In that case, the group-blind lower bound \eqref{eq:lower_blind_expected} becomes a constant when $\alpha\lesssim N^{-\frac{\beta_A}{2\beta_A+d}}$, making the group-blind excess risk larger than the group-aware one due to the larger $|\lambda^{*{\rm blind}}_\alpha|$. Our rate provides an exact quantification of how the cost of group-blindness depends on the difficulty of predicting the sensitive attribute $A$ using $X$.
\end{Remark}

\begin{Remark}[\textbf{Optimal Trade-off Between Excess Risk and Fairness }]\label{rem:tradeoff_blind}
    The optimal expected group-blind excess risk~\eqref{eq:lower_blind_expected} and \eqref{eq:lower_blind_expected_all} are decreasing in $\alpha$, therefore we reveal the trade-off between algorithmic fairness and group-blind excess risk.

    Intuitively, as $\alpha$ decreases, fewer classifiers remain $\alpha$-fair, one might expect easier identification of the Bayes optimal $\alpha$-fair classifier, which results in a smaller excess risk. 
    However, surprisingly, decreasing $\alpha$ leads to an increase in the optimal group-blind excess risk~\eqref{eq:lower_blind_expected} and \eqref{eq:lower_blind_expected_all}. 
    To explain this counterintuitive phenomenon, we decompose the excess risk as follows,
    \begin{equation}\label{eq:excess_risk_decomposition_blind}
    \begin{aligned}
    &\risk(\hat f_\alpha^{\rm blind})-\risk(f^{*{\rm blind}}_\alpha)\\
    =&\E_X\big(2\eta^{\rm blind}(X)-1\big)\big(f^{*{\rm blind}}_\alpha (X)-\hat f_\alpha^{\rm blind}(X)\big)\\
    =&\underbrace{\E_X|2\eta^{\rm blind}(X)-1-\lambda^{*{\rm blind}}_\alpha\phi^{\rm blind}(X)||f^{*{\rm blind}}_\alpha (X)-\hat f_\alpha^{\rm blind}(X)|}_{T_1}\\
    &+\underbrace{\lambda^{*{\rm blind}}_\alpha \E_X\phi^{\rm blind}(X)\big(f^{*{\rm blind}}_\alpha (X)-\hat f^{\rm blind}_\alpha(X)\big)}_{T_2}.
    \end{aligned}
    \end{equation}
    For $T_2$, due to fairness constraints, we know
    \begin{align*}
    T_2=|\lambda^{*{\rm blind}}_\alpha |\alpha-\lambda^{*{\rm blind}}_\alpha \E_X\phi^{\rm blind}(X)\hat f^{\rm blind}_\alpha(X)
    \ge |\lambda^{*{\rm blind}}_\alpha |\big(\alpha-\U(\hat f^{\rm blind}_\alpha)\big)
    \ge 0.
    \end{align*}
    Note that $|\lambda^{*{\rm blind}}_\alpha|$ can be as large as $O(\alpha^{-1})$, which is decreasing in $\alpha$. 
    On the one hand, since $T_2$ has a multiplicative dependence on $\lambda^{*{\rm blind}}_\alpha$, a decrease in $\alpha$ amplifies $T_2$. As a result, the excess risk itself as a function of $\hat f^{\rm blind}_\alpha$ is potentially decreasing in $\alpha$.
    On the other hand, although the function classes for $\eta^{\rm blind}$ and $\phi^{\rm blind}$ are fixed, the function class for $g^{*{\rm blind}}_\alpha=2\eta^{\rm blind}-1-\lambda^{*{\rm blind}}_\alpha\phi^{\rm blind}$ expands as $\alpha$ decreases. To see this, note that $\phi^{\rm blind}$ is $(\beta_Y,L)$-H\"older smooth for some smoothness coefficient $L$. As $|\lambda^{*{\rm blind}}_\alpha|$ increases, the smoothness coefficient for the function class of $g^{*{\rm blind}}_\alpha$ also increases, leading to a larger minimax lower bound. This occurs through the following mechanism. When constructing minimax lower bounds for $T_1$, we add bumps to $g^{*{\rm blind}}_\alpha$ around the classification boundary $g^{*{\rm blind}}_\alpha=0$. The increasing smoothness coefficient for the function class of $g^{*{\rm blind}}_\alpha$ allows larger bumps of $g^{*{\rm blind}}_\alpha$. Note that the margin assumption constrains the number of bumps around the classification boundary relative to the magnitude of each bump, then larger bumps allow for a greater number of them. This results in a more fluctuant $g^{*{\rm blind}}_\alpha$ around the classification boundary, making $f^{*{\rm blind}}_\alpha$ harder to estimate and consequently leading to an increase in $T_1$.
\end{Remark}

\begin{Remark}[\textbf{Proof Sketch of Theorems~\ref{thm:lower_binary_eoo} and \ref{thm:lower_expected_binary_eoo}}]
    The proof of the lower bound~\eqref{eq:lower_blind}, \eqref{eq:lower_blind_expected} and \eqref{eq:lower_blind_expected_all} is highly nontrivial. The analysis of the excess risk is based on the decomposition~\eqref{eq:excess_risk_decomposition_blind}.
Term $T_1$ can be controlled based on a similar strategy to \cite{audibert2007fast,rigollet2009optimal} using Fano's lemma and the margin assumption~\ref{ass:margin}. However, unlike $T_1$, $T_2$ cannot be bounded by a distance $d(\hat f^{\rm blind}_\alpha,f^{*{\rm blind}}_\alpha )$ from below directly and the triangle inequality fails to hold, so standard tools for proving minimax lower bounds do not apply here. In order to show $T_2$ has minimax lower bound $O_P\big(|\lambda^{*{\rm blind}}_\alpha |(N^{-\frac{1}{2}}\wedge\alpha)\big)$, it suffices to prove that for any algorithm $\A^{\rm blind}\in\mathscr{A}^{\rm blind}$, there exists a distribution $P\in\mathscr{P}$ such that 
\begin{equation}\label{eq:minimax_fhat_unfairness_eoo}
    \Prob_{\D_{\rm all}\sim P^{\otimes N}}\bigg(\lambda^{*{\rm blind}}_{\alpha,P}\E_{X}\phi_P^{\rm blind}(X)\A^{\rm blind}(\D_{\rm all})(X)\le |\lambda^{*{\rm blind}}_{\alpha,P}|\big(\alpha-c(N^{-\frac{1}{2}}\wedge \alpha)\big)\bigg)\ge c-\delta.
\end{equation}
Recall that when triangle inequalities hold, standard methods reduce the lower bound of the algorithm-dependent risk to the testing problem over a set of algorithm-independent distributions that are close in distribution but far away in terms of the risks. However, the triangle inequalities fail to hold in \eqref{eq:minimax_fhat_unfairness_eoo}, so new techniques are required to, either design a set of algorithm-dependent worst-case distributions, or eliminate the impact of specific algorithms. Here we take the second strategy and construct a specific pair of algorithm-independent distributions $P,\bar P\in\mathscr{P}$ that are close in distributions, i.e., ${\rm TV}(P^{\otimes N},\bar P^{\otimes N})\le \tilde c$, but are far away in terms of the unfairness measures simultaneously for all group-blind classifiers $f\in[0,1]^{\X}$, i.e., $\U_{\rm EOO,\bar P}(f)=\U_{\rm EOO,P}(f)\{1-c(\frac{1}{\alpha\sqrt{N}}\wedge 1)\}$. These two properties allow us to eliminate the impact of specific algorithms. Then for any classifier $f$, the fairness constraint under $P$, i.e, $\U_{\rm EOO,P}(f)\le\alpha$, implies $\U_{\rm EOO,\bar P}(f)\le \alpha-c(N^{-\frac{1}{2}}\wedge \alpha)$, hence showing that \eqref{eq:minimax_fhat_unfairness_eoo} is satisfied under $\bar P$.

\end{Remark}

\section{Proofs of Theorem \ref{thm:lower_binary_eoo} and \ref{thm:lower_expected_binary_eoo}}

Since the proofs of Theorem~\ref{thm:lower_binary_eoo} and \ref{thm:lower_expected_binary_eoo} are almost the same, we put them together in the following.

\begin{proof}[Proof of Theorem \ref{thm:lower_binary_eoo} and \ref{thm:lower_expected_binary_eoo}]
    In this proof, we only consider the group-blind scenario. The group-aware lower bound is the same with the lower bound for the unconstrained classification problem. By setting $A$ to be independent of $(X,Y)$, the fairness constrained classification problem reduces to unconstrained classification, then we can conclude the group-aware lower bound similar to the proof of the lower bound in \cite{audibert2007fast} and the group-blind lower bound below.

    Now we prove the group-blind lower bound. For any classifier $\hat f$, by Proposition 1 in \cite{tsybakov2004optimal}, the excess risk of $\hat f$ can be lower bounded as
\begin{align*}
    &\risk(\hat f)-\risk(f^*_\alpha)\\
    =&\E(2\eta(X)-1)(f^*_\alpha (X)-\hat f(X))\\
    =&\E|2\eta(X)-1-\lambda^*_\alpha \phi(X)||f^*_\alpha (X)-\hat f(X)|+\lambda^*_\alpha \E\phi(X)(f^*_\alpha (X)-\hat f(X))\\
    =&\E |g^*_\alpha(X)||f^*_\alpha (X)-\hat f(X)|+\lambda^*_\alpha \E\phi(X)(f^*_\alpha (X)-\hat f(X))\\
    \ge&\underbrace{c\big(\E|f^*_\alpha (X)-\hat f(X)|\1(g^*_\alpha(X)\ne 0)\big)^{\frac{1+\gamma}{\gamma}}}_{T_1(\hat f)}+\underbrace{\lambda^*_\alpha \E\phi(X)(f^*_\alpha (X)-\hat f(X))}_{T_2(\hat f)}.
\end{align*}
While $T_1$ corresponds to the error for estimating the classifier $f^*_\alpha$, $T_2$ is due to the unfairness difference. If $\hat f$ is $\alpha$-fair, we know 
\[T_2(\hat f)=|\lambda^*_\alpha|\alpha-\lambda^*_\alpha\E\phi(X)\hat f(X)\ge|\lambda^*_\alpha|\big(\alpha-|\E\phi(X)\hat f(X)|\big)\ge 0.\]
Therefore, the excess risk is nonnegative for $\alpha$-fair classifiers $\hat f$. However, when $\hat f$ violates the $\alpha$-fair constraint, $T_2(\hat f)$ becomes negative, which may lead to a negative excess risk.

Then, for any $\epsilon$, we have the following reduction for the minimax lower bound of excess risk,
\begin{equation}\label{eq:minimax_reduction_binary_eoo}
    \begin{aligned}
    &\inf_{\A\in\mathscr{A}^{\rm blind}}\sup_{P\in\mathscr{P}}\Prob_{\D_{\rm all}\sim P^{\otimes N}}\big(\risk_P(\A(\D_{\rm all}))-\risk_P(f^*_{\alpha,P})\ge\epsilon\big)\\
    \ge&\inf_{\A\in\mathscr{A}^{\rm blind}}\sup_{P\in\mathscr{P}}\Prob_{\D_{\rm all}\sim P^{\otimes N}}\big(T_1(\A(\D_{\rm all}))+T_2(\A(\D_{\rm all}))\ge\epsilon\big)\\
    \ge&\inf_{\A\in\mathscr{A}^{\rm blind}}\sup_{P\in\mathscr{P}}\Prob_{\D_{\rm all}\sim P^{\otimes N}}\big(T_1(\A(\D_{\rm all}))+T_2(\A(\D_{\rm all}))\ge\epsilon,T_2(\A(\D_{\rm all}))\ge 0\big)\\
    \ge&\inf_{\A\in\mathscr{A}^{\rm blind}}\sup_{P\in\mathscr{P}}\Prob_{\D_{\rm all}\sim P^{\otimes N}}\big(T_1(\A(\D_{\rm all}))\ge\epsilon\big)-\Prob_{\D_{\rm all}\sim P^{\otimes N}}\big(T_2(\A(\D_{\rm all}))<0\big)\\
    \ge&\inf_{\A\in\mathscr{A}^{\rm blind}}\sup_{P\in\mathscr{P}}\Prob_{\D_{\rm all}\sim P^{\otimes N}}\big(T_1(\A(\D_{\rm all}))\ge\epsilon\big)-\Prob_{\D_{\rm all}\sim P^{\otimes N}}\big(\U_{\rm EOO,P}(\A(\D_{\rm all}))>\alpha\big)\\
    \ge&\inf_{\A\in\mathscr{A}^{\rm blind}}\sup_{P\in\mathscr{P}}\Prob_{\D_{\rm all}\sim P^{\otimes N}}\big(T_1(\A(\D_{\rm all}))\ge\epsilon\big)-\delta.
\end{aligned}
\end{equation}
And we also have
\begin{align*}
    \inf_{\A\in\mathscr{A}^{\rm blind}}\sup_{P\in\mathscr{P}}\Prob_{\D_{\rm all}\sim P^{\otimes N}}\big(\risk_P(\A(\D_{\rm all}))-\risk_P(f^*_{\alpha,P})\ge\epsilon\big)
    \ge\inf_{\A\in\mathscr{A}^{\rm blind}}\sup_{P\in\mathscr{P}}\Prob_{\D_{\rm all}\sim P^{\otimes N}}\big(T_2(\A(\D_{\rm all}))\ge\epsilon\big).
\end{align*}

Based on the reduction above, we can make the following claims.
\begin{Claim}\label{clm:rho}
    There exist a subclass of distributions $\tilde{\mathscr{P}}\subset\mathscr{P}$ and a constant $c\in(0,1)$, such that 
    \[|\lambda_{\alpha,P}^*|\asymp\alpha^{-1}\wedge N^{\frac{\beta_A}{2\beta_A+d}},\quad\forall P\in\tilde{\mathscr{P}},\]
    and the minimax excess risk lower bound due to the estimation error of $\rho_{1|1}$ is
    \[\inf_{\A\in\mathscr{A}^{\rm blind}}\sup_{P\in\tilde{\mathscr{P}}}\Prob_{\D_{\rm all}\sim P^{\otimes N}}\bigg(T_1\big(\A(\D_{\rm all})\big)\gtrsim\big(\alpha^{-1}N^{-\frac{\beta_A}{2\beta_A+d}}\big)^{1+\gamma}\wedge 1\bigg)\ge c.\]
\end{Claim}

\begin{Claim}\label{clm:eta}
    \noindent
    \begin{enumerate}[1)]
        \item There exist a subclass of distributions $\tilde{\mathscr{P}}\subset\mathscr{P}$ and a constant $c\in(0,1)$, such that 
    \[|\lambda_{\alpha,P}^*|=0,\quad\forall P\in\tilde{\mathscr{P}},\]
    and the minimax excess risk lower bound due to the estimation error of $\eta$ is 
    \[\inf_{\A\in\mathscr{A}^{\rm blind}}\sup_{P\in\tilde{\mathscr{P}}}\Prob_{\D_{\rm all}\sim P^{\otimes N}}\bigg(T_1\big(\A(\D_{\rm all})\big)\gtrsim N^{-\frac{\beta_Y(1+\gamma)}{2\beta_Y+d}}\bigg)\ge c.\]
    
        \item If $\alpha\gtrsim N^{-\frac{\beta_Y\gamma}{(2\beta_Y+d)(1+\gamma)}}$, there exist a subclass of distributions $\tilde{\mathscr{P}}\subset\mathscr{P}$ and a constant $c\in(0,1)$, such that 
    \[|\lambda_{\alpha,P}^*|\asymp\alpha^{-1},\quad\forall P\in\tilde{\mathscr{P}},\]
    and the minimax excess risk lower bound due to the estimation error of $\eta$ is 
    \[\inf_{\A\in\mathscr{A}^{\rm blind}}\sup_{P\in\tilde{\mathscr{P}}}\Prob_{\D_{\rm all}\sim P^{\otimes N}}\bigg(T_1\big(\A(\D_{\rm all})\big)\gtrsim\big(\alpha^{-1}N^{-\frac{\beta_Y}{2\beta_Y+d}}\big)^{1+\gamma}\bigg)\ge c.\]
    \item If $\alpha\lesssim N^{-\frac{\beta_Y\gamma}{(2\beta_Y+d)(1+\gamma)}}$, there exist a subclass of distributions $\tilde{\mathscr{P}}\subset\mathscr{P}$ and a constant $c\in(0,1)$, such that 
    \[|\lambda_{\alpha,P}^*|\asymp 1,\quad\forall P\in\tilde{\mathscr{P}},\]
    and the minimax excess risk lower bound due to the estimation error of $\eta$ is 
    \[\inf_{\A\in\mathscr{A}^{\rm blind}}\sup_{P\in\tilde{\mathscr{P}}}\Prob_{\D_{\rm all}\sim P^{\otimes N}}\bigg(T_1\big(\A(\D_{\rm all})\big)\gtrsim N^{-\frac{\beta_Y(1+\gamma)}{2\beta_Y+d}}\bigg)\ge c.\]
    \end{enumerate}

\end{Claim}

\begin{Claim}\label{clm:fair}
    There exist a subclass of distributions $\tilde{\mathscr{P}}\subset\mathscr{P}$ and a constant $c\in(0,1)$, such that 
    \[|\lambda_{\alpha,P}^*|\asymp\alpha^{-1},\quad\forall P\in\tilde{\mathscr{P}},\]
    and the minimax excess risk lower bound due to the conservativeness in fairness control is
    \[\inf_{\A\in\mathscr{A}^{\rm blind}}\sup_{P\in\tilde{\mathscr{P}}}\Prob_{\D_{\rm all}\sim P^{\otimes N}}\bigg(T_2\big(\A(\D_{\rm all})\big)\gtrsim\big(\alpha^{-1}N^{-\frac{1}{2}}\wedge 1\big)\bigg)\ge c-\delta.\]
\end{Claim}

Combining the reduction \eqref{eq:minimax_reduction_binary_eoo} and Claims \ref{clm:rho}, \ref{clm:eta}, and \ref{clm:fair}, we conclude the minimax excess risk lower bound,
\begin{align*}
    &\inf_{\A\in\mathscr{A}^{\rm blind}}\sup_{P\in\mathscr{P}}\Prob_{\D_{\rm all}\sim P^{\otimes N}}\bigg(\risk_P(\A(\D_{\rm all}))-\risk_P(f^*_{\alpha,P})\gtrsim\\
    &|\lambda^*_{\alpha,P}|(N^{-\frac{1}{2}}\wedge\alpha)+\bigg(|\lambda^*_{\alpha,P}|N^{-\frac{\beta_A}{2\beta_A+d}}\bigg)^{1+\gamma}+\bigg(\big(1+|\lambda^*_{\alpha,P}|\big)N^{-\frac{\beta_Y}{2\beta_Y+d}}\bigg)^{1+\gamma}\bigg)\ge c-\delta.
\end{align*}
By Markov's inequality, we know
\[\Prob_{\D_{\rm all}\sim P^{\otimes N}}\bigg(\risk_P(\A(\D_{\rm all}))-\risk_P(f^*_{\alpha,P})\ge \epsilon\bigg)\le\frac{\E_{\D_{\rm all}\sim P^{\otimes N}}\risk_P\big(\A(\D_{\rm all})\big)-\risk_P(f^{*}_{\alpha,P})}{\epsilon},\]
then we obtain the following minimax expected excess risk lower bounds.
If $\alpha\lesssim N^{-\frac{\beta_Y\gamma}{(2\beta_Y+d)(1+\gamma)}}$, then
        \begin{align*}
            &\inf_{\A^{\rm blind}\in\mathscr{A}^{\rm blind}}\sup_{P\in\mathscr{P}}\bigg\{\E_{\D_{\rm all}\sim P^{\otimes N}}\risk_P\big(\A^{\rm blind}(\D_{\rm all})\big)-\risk_P(f^{*{\rm blind}}_{\alpha,P})\bigg\}\\
            \gtrsim& \bigg[\bigg\{\alpha^{-1}N^{-\frac{1}{2}}+\bigg(\alpha^{-1}N^{-\frac{\beta_A}{2\beta_A+d}}\bigg)^{1+\gamma}+N^{-\frac{\beta_Y(1+\gamma)}{2\beta_Y+d}}\bigg\}\wedge 1\bigg](c-\delta),
        \end{align*}
    if $\alpha\gtrsim N^{-\frac{\beta_Y\gamma}{(2\beta_Y+d)(1+\gamma)}}$, then
        \begin{align*}
            &\inf_{\A^{\rm blind}\in\mathscr{A}^{\rm blind}}\sup_{P\in\mathscr{P}}\bigg\{\E_{\D_{\rm all}\sim P^{\otimes N}}\risk_P\big(\A^{\rm blind}(\D_{\rm all})\big)-\risk_P(f^{*{\rm blind}}_{\alpha,P})\bigg\}\\
            \gtrsim& \bigg[\bigg\{\alpha^{-1}N^{-\frac{1}{2}}+\bigg(\alpha^{-1}N^{-\frac{\beta_A}{2\beta_A+d}}\bigg)^{1+\gamma}+\big(\alpha^{-1}N^{-\frac{\beta_Y}{2\beta_Y+d}}\big)^{1+\gamma}\bigg\}\wedge 1\bigg](c-\delta).
        \end{align*}

\end{proof}

Note that $P_{X,A,Y}=P_XP_{Y|X}P_{A|X,Y}$. In the following, we prove Claims \ref{clm:rho}, \ref{clm:eta}, and \ref{clm:fair} separately based on some specified family of $P_{X,A,Y}$.

\subsection{Proof of Claim \ref{clm:rho}}

\begin{proof}[Proof of Claim \ref{clm:rho}]
This proof focuses on constructing the subclass $\tilde{\mathscr{P}}$ of distribution $P_{X,A,Y}=P_XP_{Y|X}P_{A|X,Y}$. Recall $p_X$ is the density of $P_X$, $\eta(X)=\Prob(Y=1|X)$, and $\rho_{a|y}(X)=\Prob(A=a|X,Y=y)$. Then it suffices to specify $p_X$, $\eta$, and $\rho_{a|y}$ respectively.

For some integer $M$, we define the index vector $j$ as $j=(j_1,\ldots,j_d)$ and denote the grids on $[0,1]^d$ as
\[G_M=\bigg\{\frac{2j-1}{14M}:j\in[7M]^d\bigg\},\]
with
\[\frac{2j-1}{14M}=\bigg(\frac{2j_1-1}{14M},\ldots,\frac{2j_d-1}{14M}\bigg).\]
For any $x\in[0,1]^d$, we denote $n_M(x)\in G_M$ to be the closest point to $x$ among $G_M$. Then we can construct a partition of $[0,1]^d$ as $\{\X_j:j\in[7M]^d\}$ with
\[\X_j=\bigg\{x:n_M(x)=\frac{2j-1}{14M}\bigg\}.\]
For some integer $m\le 7^{d-1}M^d$, denote $\II$ to be a set of indexes with $|\II|= m$, $\II\subset([M-2]+M+1)\times[7M]^{d-1}$. Then we define $\X_0=[\frac{1}{7},\frac{2}{7}]\times[0,1]^{d-1}\setminus\cup_{j\in\II}\X_j$.

Let
\[h(z)=\frac{\int_{z}^{\frac{1}{2}}h_1(t)dt}{\int_0^{\frac{1}{2}}h_1(t)dt},\quad h_1(z)=\left\{\begin{matrix}
e^{-\frac{1}{z(1-z)}},& {\rm if}~ z\in[0,1],\\
0,& {\rm otherwise},
\end{matrix}\right.\]
\[u(z)=\frac{\int_z^\infty u_1(t)dt}{\int_{\frac{1}{28}}^{\frac{1}{14}}u_1(t)dt},\quad u_1(z)=\left\{\begin{matrix}
    e^{-\frac{1}{(\frac{1}{14}-z)(z-\frac{1}{28})}},& {\rm if}~z\in[\frac{1}{28},\frac{1}{14}],\\
    0,&{\rm otherwise},
\end{matrix}\right.\]
then both $h$ and $u$ are infinitely differentiable, $h$ takes value 1 on $(-\infty,0]$ and -1 on $[1,\infty)$, $u$ takes value 1 on $[0,\frac{1}{28}]$ and 0 on $[\frac{1}{14},\infty)$. Let
\[\psi(x)=C_{\psi}u(\|x\|_2),\]
where $C_{\psi}$ is taken small enough such that $\psi\in\HH(\beta_A,L_A,\R^d)$.

\textbf{Construction of $\eta$:}

Under the assumption $\frac{d}{\gamma}\ge\beta_A$, we define the regression function $\eta$ as
\[\eta(x)=\left\{\begin{matrix}
    C_\eta-\tilde C_\eta(\frac{1}{7}-x_1)^{\frac{d}{\gamma}},&{\rm if}~ x_1\in [0,\frac{1}{7}],\\
    C_\eta,&{\rm if~} x_1\in[\frac{1}{7},\frac{2}{7}],\\
    C_\eta+\tilde C_\eta(x_1-\frac{2}{7})^{\frac{d}{\gamma}},&{\rm if~} x_1\in [\frac{2}{7},\frac{3}{7}],\\
    \tilde h(x),& {\rm if}~x_1\in[\frac{3}{7},\frac{4}{7}],\\
    \frac{1}{2},&{\rm if~}x_1\in[\frac{4}{7},\frac{5}{7}],\\
    \frac{3}{4}-\frac{1}{2}C_\eta-(\frac{1}{4}-\frac{1}{2}C_\eta)h(7x_1-5),&{\rm if~}x_1\in[\frac{5}{7},\frac{6}{7}],\\
    1-C_\eta,&{\rm if}~ x_1\in[\frac{6}{7},1].
\end{matrix}\right.\]
with $C_\eta,\tilde C_\eta>0$ to be small enough and $\tilde h$ to be a polynomial such that $\eta\in\HH(\beta_Y,L_Y,\R^d)$. Without the loss of generality, we assume the existence of $\tilde h$, otherwise, we can always extend the interval $[\frac{3}{7},\frac{4}{7}]$ to fulfill this.

\textbf{Construction of $\rho_{a|y}$:}

For any $\sigma=(\sigma_j)_{j\in\II}\in\{-1,1\}^m$, we define the regression function $\rho_{1|y}^\sigma$ as
\[\rho_{1|1}^\sigma(x)-\frac{1}{2}=\left\{\begin{matrix}
    -C_\rho,&{\rm if~}x_1\in[0,\frac{1}{7}],\\
    -C_\rho-\sigma_jM^{-\beta_A}\psi(M(x-n_M(x))),&{\rm if ~}x\in\X_j,j\in\II,\\
    -C_\rho,&{\rm if~}x\in\X_0,\\
    -C_\rho,&{\rm if~}x_1\in[\frac{2}{7},\frac{3}{7}],\\
    \frac{1}{4}C_\eta C_\rho-(C_\rho+\frac{1}{4}C_\eta C_\rho) h(7x_1-3),&{\rm if}~x_1\in[\frac{3}{7},\frac{4}{7}],\\
    C_\rho+\frac{1}{2}C_\eta C_\rho,&{\rm if}~x_1\in[\frac{4}{7},\frac{5}{7}],\\
    C_\rho+\frac{1}{4}C_\eta C_\rho h(7x_1-5),&{\rm if~}x_1\in[\frac{5}{7},\frac{6}{7}],\\
    C_\rho,&{\rm if~}x_1\in[\frac{6}{7},1],
\end{matrix}\right.\]
where $C_\rho>0$ is small enough such that $\rho^\sigma_{1|1}\in\HH(\beta_A,L_A,\R^d)$. Then Assumption~\ref{ass:holder_blind} is satisfied. We also define $\rho_{1|0}$ as
\[\rho_{1|0}(x)=\left\{\begin{matrix}
    \frac{1}{4}, &{\rm if~}x_1\in[0,\frac{3}{7}],\\
    \frac{1}{2}+C_\rho+\frac{1}{2}C_\eta C_\rho-\tilde C_\eta(\frac{9}{14}-x_1)^{\frac{d}{\gamma}}, &{\rm if~}x_1\in[\frac{4}{7},\frac{9}{14}],\\
    \frac{1}{2}+C_\rho+\frac{1}{2}C_\eta C_\rho+\tilde C_\eta(x_1-\frac{9}{14})^{\frac{d}{\gamma}}, &{\rm if~}x_1\in[\frac{9}{14},\frac{5}{7}],\\
    \frac{3}{4}, &{\rm if~}x_1\in[\frac{6}{7},1].
\end{matrix}\right.\]
And $\rho_{1|0}$ on $([\frac{3}{7},\frac{4}{7}]\cup[\frac{5}{7},\frac{6}{7}])\times[0,1]^{d-1}$ is defined such that $\rho_{1|0}$ is $\beta_Y$-H\"older smooth.

For this part of the proof, we have $C_\eta,\tilde C_\eta$ to be small constants but $C_\rho$ may become small when $\alpha$ varies.

\textbf{Construction of $p_X$:}

Suppose $\sum_{j\in\II}\1(\sigma_j=1)=C_\sigma m$ for some constant $C_\sigma>0$. Denote $\Delta=C_\psi m\omega M^{-\beta_A}$. Denote the $\ell_p$ ball $B_p(c,r)$ in $\R^d$ as $\{x:\|x-c\|_p\le r,x\in\R^d\}$ and the Lebesgue measure to be ${\rm Leb}(\cdot)$. For some $\omega\in(0,\frac{1}{6m})$, we define the density of $X\in[0,1]^d$ as
\[p_X(x)=\left\{\begin{matrix}
    \frac{\frac{1}{6}-m\omega}{{\rm Leb}(B_1(0,\frac{1}{14}))},&{\rm if~}x\in B_1(\frac{e_1}{14},\frac{1}{14}),\\
    \frac{2\omega}{{\rm Leb}(B_2(0,\frac{1}{28M}))},&{\rm if} ~x\in B_2(\frac{2j-1}{14M},\frac{1}{28M}), j\in\II,\\
    \frac{\frac{1}{6}-m\omega}{{\rm Leb}(B_1(0,\frac{1}{14}))},&{\rm if }~x\in B_1(\frac{5}{14}e_1,\frac{1}{14}),\\
    \frac{\mu}{3{\rm Leb}(B_1(0,\frac{1}{28}))},&{\rm if~}x\in B_1(\frac{17}{28}e_1,\frac{1}{28}),\\
    \frac{1-\mu}{3{\rm Leb}(B_1(0,\frac{1}{28}))},&{\rm if~}x\in B_1(\frac{19}{28}e_1,\frac{1}{28}),\\
    \frac{7}{3},&{\rm if~}x_1\in[\frac{6}{7},1],
\end{matrix}\right.\]
where $e_1\in\R^d$ has the first element to be 1 and all other elements to be 0, and
\[\mu=\frac{\frac{\frac{3}{4}-\frac{3}{2}C_\rho-\frac{1}{2}C_\eta+\frac{3}{4}C_\rho C_\eta}{\frac{1}{4}-(\frac{1}{2}-\frac{7}{12}C_\eta)C_\rho-2C_\eta(1-2C_\sigma)\Delta}-\frac{\frac{1}{2}+C_\rho-\frac{1}{2}C_\eta-C_\rho C_\eta}{\frac{1}{4}+(\frac{1}{2}-\frac{7}{12}C_\eta)C_\rho+2C_\eta(1-2C_\sigma)\Delta}}{\frac{\frac{1}{4}+\frac{1}{2}C_\rho+\frac{1}{4}C_\rho C_\eta}{\frac{1}{4}+(\frac{1}{2}-\frac{7}{12}C_\eta)C_\rho+2C_\eta(1-2C_\sigma)\Delta}+\frac{\frac{1}{4}-\frac{1}{2}C_\rho-\frac{1}{4}C_\rho C_\eta}{\frac{1}{4}-(\frac{1}{2}-\frac{7}{12}C_\eta)C_\rho-2C_\eta(1-2C_\sigma)\Delta}}.\]
Since $C_\eta,C_\rho,\Delta$ are small enough, we have $\mu\approx\frac{1}{2}$. So $p_X$ is piecewise uniform.

\textbf{Verification of $\tilde{\mathscr{P}}\subset\mathscr{P}$:}

\textbf{Assumption \ref{ass:observe}:}

Suppose $\sum_{j\in\II}\1(\sigma_j=1)=C_\sigma m$ for some constant $C_\sigma>0$. For the specified distribution, if we denote $\Delta=C_\psi m\omega M^{-\beta_A}$, then
\[p_Y=\E\eta(X)=\frac{1}{2},\quad p_{1,1}^\sigma=\E\rho_{1|1}^\sigma(X)\eta(X)=\frac{1}{4}+\bigg(\frac{1}{2}-\frac{7}{12}C_\eta\bigg)C_\rho+2C_\eta(1-2C_\sigma)\Delta.\]
So Assumption~\ref{ass:observe} is satisfied if $C_\eta, C_\rho$ and $\Delta$ are small enough. 

\textbf{Group-Blind Assumptions:}

1) At first, we verify the group-blind assumptions. 

On the support of $p_X$, $\phi^\sigma$ equals
\begin{align*}
    &p_{1,1}^\sigma(p_Y-p_{1,1}^\sigma)\phi^\sigma(x)\\
    =&(p_Y\rho_{1|1}^\sigma(x)-p_{1,1}^\sigma)\eta(x)\\
    =&\left\{\begin{matrix}
    -\big\{C_\eta-\tilde C_\eta(\frac{1}{7}-x_1)^{\frac{d}{\gamma}}\big\}\big\{(1-\frac{7}{12}C_\eta)C_\rho+2C_\eta(1-2C_\sigma)\Delta\big\},&{\rm if~}x\in B_1(\frac{e_1}{14},\frac{1}{14}),\\
    -C_\eta\big\{(1-\frac{7}{12}C_\eta)C_\rho+2C_\eta(1-2C_\sigma)\Delta+\frac{1}{2}\sigma_jM^{-\beta_A}\psi(M(x-n_M(x)))\big\},&{\rm if~}x\in\X_j,j\in\II,\\
    -\big\{C_\eta+\tilde C_\eta(x_1-\frac{2}{7})^{\frac{d}{\gamma}}\big\}\big\{(1-\frac{7}{12}C_\eta)C_\rho+2C_\eta(1-2C_\sigma)\Delta\big\},&{\rm if~}x\in B_1(\frac{5}{14}e_1,\frac{1}{14}),\\
    C_\eta(\frac{5}{12} C_\rho-(1-2C_\sigma)\Delta),&{\rm if~}x_1\in[\frac{4}{7},\frac{5}{7}],\\
    (1-C_\eta)C_\eta(\frac{7}{12} C_\rho-2(1-2C_\sigma)\Delta),&{\rm if~}x_1\in[\frac{6}{7},1].
\end{matrix}\right.
\end{align*}

Now we identify $\lambda^{*\sigma}_\alpha$. Since
\begin{align*}
    \E\phi^\sigma(X)\1(2\eta(X)>1)
    =\frac{(1-C_\eta)C_\eta\big(\frac{7}{12} C_\rho+2\Delta\big)}{p_{1,1}^\sigma(p_Y-p_{1,1}^\sigma)}
    >0,
\end{align*}
then $\lambda^{*\sigma}_\alpha\ge 0$. If $C_\rho>-\frac{2C_\eta(1-2C_\sigma)\Delta}{1-\frac{7}{12}C_\eta}$, we set
\[\tilde\lambda=\frac{p_{1,1}^\sigma(p_Y-p_{1,1}^\sigma)(1-2C_\eta)}{C_\eta\{(1-\frac{7}{12}C_\eta)C_\rho+2C_\eta(1-2C_\sigma)\Delta\}},\]
and choose $C_\rho$ such that
\[\E\phi^\sigma(X)\1\big(2\eta(X)-1>\tilde\lambda\phi^\sigma(X)\big)=\alpha.\]
By monotonicity, it follows that $\lambda^{*\sigma}_\alpha=\tilde\lambda$. In this case, $g^{*\sigma}_\alpha=2\eta-1-\lambda^{*\sigma}_\alpha\phi^\sigma$ equals
\begin{align*}
    g^{*\sigma}_\alpha(x)
    =\left\{\begin{matrix}
    -\frac{\tilde C_\eta}{C_\eta}(\frac{1}{7}-x_1)^{\frac{d}{\gamma}},&{\rm if~}x\in B_1(\frac{e_1}{14},\frac{1}{14}),\\
    \frac{(\frac{1}{2}-C_\eta)\sigma_jM^{-\beta_A}\psi(M(x-n_M(x)))}{(1-\frac{7}{12}C_\eta)C_\rho+2C_\eta(1-2C_\sigma)\Delta},&{\rm if~}x\in\X_j,j\in\II,\\
    \frac{\tilde C_\eta}{C_\eta}(x_1-\frac{2}{7})^{\frac{d}{\gamma}},&{\rm if~}x\in B_1(\frac{5}{14}e_1,\frac{1}{14}),\\
    -\frac{(1-2C_\eta)(\frac{5}{12}C_\rho-(1-2C_\sigma)\Delta)}{(1-\frac{7}{12}C_\eta)C_\rho+2C_\eta(1-2C_\sigma)\Delta},&{\rm if~}x_1\in[\frac{4}{7},\frac{5}{7}],\\
    \frac{(1-2C_\eta)(\frac{5}{12}C_\rho+2(1-2C_\sigma)\Delta)}{(1-\frac{7}{12}C_\eta)C_\rho+2C_\eta(1-2C_\sigma)\Delta},&{\rm if~}x_1\in[\frac{6}{7},1].
\end{matrix}\right.
\end{align*}
Denote $C_B=\frac{1}{{\rm Leb}(B_1(0,\frac{1}{14}))}\int_{x\in B_1(\frac{e_1}{14},\frac{1}{14})}(\frac{1}{7}-x_1)^{\frac{d}{\gamma}}dx$, we choose $C_\rho$ such that
\begin{align*}
    \alpha=&\frac{1}{p_{1,1}^\sigma(p_Y-p_{1,1}^\sigma)}\bigg\{\bigg[\big(\frac{1}{36}+(1-2C_\sigma)m\omega\big)C_\eta-\frac{7}{12}\big(\frac{1}{6}+(1-2C_\sigma)m\omega\big)C_\eta^2\\
    &-(\frac{1}{6}-m\omega)(1-\frac{7}{12}C_\eta)\tilde C_\eta C_B\bigg]C_\rho-\bigg[\frac{2}{3}-\frac{1}{3}C_\sigma\\
    &-(1-2C_\sigma)\big(\frac{1}{3}+2(1-2C_\sigma)m\omega\big)C_\eta+(1-2C_\sigma)(\frac{1}{3}-2m\omega)\tilde C_\eta C_B\bigg]C_\eta\Delta\bigg\}.
\end{align*}
By choosing $(1+6(1-2C_\sigma)m\omega)C_\eta\ge(1-6m\omega)\tilde C_\eta C_B$, we get
\[\Delta\le\bigg(\frac{1}{12}+3m\omega\bigg)C_\rho,\]
then $g^{*\sigma}_\alpha(x)\le-\frac{(1-2C_\eta)(\frac{1}{3}-3m\omega)}{1+C_\eta(\frac{1}{6}+6m\omega)}$ for $x_1\in[\frac{4}{7},\frac{5}{7}]$, then the $C_\rho$ we choose satisfies
\[\E\phi^\sigma(X)\1\big(2\eta(X)-1>\lambda^{*\sigma}_\alpha\phi^\sigma(X)\big)=\alpha.\]
By setting $m\omega,C_\eta,\tilde C_\eta$ small enough, we get
\[C_\rho\asymp\alpha+m\omega M^{-\beta_A},\quad \lambda^{*\sigma}_\alpha\asymp\frac{1}{C_\rho}.\]

\textbf{Group-Blind Assumption \ref{ass:margin}:}

Firstly, we verify the margin assumption \ref{ass:margin}. For any $\epsilon<\frac{(1-2C_\eta)(\frac{1}{3}-3m\omega)}{1+C_\eta(\frac{1}{6}+6m\omega)}$, fix some $\tilde j\in\II$, we have
\begin{align*}
    &\Prob(|g^{*\sigma}_\alpha(X)|\le\epsilon)\\
    =&m\Prob\bigg(\frac{(\frac{1}{2}-C_\eta)M^{-\beta_A}}{(1-\frac{7}{12}C_\eta)C_\rho+2C_\eta(1-2C_\sigma)\Delta}\psi\bigg(M\bigg(X-\frac{2\tilde j-1}{14M}\bigg)\bigg)\le\epsilon\bigg)\\
    &+\Prob\bigg(0<\frac{\tilde C_\eta}{C_\eta}\bigg(\frac{1}{7}-X_1\bigg)^{\frac{d}{\gamma}}\le\epsilon,X\in B_1\bigg(\frac{e_1}{14},\frac{1}{14}\bigg)\bigg)\\
    &+\Prob\bigg(0<\frac{\tilde C_\eta}{C_\eta}\bigg(X_1-\frac{2}{7}\bigg)^{\frac{d}{\gamma}}\le\epsilon,X\in B_1\bigg(\frac{5}{14}e_1,\frac{1}{14}\bigg)\bigg)\\
    =&m\int_{B_2(0,\frac{1}{28M})}\1\bigg(\frac{(\frac{1}{2}-C_\eta)M^{-\beta_A}C_\psi}{(1-\frac{7}{12}C_\eta)C_\rho+2C_\eta(1-2C_\sigma)\Delta}\le\epsilon\bigg)\frac{2\omega}{{\rm Leb}(B_2(0,\frac{1}{28M}))}dx\\
    &+\Prob\bigg(\frac{1}{7}-\bigg(\frac{C_\eta\epsilon}{\tilde C_\eta}\bigg)^{\frac{\gamma}{d}}\wedge\frac{1}{7}\le X_1<\frac{1}{7},\sum_{j=2}^d|X_j|\le\bigg(\frac{1}{7}-X_1\bigg)\wedge X_1\bigg)\\
    &+\Prob\bigg(\frac{2}{7}< X_1\le\frac{2}{7}+\bigg(\frac{C_\eta\epsilon}{\tilde C_\eta}\bigg)^{\frac{\gamma}{d}}\wedge \frac{1}{7},\sum_{j=2}^d|X_j|\le\bigg(X_1-\frac{2}{7}\bigg)\wedge \bigg(\frac{3}{7}-X_1\bigg)\bigg)\\
    =&2m\omega\1\bigg(\frac{(\frac{1}{2}-C_\eta)M^{-\beta_A}C_\psi}{(1-\frac{7}{12}C_\eta)C_\rho+2C_\eta(1-2C_\sigma)\Delta}\le\epsilon\bigg)+c\epsilon^\gamma.
\end{align*}
If we set
\begin{equation}\label{eq:minimax_momega}
    m\omega\lesssim \big(C_\rho^{-1} M^{-\beta_A}\big)^\gamma,
\end{equation}
it follows that for any $\epsilon<\frac{(1-2C_\eta)(\frac{1}{3}-3m\omega)}{1+C_\eta(\frac{1}{6}+6m\omega)}$,
\[\Prob(|g^{*\sigma}_\alpha(X)|\le\epsilon)\lesssim\epsilon^\gamma,\]
then for any $\epsilon>0$, it still holds that
\[\Prob(|g^{*\sigma}_\alpha(X)|\le\epsilon)\lesssim\epsilon^\gamma.\]

\textbf{Group-Blind Assumption \ref{ass:ratio_poly}:}

Secondly, we check Assumption~\ref{ass:ratio_poly}. Denote $z=p_{1,1}^\sigma(p_Y-p_{1,1}^\sigma)\tilde z$, we have
\begin{align*}
    &g^{*\sigma}_\alpha(x)-z\phi^\sigma(x)\\
    =&\left\{\begin{matrix}
    -\big\{\frac{1}{C_\eta}+[(1-\frac{7}{12}C_\eta)C_\rho+2C_\eta(1-2C_\sigma)\Delta]\tilde z\big\}\tilde C_\eta(\frac{1}{7}-x_1)^{\frac{d}{\gamma}}\\
    +\big\{(1-\frac{7}{12}C_\eta)C_\rho+2C_\eta(1-2C_\sigma)\Delta\big\}C_\eta \tilde z,&{\rm if~}x\in B_1(\frac{e_1}{14},\frac{1}{14}),\\
    \big(\frac{\frac{1}{2}-C_\eta}{(1-\frac{7}{12}C_\eta)C_\rho+2C_\eta(1-2C_\sigma)\Delta}+\frac{1}{2}C_\eta \tilde z\big)\sigma_jM^{-\beta_A}\psi(M(x-n_M(x)))\\
    +\big\{(1-\frac{7}{12}C_\eta)C_\rho+2C_\eta(1-2C_\sigma)\Delta\big\}C_\eta \tilde z,&{\rm if~}x\in\X_j,j\in\II,\\
    \big\{\frac{1}{C_\eta}+[(1-\frac{7}{12}C_\eta)C_\rho+2C_\eta(1-2C_\sigma)\Delta] \tilde z\big\}\tilde C_\eta(x_1-\frac{2}{7})^{\frac{d}{\gamma}}\\
    +\big\{(1-\frac{7}{12}C_\eta)C_\rho+2C_\eta(1-2C_\sigma)\Delta\big\}C_\eta \tilde z,&{\rm if~}x\in B_1(\frac{5}{14}e_1,\frac{1}{14}),\\
    -\frac{(1-2C_\eta)(\frac{5}{12}C_\rho-(1-2C_\sigma)\Delta)}{(1-\frac{7}{12}C_\eta)C_\rho+2C_\eta(1-2C_\sigma)\Delta}-(\frac{5}{12}C_\rho-(1-2C_\sigma)\Delta)C_\eta \tilde z,&{\rm if~}x_1\in[\frac{4}{7},\frac{5}{7}],\\
    \frac{(1-2C_\eta)(\frac{5}{12}C_\rho+2(1-2C_\sigma)\Delta)}{(1-\frac{7}{12}C_\eta)C_\rho+2C_\eta(1-2C_\sigma)\Delta}-(1-C_\eta)(\frac{7}{12}C_\rho-2(1-2C_\sigma)\Delta)C_\eta \tilde z,&{\rm if~}x_1\in[\frac{6}{7},1].
    \end{matrix}\right.
\end{align*}
Note that $s={\rm sgn}(\lambda^{*\sigma}_\alpha)=1$, then for $z>0$, some calculation implies
\begin{align*}
    &\E|\phi^\sigma(X)|\1\bigg(0<\frac{g^{*\sigma}_\alpha(X)}{s\phi^\sigma(X)}< z\bigg)\\
    =&\E|\phi^\sigma(X)|\1\big(s_\phi(X)sg^{*\sigma}_\alpha(X)>0,s_\phi(X)s(g^{*\sigma}_\alpha(X)-sz\phi^\sigma(X))<0\big)\\
    \asymp&C_\rho\bigg(\frac{C_\rho |z|}{1+C_\rho |z|}\bigg)^\gamma.
\end{align*}
Similarly, for $z<0$, Equation \eqref{eq:minimax_momega} and some calculation imply
\begin{align*}
    &\E|\phi^\sigma(X)|\1\bigg(0>\frac{g^{*\sigma}_\alpha(X)}{s\phi^\sigma(X)}> z\bigg)\\
    =&\E|\phi^\sigma(X)|\1\big(s_\phi(X)sg^{*\sigma}_\alpha(X)<0,s_\phi(X)s(g^{*\sigma}_\alpha(X)-sz\phi^\sigma(X))>0\big)\\
    \asymp&C_\rho\bigg(\frac{C_\rho |z|}{1+C_\rho |z|}\bigg)^\gamma.
\end{align*}
So Assumption~\ref{ass:ratio_poly} is satisfied as long as $c_2$ is large enough. 

\textbf{Group-Blind Assumption \ref{ass:ratio_balance}:}

Thirdly, by taking $z=c_4\lambda^{*\sigma}_\alpha$ and $z=-\lambda^{*\sigma}_\alpha$ respectively, similar argument implies Assumption~\ref{ass:ratio_balance} is satisfied as long as $c_3,c_4$ are large enough. 

\textbf{Group-Blind Assumption \ref{ass:holder_blind}:}

It follows from the construction of $\eta$ and $\rho_{a|y}$ that Assumption \ref{ass:holder_blind} is satisfied.

\textbf{Group-Blind Assumption \ref{ass:density_blind}:}

If we choose $m\omega\lesssim 1$ small enough, then $p_X\asymp 1$ on the support of $p_X$ and Assumption \ref{ass:density_blind} is satisfied.

\textbf{Group-Aware Assumptions:}

2) Then we verify the group-aware assumptions. 

\textbf{Group-Aware Assumption \ref{ass:holder_aware}:}

Note that
\[\eta^{\rm aware}(x,a)=\frac{\eta^{\rm blind}(x)\rho_{a|1}(x)}{\eta^{\rm blind}(x)\rho_{a|1}(x)+(1-\eta^{\rm blind}(x))\rho_{a|0}(x)},\]
then it is straightforward to verify that $\eta^{\rm aware}(\cdot, 1),\eta^{\rm aware}(\cdot,2)$ are also $\beta_Y$-H\"older smooth. 

\textbf{Group-Aware Assumption \ref{ass:density_aware}:}

Note that
\[p_{X|a}=\frac{(\rho_{a|1}\eta+\rho_{a|0}(1-\eta))p_X}{\Prob(A=a)},\]
then it is clear that if $C_\eta$, $\tilde C_\eta$, $C_\rho$ do not exceed some small constant and $m\omega\lesssim 1$, we have $p_{X|a}\asymp 1$ and satisfies Assumption \ref{ass:density_aware}.

\textbf{Group-Aware Assumptions \ref{ass:margin}, \ref{ass:ratio_poly}, and \ref{ass:ratio_balance}:}

Moreover, we have
\[\U\big(\1(2\eta^{\rm aware}(X,A)>1)\big)=0.\]
Therefore $g^{*{\rm aware}}_\alpha(x,a)=2\eta^{\rm aware}(x,a)-1$. Similar to the group-blind scenario, direct calculation implies that the group-aware Assumptions~\ref{ass:margin}, \ref{ass:ratio_poly}, \ref{ass:ratio_balance} are also satisfied.

\textbf{Derivation of the Lower Bound:}

Now we derive the minimax lower bound. Note that for $\sigma,\sigma'\in\{-1,1\}^m$, we have the inequality
\begin{align*}
    &\E|f^{*\sigma}_\alpha(X)-f^{*\sigma'}_\alpha(X)|\1\big(g^{*\sigma}_\alpha(X)\ne 0,g^{*\sigma'}_\alpha(X)\ne 0\big)\\
    \le&\E|f^{*\sigma}_\alpha(X)-\hat f(X)|\1(g^{*\sigma}_\alpha(X)\ne 0)+\E|f^{*\sigma'}_\alpha(X)-\hat f(X)|\1(g^{*\sigma'}(X)\ne 0).
\end{align*}
Denote $H(\sigma,\sigma')=\sum_{j\in\II}\1(\sigma_j\ne\sigma'_j)$ to be the Hamming distance between $\sigma$ and $\sigma'$. Suppose $C_\sigma\le\frac{1}{3}$, it follows from Lemma A.1 in \cite{rigollet2009optimal} that there exists a set $\Omega\subset\{-1,1\}^m$ of $\sigma$'s with $\log|\Omega|\ge Cm$ and
\[\Omega=\bigg\{\sigma:\sigma\in\{-1,1\}^m, \sum_{j\in\II}\1(\sigma_j=1)=C_\sigma m\bigg\},\quad H(\sigma,\sigma')\ge \frac{C_\sigma}{2} m,\quad \forall \sigma\ne\sigma'\in\Omega.\]
Then for any $\sigma\ne\sigma'\in\Omega$, we have
\begin{align*}
    \E|f^{*\sigma}_\alpha(X)-f^{*\sigma'}_\alpha(X)|\1\big(g^{*\sigma}_\alpha(X)\ne 0,g^{*\sigma'}_\alpha(X)\ne 0\big)=H(\sigma,\sigma')2\omega\ge C_\sigma m\omega.
\end{align*}

Denote $P_{X,A,Y}^\sigma=P_X P_{Y|X} P_{A|X,Y}^\sigma$, using the inequality
\[(1+a)\log\frac{1+a}{1+b}\le a-b+(a-b)^2,\quad \forall|a|<\frac{1}{2},|b|<\frac{1}{2},\]
we have
\begin{align*}
    &{\rm KL}(P_{X,A,Y}^{\sigma\otimes N},P_{X,A,Y}^{\sigma'\otimes N})\\
    =&N{\rm KL}(P_{X,A,Y}^\sigma,P_{X,A,Y}^{\sigma'})\\
    =&N\int\eta(x)\rho_{1|1}^\sigma(x)\log\frac{\rho_{1|1}^\sigma(x)}{\rho_{1|1}^{\sigma'}(x)}p_X(x)dx+N\int\eta(x)(1-\rho_{1|1}^\sigma(x))\log\frac{1-\rho_{1|1}^\sigma(x)}{1-\rho_{1|1}^{\sigma'}(x)}p_X(x)dx\\
    \le&4N\int\eta(x)\big(\rho_{1|1}^\sigma(x)-\rho_{1|1}^{\sigma'}\big)^2p_X(x)dx\\
    \le&CNM^{-2\beta_A}\omega H(\sigma,\sigma')\\
    \le&CN M^{-2\beta_A}m\omega.
\end{align*}

Since $\beta_A\gamma\le d$, we set
\[M\asymp N^{\frac{1}{2\beta_A+d}},\quad \omega\asymp N^{-\frac{d}{2\beta_A+d}},\quad m\asymp\alpha^{-\gamma} N^{\frac{d-\beta_A\gamma}{2\beta_A+d}}\wedge N^{\frac{d}{2\beta_A+d}},\quad \lambda^*_\alpha\asymp\alpha^{-1}\wedge N^{\frac{\beta_A}{2\beta_A+d}},\]
then we have $p_X\asymp 1$ on the support of $p_X$ and Assumptions~\ref{ass:margin} and \ref{ass:density_blind} are satisfied. Moreover, we have
\[\max_{\sigma,\sigma'\in\Omega}{\rm KL}(P_{X,A,Y}^{\sigma\otimes n},P_{X,A,Y}^{\sigma'\otimes n})\lesssim \log|\Omega|,\]
then if we denote
\[\tilde\epsilon_\rho\asymp \big(|\lambda^*_\alpha|N^{-\frac{\beta_A}{2\beta_A+d}}\big)^{1+\gamma}\asymp \big(\alpha^{-1}N^{-\frac{\beta_A}{2\beta_A+d}}\big)^{1+\gamma}\wedge 1.\]
Fano's Lemma and Equation \eqref{eq:minimax_reduction_binary_eoo} imply
\begin{align*}
    \inf_{\A\in\mathscr{A}^{\rm blind}}\sup_{P\in\mathscr{P}}\Prob_{\D_{\rm all}\sim P^{\otimes N}}\big(T_1(\A(\D_{\rm all}))\gtrsim\tilde\epsilon_\rho\big)
    \ge c.
\end{align*}
\end{proof}

\subsection{Proof of Claim \ref{clm:eta}}

\begin{proof}[Proof of Claim \ref{clm:eta}]
   This proof focuses on constructing the subclass $\tilde{\mathscr{P}}$ of distribution $P_{X,A,Y}=P_XP_{Y|X}P_{A|X,Y}$. Recall $p_X$ is the density of $P_X$, $\eta(X)=\Prob(Y=1|X)$, and $\rho_{a|y}(X)=\Prob(A=a|X,Y=y)$. Then it suffices to specify $p_X$, $\eta$, and $\rho_{a|y}$ respectively.
   
At first, we consider the case where $|\lambda^*_\alpha|$ is large. 

We use the same notations as in the proof of Claim \ref{clm:rho}, but redefine $p_X$, $\eta^\sigma$ and $\rho_{1|1}$ as follows.

\textbf{Construction of $\eta$:}

\[\eta^\sigma(x)=\left\{\begin{matrix}
    C_\eta,&{\rm if~}x_1\in[0,\frac{1}{7}],\\
    C_\eta+\sigma_jM^{-\beta_Y}\psi(M(x-n_M(x))),&{\rm if~}x_1\in[\frac{1}{7},\frac{2}{7}],\\
    C_\eta,&{\rm if~}x_1\in[\frac{2}{7},\frac{3}{7}],\\
    \frac{1}{4}+\frac{1}{2}C_\eta-\frac{1}{2}(\frac{1}{2}-C_\eta)h(7x_1-3),&{\rm if~}x_1\in[\frac{3}{7},\frac{4}{7}],\\
    \frac{1}{2},&{\rm if~}x_1\in[\frac{4}{7},\frac{5}{7}],\\
    \frac{3}{4}-\frac{1}{2}C_\eta-\frac{1}{2}(\frac{1}{2}-C_\eta)h(7x_1-5),&{\rm if~}x_1\in[\frac{5}{7},\frac{6}{7}],\\
    1-C_\eta,&{\rm if~}x_1\in[\frac{6}{7},1],
\end{matrix}\right.\]
where $C_\eta$ is small enough such that $\eta^\sigma\in\HH(\beta_A,L_A,\R^d)$. Here $C_\eta$ may decrease when $\alpha$ varies. 

\textbf{Construction of $\rho_{a|y}$:}

\[\rho_{1|1}(x)-\frac{1}{2}=\left\{\begin{matrix}
    -C_\rho-\tilde C_\rho(\frac{1}{7}-x_1)^{\frac{d}{\gamma}},&{\rm if~}x_1\in[0,\frac{1}{7}],\\
    -C_\rho,&{\rm if~}x_1\in[\frac{1}{7},\frac{2}{7}],\\
    -C_\rho+\tilde C_\rho(x_1-\frac{2}{7})^{\frac{d}{\gamma}},&{\rm if~}x_1\in[\frac{2}{7},\frac{3}{7}],\\
    \tilde h(x),&{\rm if~}x_1\in[\frac{3}{7},\frac{4}{7}],\\
    C_\rho+\frac{1}{2}C_\eta C_\rho,&{\rm if~}x_1\in[\frac{4}{7},\frac{5}{7}],\\
    C_\rho+\frac{1}{4}C_\eta C_\rho+\frac{1}{4}C_\eta C_\rho h(7x_1-5),&{\rm if~}x_1\in[\frac{5}{7},\frac{6}{7}],\\
    C_\rho,&{\rm if~}x_1\in[\frac{6}{7},1],
\end{matrix}\right.\]
$C_\rho,\tilde C_\rho$ are small constants and $\tilde h$ is a polynomial such that $\rho_{1|1}\in\HH(\beta_Y,L_Y,\R^d)$. We assume the existence of $\tilde h$, otherwise, we can always extend the interval $[\frac{3}{7},\frac{4}{7}]$ to fulfill it. So Assumption~\ref{ass:holder_blind} is satisfied. We also define $\rho_{1|0}$ as
\[\rho_{1|0}(x)=\left\{\begin{matrix}
    \frac{1}{4}, &{\rm if~}x_1\in[0,\frac{3}{7}],\\
    \frac{1}{2}+C_\rho+\frac{1}{2}C_\eta C_\rho-\tilde C_\eta(\frac{9}{14}-x_1)^{\frac{d}{\gamma}}, &{\rm if~}x_1\in[\frac{4}{7},\frac{9}{14}],\\
    \frac{1}{2}+C_\rho+\frac{1}{2}C_\eta C_\rho+\tilde C_\eta(x_1-\frac{9}{14})^{\frac{d}{\gamma}}, &{\rm if~}x_1\in[\frac{9}{14},\frac{5}{7}],\\
    \frac{3}{4}, &{\rm if~}x_1\in[\frac{6}{7},1].
\end{matrix}\right.\]
And $\rho_{1|0}$ on $([\frac{3}{7},\frac{4}{7}]\cup[\frac{5}{7},\frac{6}{7}])\times[0,1]^{d-1}$ is defined such that $\rho_{1|0}$ is $\beta_Y$-H\"older smooth. 

\textbf{Construction of $p_X$:}

Denote
\[\Delta=C_\psi m\omega M^{-\beta_Y},\]
\[p_X(x)=\left\{\begin{matrix}
    \frac{\frac{1}{6}-m\omega}{{\rm Leb}(B_1(0,\frac{1}{14}))},&{\rm if~}x\in B_1(\frac{e_1}{14},\frac{1}{14}),\\
    \frac{2\omega}{{\rm Leb}(B_2(0,\frac{C_\omega}{28M}))},&{\rm if} ~x\in B_2(\frac{2j-1}{14M},\frac{C_\omega}{28M}), j\in\II,\\
    \frac{\frac{1}{6}-m\omega}{{\rm Leb}(B_1(0,\frac{1}{14}))},&{\rm if }~x\in B_1(\frac{5}{14}e_1,\frac{1}{14}),\\
    \frac{\mu}{3{\rm Leb}(B_1(0,\frac{1}{28}))},&{\rm if~}x\in B_1(\frac{17}{28}e_1,\frac{1}{28}),\\
    \frac{1-\mu}{3{\rm Leb}(B_1(0,\frac{1}{28}))},&{\rm if~}x\in B_1(\frac{19}{28}e_1,\frac{1}{28}),\\
    \frac{7}{3},&{\rm if~}x_1\in[\frac{6}{7},1],
\end{matrix}\right.\]
with $C_\omega\in(0,1]$ to be specified later and 
\[\mu=\frac{\frac{\frac{3}{4}-\frac{3}{2}C_\rho-\frac{1}{2}C_\eta+\frac{3}{4}C_\rho C_\eta}{\frac{1}{4}-(\frac{1}{2}-\frac{7}{12}C_\eta)C_\rho-(1+2C_\rho)(1-2C_\sigma)\Delta}-\frac{\frac{1}{2}+C_\rho-\frac{1}{2}C_\eta-C_\rho C_\eta}{\frac{1}{4}+(\frac{1}{2}-\frac{7}{12}C_\eta)C_\rho-(1-2C_\rho)(1-2C_\sigma)\Delta}}{\frac{\frac{1}{4}+\frac{1}{2}C_\rho+\frac{1}{4}C_\rho C_\eta}{\frac{1}{4}+(\frac{1}{2}-\frac{7}{12}C_\eta)C_\rho-(1-2C_\rho)(1-2C_\sigma)\Delta}+\frac{\frac{1}{4}-\frac{1}{2}C_\rho-\frac{1}{4}C_\rho C_\eta}{\frac{1}{4}-(\frac{1}{2}-\frac{7}{12}C_\eta)C_\rho-(1+2C_\rho)(1-2C_\sigma)\Delta}}.\]

\textbf{Verification of $\tilde{\mathscr{P}}\subset\mathscr{P}$:}

\textbf{Assumption \ref{ass:observe}:}

We have
\[p_Y^\sigma=\E\eta^\sigma(X)=\frac{1}{2}-2(1-2C_\sigma)\Delta,\]
\[p_{1,1}^\sigma=\E\eta^\sigma(X)\rho_{1|1}(X)=\frac{1}{4}+(\frac{1}{2}-\frac{7}{12}C_\eta)C_\rho-(1-2C_\rho)(1-2C_\sigma)\Delta.\]
So Assumption~\ref{ass:observe} is satisfied if $\Delta$ and $C_\rho$ are small enough.

\textbf{Group-Blind Assumptions:}

1) We start from the group-blind assumptions. Now we have on the support of $p_x$, $\phi^\sigma=\frac{p_Y^\sigma\rho_{1|1}-p_{1,1}^\sigma}{p_{1,1}^\sigma(p_Y^\sigma-p_{1,1}^\sigma)}\eta^\sigma$ equals
\begin{align*}
    &p_{1,1}^\sigma(p_Y^\sigma-p_{1,1}^\sigma)\phi^\sigma(x)\\
    =&\left\{\begin{matrix}
        -\big\{(1-\frac{7}{12}C_\eta)C_\rho+(\frac{1}{2}-2(1-2C_\sigma)\Delta)\tilde C_\rho(\frac{1}{7}-x_1)^{\frac{d}{\gamma}}\big\}C_\eta,&{\rm if~}x\in B_1(\frac{e_1}{14},\frac{1}{14}),\\
        -(1-\frac{7}{12}C_\eta)C_\rho\big(C_\eta+\sigma_jM^{-\beta_Y}\psi(M(x-n_M(x)))\big),&{\rm if~}x\in\X_j,j\in\II,\\
        -\big\{(1-\frac{7}{12}C_\eta)C_\rho-(\frac{1}{2}-2(1-2C_\sigma)\Delta)\tilde C_\rho(x_1-\frac{2}{7})^{\frac{d}{\gamma}}\big\}C_\eta,&{\rm if~}x\in B_1(\frac{5}{14}e_1,\frac{1}{14}),\\
        \frac{5}{12}C_\rho C_\eta-\frac{1}{2}(4+C_\eta)C_\rho(1-2C_\sigma)\Delta,&{\rm if~}x_1\in[\frac{4}{7},\frac{5}{7}],\\
        (1-C_\eta)C_\rho(\frac{7}{12}C_\eta-4(1-2C_\sigma)\Delta),&{\rm if~}x_1\in[\frac{6}{7},1],
    \end{matrix}\right.
\end{align*}
then
\[\E\phi^\sigma(X)\1(2\eta^\sigma(X)>1)>0\]
which implies $\lambda^{*\sigma}_\alpha\ge 0$. Moreover, we set
\[\tilde\lambda=\frac{p_{1,1}^\sigma(p_Y^\sigma-p_{1,1}^\sigma)(1-2C_\eta)}{(1-\frac{7}{12}C_\eta)C_\rho C_\eta},\]
if $C_\eta$ is chosen such that
\[\E\phi^\sigma(X)\1\big(2\eta^\sigma(X)-1>\tilde\lambda\phi^\sigma(X)\big)=\alpha,\]
by monotonicity, it follows that $\lambda^{*\sigma}_\alpha=\tilde\lambda$. Then, on the support of $p_X$, we have $g^{*\sigma}_\alpha=2\eta^\sigma-1-\lambda^{*\sigma}_\alpha\phi^\sigma$ equals
\[g^{*\sigma}_\alpha(x)=\left\{\begin{matrix}
    \frac{(1-2C_\eta)(\frac{1}{2}-2(1-2C_\sigma)\Delta)\tilde C_\rho}{(1-\frac{7}{12}C_\eta)C_\rho}(\frac{1}{7}-x_1)^{\frac{d}{\gamma}},&{\rm if~}x\in B_1(\frac{e_1}{14},\frac{1}{14}),\\
    \frac{1}{C_\eta}\sigma_jM^{-\beta_Y}\psi(M(x-n_M(x))),&{\rm if~}x\in\X_j,j\in\II,\\
    -\frac{(1-2C_\eta)(\frac{1}{2}-2(1-2C_\sigma)\Delta)\tilde C_\rho}{(1-\frac{7}{12}C_\eta)C_\rho}(x_1-\frac{2}{7})^{\frac{d}{\gamma}},&{\rm if~}x\in B_1(\frac{5}{14}e_1,\frac{1}{14}),\\
    -\frac{1-2C_\eta}{(1-\frac{7}{12}C_\eta)C_\eta}\big\{\frac{5}{12}C_\eta-(2+\frac{1}{2}C_\eta)(1-2C_\sigma)\Delta\big\},&{\rm if~}x_1\in[\frac{4}{7},\frac{5}{7}],\\
    \frac{1-2C_\eta}{(1-\frac{7}{12}C_\eta)C_\eta}\big\{\frac{5}{12}C_\eta+4(1-C_\eta)(1-2C_\sigma)\Delta\big\},&{\rm if~} x_1\in[\frac{6}{7},1].
\end{matrix}\right.\]
Now we choose $C_\eta$ satisfies
\begin{equation}
    \begin{aligned}\label{eq:minimax_Ceta}
    \alpha=&\frac{1}{p_{1,1}^\sigma(p_Y^\sigma-p_{1,1}^\sigma)}\bigg\{\bigg[\big(\frac{1}{36}+(1-2C_\sigma)m\omega\big)C_\rho-\frac{1}{2}(\frac{1}{6}-m\omega)\tilde C_\rho C_B\\
    &-\frac{7}{12}\big(\frac{1}{6}+(1-2C_\sigma)m\omega\big)C_\rho C_\eta\bigg]C_\eta-\bigg[(\frac{4}{3}-\frac{2}{3}C_\sigma)C_\rho\\
    &-(\frac{4}{3}-\frac{3}{2}C_\sigma)-(\frac{1}{3}-2m\omega)(1-2C_\sigma)\tilde C_\rho C_BC_\eta\bigg]\Delta\bigg\}.
\end{aligned}
\end{equation}
By choosing $C_\eta$ small enough, we know
\[\Delta\le\bigg(\frac{1}{12}+3m\omega\bigg)C_\eta,\]
and it follows
\[g^{*\sigma}_\alpha(x)<0,\quad \forall x_1\in[\frac{4}{7},\frac{5}{7}].\]
Then we have
\[\E\phi^\sigma(X)\1(g^{*\sigma}_\alpha(X)>0)=\alpha.\]

Note that only small values of $\alpha$ are of interest. Since Equation \eqref{eq:minimax_Ceta} is a quadratic equation of $C_\eta$, it has two solutions. Then we will choose these two solutions for different settings.
\begin{itemize}
    \item[a)]For small $\alpha$ with $\alpha\gtrsim N^{-\frac{\beta_Y\gamma}{(2\beta_Y+d)(1+\gamma)}}$, we set $m\omega, \tilde C_\rho, C_\eta$ to be small enough such that
    \[\big(\frac{1}{36}+(1-2C_\sigma)m\omega\big)C_\rho-\frac{1}{2}(\frac{1}{6}-m\omega)\tilde C_\rho C_B-\frac{7}{12}\big(\frac{1}{6}+(1-2C_\sigma)m\omega\big)C_\rho C_\eta\gtrsim 1,\]
    then we get
    \[C_\eta\asymp\alpha+m\omega M^{-\beta_Y},\quad\lambda^{*\sigma}_\alpha\asymp\frac{1}{C_\eta}.\]
    In this case, $\lambda_\alpha^*$ is of order $\alpha^{-1}$, and we will prove the lower bound 
    \[(|\lambda^*_\alpha|N^{-\frac{\beta_Y}{2\beta_Y+d}})^{1+\gamma}\asymp(\alpha^{-1}N^{-\frac{\beta_Y}{2\beta_Y+d}})^{1+\gamma}\]
    for the excess risk.
    \item[b)]For smaller $\alpha$ with $\alpha\lesssim N^{-\frac{\beta_Y\gamma}{(2\beta_Y+d)(1+\gamma)}}$. We will set $C_\eta\gtrsim 1$ to be a constant, therefore Equation \eqref{eq:minimax_Ceta} implies $C_\eta$ satisfies
    \[\big(\frac{1}{36}+(1-2C_\sigma)m\omega\big)C_\rho-\frac{1}{2}(\frac{1}{6}-m\omega)\tilde C_\rho C_B-\frac{7}{12}\big(\frac{1}{6}+(1-2C_\sigma)m\omega\big)C_\rho C_\eta\asymp\alpha+m\omega M^{-\beta_Y}.\]
    For $m\omega,\tilde C_\rho$ small enough, we get $C_\eta\approx\frac{2}{7}<\frac{1}{2}$, so the construction is valid. In this case, $\lambda^*_\alpha\asymp 1$, and similar argument concludes the lower bound
    \[(|\lambda^*_\alpha|N^{-\frac{-\beta_Y}{2\beta_Y+d}})^{1+\gamma}.\]
\end{itemize}

In the following, we only analyze the more complicated case (a), and case (b) can be derived similarly.

\textbf{Group-Blind Assumption \ref{ass:margin}:}

Firstly, we verify the margin assumption~\ref{ass:margin}. For any $\epsilon<\frac{1-2C_\eta}{4-\frac{7}{3}C_\eta}$, fix some $\tilde j\in\II$, we have
\begin{align*}
    &\Prob(|g^{*\sigma}_\alpha(X)|\le\epsilon)\\
    =&m\Prob\bigg(0<\frac{1}{C_\eta} M^{-\beta_Y}\psi\bigg(M\bigg(X-\frac{2\tilde j-1}{14M}\bigg)\bigg)\le\epsilon\bigg)\\
    &+\Prob\bigg(0<\frac{(1-2C_\eta)(\frac{1}{2}-2(1-2C_\sigma)\Delta)\tilde C_\rho}{(1-\frac{7}{12}C_\eta)C_\rho}\bigg(\frac{1}{7}-X_1\bigg)^{\frac{d}{\gamma}}\le\epsilon,X\in B_1(\frac{e_1}{14},\frac{1}{14})\bigg)\\
    &+\Prob\bigg(0<\frac{(1-2C_\eta)(\frac{1}{2}-2(1-2C_\sigma)\Delta)\tilde C_\rho}{(1-\frac{7}{12}C_\eta)C_\rho}\bigg(X_1-\frac{2}{7}\bigg)^{\frac{d}{\gamma}}\le\epsilon,X\in B_1(\frac{5}{14}e_1,\frac{1}{14})\bigg)\\
    =&2 m\omega\1\bigg(\frac{1}{C_\eta} M^{-\beta_Y}C_\psi\le\epsilon\bigg)+c\epsilon^\gamma.
\end{align*}
If we set
\[m\omega\lesssim \big(C_\eta^{-1} M^{-\beta_Y}\big)^\gamma,\]
then for any $\epsilon<c$,
\[\Prob(|g^{*\sigma}_\alpha(X)|\le\epsilon)\lesssim\epsilon^\gamma,\]
furthermore, we have for any $\epsilon>0$,
\[\Prob(|g^{*\sigma}_\alpha(X)|\le\epsilon)\lesssim\epsilon^\gamma.\]

\textbf{Group-Blind Assumptions \ref{ass:ratio_poly} and \ref{ass:ratio_balance}:}

Then we check the Assumptions~\ref{ass:ratio_poly} and \ref{ass:ratio_balance}. Denote $z=p^\sigma_{1,1}(p^\sigma_Y-p^\sigma_{1,1})\tilde z$, we have
\begin{align*}
    &g^{*\sigma}_\alpha(x)-z\phi^\sigma(x)\\
    =&\left\{\begin{matrix}
        (\frac{1}{2}-2(1-2C_\sigma)\Delta)(\frac{1-2C_\eta}{(1-\frac{7}{12}C_\eta)C_\rho}+C_\eta\tilde z)\tilde C_\rho(\frac{1}{7}-x_1)^{\frac{d}{\gamma}}\\
        +(1-\frac{7}{12}C_\eta)C_\rho C_\eta\tilde z,&{\rm if~}x\in B_1(\frac{e_1}{14},\frac{1}{14}),\\
        \big\{\frac{1}{C_\eta}+(1-\frac{7}{12}C_\eta)C_\rho\tilde z\big\}\sigma_jM^{-\beta_Y}\psi(M(x-n_M(x)))\\
        +(1-\frac{7}{12}C_\eta)C_\rho C_\eta\tilde z,&{\rm if~}x\in\X_j,j\in\II,\\
        -(\frac{1}{2}-2(1-2C_\sigma)\Delta)(\frac{1-2C_\eta}{(1-\frac{7}{12}C_\eta)C_\rho}+C_\eta\tilde z)\tilde C_\rho(x_1-\frac{2}{7})^{\frac{d}{\gamma}}\\
        +(1-\frac{7}{12}C_\eta)C_\rho C_\eta\tilde z,&{\rm if~}x\in B_1(\frac{5}{14}e_1,\frac{1}{14}),\\
        -(\frac{1-2C_\eta}{(1-\frac{7}{12}C_\eta)C_\eta}+C_\rho\tilde z)\big\{\frac{5}{12}C_\eta-(2+\frac{1}{2}C_\eta)(1-2C_\sigma)\Delta\big\},&{\rm if~}x_1\in[\frac{4}{7},\frac{5}{7}],\\
        \frac{1-2C_\eta}{(1-\frac{7}{12}C_\eta)C_\eta}\big\{\frac{5}{12}C_\eta+4(1-C_\eta)(1-2C_\sigma)\Delta\big\}\\
        -(1-C_\eta)C_\rho(\frac{7}{12}C_\eta-4(1-2C_\sigma)\Delta)\tilde z,&{\rm if~}x_1\in[\frac{6}{7},1].
    \end{matrix}\right.
\end{align*}
Similar to the analysis of the error of $\rho_{1,1}$, for $z>0$,
\begin{align*}
    &\E|\phi^\sigma(X)|\1\bigg(0<\frac{g^{*\sigma}_\alpha(X)}{s\phi^\sigma(X)}<z\bigg)\\
    =&\E|\phi^\sigma(X)|\1\big(s_\phi(X)sg^{*\sigma}_\alpha(X)>0,s_\phi(X)s(g^{*\sigma}_\alpha(X)-sz\phi^\sigma(X))<0\big)\\
    \asymp&C_\eta\bigg(\frac{C_\eta|z|}{1+C_\eta|z|}\bigg)^\gamma,
\end{align*}
and for $z<0$,
\begin{align*}
    &\E|\phi^\sigma(X)|\1\bigg(0>\frac{g^{*\sigma}_\alpha(X)}{s\phi^\sigma(X)}>z\bigg)\\
    =&\E|\phi^\sigma(X)|\1\big(s_\phi(X)sg^{*\sigma}_\alpha(X)<0,s_\phi(X)s(g^{*\sigma}_\alpha(X)-sz\phi^\sigma(X))>0\big)\\
    \asymp&C_\eta\bigg(\frac{C_\eta|z|}{1+C_\eta|z|}\bigg)^\gamma,
\end{align*}
so Assumptions~\ref{ass:ratio_poly} and \ref{ass:ratio_balance} are satisfied if $c_2$, $c_3$ and $c_4$ are large enough.

\textbf{Group-Blind Assumption \ref{ass:holder_blind}:}

The H\"older smoothness follows from the definition of $\eta$ and $\rho_{a|y}$.

\textbf{Group-Blind Assumption \ref{ass:density_blind}:}

If we choose $m\omega\lesssim 1$, then $p_X\asymp 1$ on the support of $p_X$, so Assumption \ref{ass:density_blind} is satisfied.

\textbf{Group-Aware Assumptions:}

2) Then we verify the group-aware assumptions. Similar to the proof of Claim \ref{clm:rho}, it is straightforward to verify that $\eta^{\rm aware}(\cdot,a)$ are $\beta_Y$-H\"older smooth, $\U\big(\1(2\eta^{\rm aware}(X,A)>1)\big)=0$, so $g^{*{\rm aware}}_\alpha(x,a)=2\eta^{\rm aware}(x,a)-1$, and the group-aware Assumptions~\ref{ass:margin}, \ref{ass:ratio_poly}, \ref{ass:ratio_balance} are also satisfied.

\textbf{Derivation of the Lower Bound:}

Then we are ready to prove the lower bound. For the same $\Omega$ defined in the proof of Claim \ref{clm:rho}, since $\eta^\sigma\ge C_\eta$ for any $\sigma\in\Omega$, then for all $\sigma\ne\sigma'\in\Omega$, we have
\[\eta^\sigma\log\frac{\eta^\sigma}{\eta^{\sigma'}}\le \eta^\sigma-\eta^{\sigma'}+\frac{1}{2C_\eta}(\eta^\sigma-\eta^{\sigma'})^2.\]
It follows
\begin{align*}
    &{\rm KL}(P_{X,A,Y}^{\sigma\otimes N},P_{X,A,Y}^{\sigma'\otimes N})\\
    =&N{\rm KL}(P_{X,A,Y}^\sigma,P_{X,A,Y}^{\sigma'})\\
    =&N\int\eta^\sigma(x)\log\frac{\eta^\sigma(x)}{\eta^{\sigma'}(x)}p_X(x)dx+N\int(1-\eta^\sigma(x))\log\frac{1-\eta^\sigma(x)}{1-\eta^{\sigma'}(x)}p_X(x)dx\\
    \le&\frac{1}{C_\eta}N\int(\eta^\sigma(x)-\eta^{\sigma'}(x))^2p_X(x)dx\\
    \le&\frac{4C_\sigma C_\psi^2}{C_\eta}Nm\omega M^{-2\beta_Y}.
\end{align*}
Since $\beta_Y\gamma\le d$, $\alpha\gtrsim N^{-\frac{\beta_Y\gamma}{(2\beta_Y+d)(1+\gamma)}}\gtrsim N^{-\frac{\beta_Y}{2\beta_Y+d}}$, we set
\[M\asymp N^{\frac{1}{2\beta_Y+d}},\quad C_\omega\asymp \alpha^{\frac{1}{d}},\quad\omega\asymp \alpha N^{-\frac{d}{2\beta_Y+d}},\quad m\asymp\alpha^{-(\gamma+1)} N^{\frac{d-\beta_Y\gamma}{2\beta_Y+d}},\quad \lambda^*_\alpha\asymp \alpha^{-1},\]
then we have $m\lesssim M^d$, $p_X\asymp 1$ on its support and Assumptions~\ref{ass:margin} and \ref{ass:density_blind} are satisfied. Moreover, we have
\[\max_{\sigma,\sigma'\in\Omega}{\rm KL}(P_{X,A,Y}^{\sigma\otimes N},P_{X,A,Y}^{\sigma'\otimes N})\lesssim\log|\Omega|.\]
If we denote
\[\epsilon'_\eta\asymp  \big(|\lambda^*_\alpha|N^{-\frac{\beta_Y}{2\beta_Y+d}}\big)^{1+\gamma}\asymp \big(\alpha^{-1}N^{-\frac{\beta_Y}{2\beta_Y+d}}\big)^{1+\gamma}.\]
Using the same reasoning in the proof of Claim \ref{clm:rho}, Fano's Lemma and Equation \eqref{eq:minimax_reduction_binary_eoo} imply
\begin{align*}
    \inf_{\A\in\mathscr{A}^{\rm blind}}\sup_{P\in\mathscr{P}}\Prob_{\D_{\rm all}\sim P^{\otimes N}}\big(T_1(\A(\D_{\rm all}))\gtrsim\epsilon'_\eta\big)
    \ge c.
\end{align*}

Then we consider the case where $|\lambda^*_\alpha|$ is small. Specifically, if we set $\rho_{1|1}=\rho_{1|0}=\frac{1}{2}$, then we know $f^*_\alpha=\1(2\eta>1)$, then similar to the proof of Theorem 3.5 in \cite{audibert2007fast}, we can get
\begin{align*}
    \inf_{\A\in\mathscr{A}^{\rm blind}}\sup_{P\in\mathscr{P}}\Prob_{\D_{\rm all}\sim P^{\otimes N}}\big(T_1(\A(\D_{\rm all}))\gtrsim N^{-\frac{\beta_Y(1+\gamma)}{2\beta_Y+d}}\big)
    \ge c.
\end{align*}

\end{proof}

\subsection{Proof of Claim \ref{clm:fair}}

\begin{proof}[Proof of Claim \ref{clm:fair}]

In this proof, we construct a pair of distributions $\tilde{\mathscr{P}}=\{P_{X,A,Y},\bar P_{X,A,Y}\}$ with $P_{X,A,Y}=P_XP_{Y|X}P_{A|X,Y}$ and $\bar P_{X,A,Y}=P_XP_{Y|X}\bar P_{A|X,Y}$. Recall $p_X$ is the density of $P_X$, $\eta(X)=\Prob(Y=1|X)$, $\rho_{a|y}(X)=\Prob_P(A=a|X,Y=y)$, and $\bar\rho_{a|y}(X)=\Prob_{\bar P}(A=a|X,Y=y)$. Then it suffices to specify $p_X$, $\eta$, $\rho_{a|y}$, and $\bar\rho_{a|y}$, respectively.

\textbf{Construction of $\eta$:}

With the same notation as in the proof of Claim \ref{clm:rho}, we redefine $P_{X,A,Y}$ as follows. 
\[\eta(x)=\left\{\begin{matrix}
    C_\eta-\tilde C_\eta(\frac{1}{7}-x_1)^{\frac{d}{\gamma}},&{\rm if~}x_1\in[0,\frac{1}{7}],\\
    C_\eta+\tilde C_\eta(x_1-\frac{1}{7})^{\frac{d}{\gamma}},&{\rm if~}x_1\in[\frac{1}{7},\frac{1}{7}+C_X],\\
    C_\eta+2\tilde C_\eta C_X^{\frac{d}{\gamma}}-\tilde C_\eta(\frac{2}{7}-x_1)^{\frac{d}{\gamma}},&{\rm if~}x_1\in[\frac{2}{7}-C_X,\frac{2}{7}],\\
    C_\eta+2\tilde C_\eta C_X^{\frac{d}{\gamma}}+\tilde C_\eta(x_1-\frac{2}{7})^{\frac{d}{\gamma}},&{\rm if~}x_1\in[\frac{2}{7},\frac{3}{7}],\\
    \frac{1}{2},&{\rm if~}x_1\in[\frac{4}{7},\frac{5}{7}],\\
    1-C_\eta,&{\rm if~}x_1\in[\frac{6}{7},1],
\end{matrix}\right.\]
where similar to the proof of Claim \ref{clm:rho}, $C_\eta,\tilde C_\eta>0$ are small constants, $C_X$ is also small enough whose value will be specified later, and $\eta$ is interpolated elsewhere such that $\eta\in\HH(\beta_Y,L_Y,\R^d)$.

\textbf{Construction of $\rho_{a|y}$:}

\[\rho_{1|1}(x)-\frac{1}{2}=\left\{\begin{matrix}
    -C_\rho, &{\rm if~}x_1\in[0,\frac{3}{7}],\\
    C_\rho+\frac{1}{2}C_\eta C_\rho,&{\rm if~}[\frac{4}{7},\frac{5}{7}],\\
    C_\rho,&{\rm if~}x_1\in[\frac{6}{7},1],
\end{matrix}\right.\]
where $C_\rho>0$ is small enough and $\rho_{1|1}$ is interpolated elsewhere such that $\rho_{1|1}\in\HH(\beta_A,L_A,\R^d)$. 
So Assumption~\ref{ass:holder_blind} is satisfied. We also define $\rho_{1|0}$ as
\[\rho_{1|0}(x)=\left\{\begin{matrix}
    \frac{1}{4}, &{\rm if~}x_1\in[0,\frac{3}{7}],\\
    \frac{1}{2}+C_\rho+\frac{1}{2}C_\eta C_\rho-\tilde C_\eta(\frac{9}{14}-x_1)^{\frac{d}{\gamma}}, &{\rm if~}x_1\in[\frac{4}{7},\frac{9}{14}],\\
    \frac{1}{2}+C_\rho+\frac{1}{2}C_\eta C_\rho+\tilde C_\eta(x_1-\frac{9}{14})^{\frac{d}{\gamma}}, &{\rm if~}x_1\in[\frac{9}{14},\frac{5}{7}],\\
    \frac{3}{4}, &{\rm if~}x_1\in[\frac{6}{7},1].
\end{matrix}\right.\]
And $\rho_{1|0}$ on $([\frac{3}{7},\frac{4}{7}]\cup[\frac{5}{7},\frac{6}{7}])\times[0,1]^{d-1}$ is defined such that $\rho_{1|0}$ is $\beta_Y$-H\"older smooth. Here $C_\eta,\tilde C_\eta$ are small constants but $C_X,C_\rho$ may become small when $\alpha$ varies.

\textbf{Construction of $p_X$:}

Denote $\mathcal{B}=\{x:\|x_{-1}\|_1\le|x_1-\frac{1}{7}|,x_1\in[0,\frac{1}{7}+C_X]\}$, we define the density of $X$ as
\[p_X(x)=\left\{\begin{matrix}
    \frac{1}{6{\rm Leb}(\mathcal{B})},&{\rm if~}\|x_{-1}\|_1\le|x_1-\frac{1}{7}|,x_1\in[0,\frac{1}{7}+C_X],\\
    \frac{1}{6{\rm Leb}(\mathcal{B})},&{\rm if~}\|x_{-1}\|_1\le|x_1-\frac{2}{7}|,x_1\in[\frac{2}{7}-C_X,\frac{3}{7}],\\
    \frac{\mu}{3{\rm Leb}(B_1(0,\frac{1}{28}))},&{\rm if~}x\in B_1(\frac{17}{28}e_1,\frac{1}{28}),\\
    \frac{1-\mu}{3{\rm Leb}(B_1(0,\frac{1}{28}))},&{\rm if~}x\in B_1(\frac{19}{28}e_1,\frac{1}{28}),\\
    \frac{7}{3},&{\rm if~}x_1\in[\frac{6}{7},1],
\end{matrix}\right.\]
with 
\[\mu=\frac{\frac{\frac{3}{4}-\frac{3}{2}C_\rho-\frac{1}{2}C_\eta+\frac{3}{4}C_\rho C_\eta}{\frac{1}{4}-(\frac{1}{2}-\frac{7}{12}C_\eta)C_\rho+(\frac{1}{6}+\frac{1}{3}C_\rho)\tilde C_\eta C_X^{\frac{d}{\gamma}}}-\frac{\frac{1}{2}+C_\rho-\frac{1}{2}C_\eta-C_\rho C_\eta}{\frac{1}{4}+(\frac{1}{2}-\frac{7}{12}C_\eta)C_\rho+(\frac{1}{6}-\frac{1}{3}C_\rho)\tilde C_\eta C_X^{\frac{d}{\gamma}}}}{\frac{\frac{1}{4}+\frac{1}{2}C_\rho+\frac{1}{4}C_\rho C_\eta}{\frac{1}{4}+(\frac{1}{2}-\frac{7}{12}C_\eta)C_\rho+(\frac{1}{6}-\frac{1}{3}C_\rho)\tilde C_\eta C_X^{\frac{d}{\gamma}}}+\frac{\frac{1}{4}-\frac{1}{2}C_\rho-\frac{1}{4}C_\rho C_\eta}{\frac{1}{4}-(\frac{1}{2}-\frac{7}{12}C_\eta)C_\rho+(\frac{1}{6}+\frac{1}{3}C_\rho)\tilde C_\eta C_X^{\frac{d}{\gamma}}}}.\]

\textbf{Verification of $P_{X,A,Y}\in\mathscr{P}$:}

\textbf{Assumptions \ref{ass:density_aware} and \ref{ass:density_blind}:}

It is straightforward to verify that Assumptions~\ref{ass:density_aware} and \ref{ass:density_blind} are satisfied.

\textbf{Assumption \ref{ass:observe}:}

For the specified distribution, we have
\[p_Y=\E\eta(X)=\frac{1}{2}+\frac{1}{3}\tilde C_\eta C_X^{\frac{d}{\gamma}},\quad p_{1,1}=\E\rho_{1|1}(X)\eta(X)=\frac{1}{4}+\bigg(\frac{1}{2}-\frac{7}{12}C_\eta\bigg)C_\rho+\bigg(\frac{1}{6}-\frac{1}{3}C_\rho\bigg)\tilde C_\eta C_X^{\frac{d}{\gamma}}.\]
Then Assumption~\ref{ass:observe} is satisfied if $\tilde C_\eta, C_\rho$ and $C_X$ are small enough.

\textbf{Group-Blind Assumptions:}

1) At first, we verify the group-blind assumptions. On the support of $p_X$, $\phi$ equals
\begin{equation}
    \begin{aligned}
        &p_{1,1}(p_Y-p_{1,1})\phi(x)\\
        =&\left\{\begin{matrix}
            -\{C_\eta-\tilde C_\eta(\frac{1}{7}-x_1)^{\frac{d}{\gamma}}\}(1-\frac{7}{12}C_\eta)C_\rho,&{\rm if~}\|x_{-1}\|_1\le\frac{1}{7}-x_1,x_1\in[0,\frac{1}{7}],\\
            -\{C_\eta+\tilde C_\eta(x_1-\frac{1}{7})^{\frac{d}{\gamma}}\}(1-\frac{7}{12}C_\eta)C_\rho,&{\rm if~}\|x_{-1}\|_1\le x_1-\frac{1}{7},x_1\in[\frac{1}{7},\frac{1}{7}+C_X],\\
            -\{C_\eta+2\tilde C_\eta C_X^{\frac{d}{\gamma}}-\tilde C_\eta(\frac{2}{7}-x_1)^{\frac{d}{\gamma}}\}(1-\frac{7}{12}C_\eta)C_\rho,&{\rm if~}\|x_{-1}\|_1\le\frac{2}{7}-x_1,x_1\in[\frac{2}{7}-C_X,\frac{2}{7}],\\
            -\{C_\eta+2\tilde C_\eta C_X^{\frac{d}{\gamma}}+\tilde C_\eta(x_1-\frac{2}{7})^{\frac{d}{\gamma}}\}(1-\frac{7}{12}C_\eta)C_\rho,&{\rm if~}\|x_{-1}\|_1\le x_1-\frac{2}{7},x_1\in[\frac{2}{7},\frac{3}{7}],\\
            \{\frac{5}{12}C_\eta+(\frac{1}{3}+\frac{1}{12}C_\eta)\tilde C_\eta C_X^{\frac{d}{\gamma}}\} C_\rho,&{\rm if~}x_1\in[\frac{4}{7},\frac{5}{7}],\\
            (1-C_\eta)(\frac{7}{12}C_\eta+\frac{2}{3}\tilde C_\eta C_X^{\frac{d}{\gamma}})C_\rho,&{\rm if~}x_1\in[\frac{6}{7},1].
        \end{matrix}\right.
    \end{aligned}
\end{equation}
Since $\E\phi(X)\1(2\eta(X)>1)>0$, we know $\lambda^*_\alpha\ge 0$. Set
\[\tilde\lambda=\frac{p_{1,1}(p_Y-p_{1,1})(1-2C_\eta)}{(1-\frac{7}{12}C_\eta)C_\eta C_\rho},\]
if $C_\rho$ is chosen such that
\[\E\phi(X)\1(2\eta(X)-1>\tilde\lambda\phi(X))=\alpha,\]
then $\lambda^*_\alpha=\tilde\lambda$. In this case, $g^*_\alpha=2\eta-1-\lambda^*_\alpha\phi$ equals
\[g^*_\alpha(x)=\left\{\begin{matrix}
    -\frac{\tilde C_\eta(\frac{1}{7}-x_1)^{\frac{d}{\gamma}}}{C_\eta},&{\rm if~}\|x_{-1}\|_1\le\frac{1}{7}-x_1,x_1\in[0,\frac{1}{7}],\\
    \frac{\tilde C_\eta(x_1-\frac{1}{7})^{\frac{d}{\gamma}}}{C_\eta},&{\rm if~}\|x_{-1}\|_1\le x_1-\frac{1}{7},x_1\in[\frac{1}{7},\frac{1}{7}+C_X],\\
    \frac{\tilde C_\eta\{2C_X^{\frac{d}{\gamma}}-(\frac{2}{7}-x_1)^{\frac{d}{\gamma}}\}}{C_\eta},&{\rm if~}\|x_{-1}\|_1\le\frac{2}{7}-x_1,x_1\in[\frac{2}{7}-C_X,\frac{2}{7}],\\
    \frac{\tilde C_\eta\{2C_X^{\frac{d}{\gamma}}+(x_1-\frac{2}{7})^{\frac{d}{\gamma}}\}}{C_\eta},&{\rm if~}\|x_{-1}\|_1\le x_1-\frac{2}{7},x_1\in[\frac{2}{7},\frac{3}{7}],\\
    -\frac{(1-2C_\eta)\{\frac{5}{12}C_\eta+(\frac{1}{3}+\frac{1}{12}C_\eta)\tilde C_\eta C_X^{\frac{d}{\gamma}}\}}{(1-\frac{7}{12}C_\eta)C_\eta},&{\rm if~}x_1\in[\frac{4}{7},\frac{5}{7}],\\
    \frac{(1-2C_\eta)\{\frac{5}{12}C_\eta-\frac{2}{3}(1-C_\eta)\tilde C_\eta C_X^{\frac{d}{\gamma}}\}}{(1-\frac{7}{12}C_\eta)C_\eta},&{\rm if~}x_1\in[\frac{6}{7},1].
\end{matrix}\right.\]

Now we set $C_\rho$ such that 
\begin{equation}\label{eq:lower_alpha_c_rho}
    \begin{aligned}
    \alpha=&\E\phi(X)\1(g^*_\alpha(X)>0)\\
    =&\frac{\{\frac{1}{36}C_\eta-\frac{7}{72}C_\eta^2-\frac{\tilde C_\eta C_B(\frac{1}{6}-\frac{7}{72}C_\eta)}{1+(7C_X)^d}-(\frac{1}{9}+\frac{1}{36}C_\eta)\tilde C_\eta C_X^{\frac{d}{\gamma}}-\frac{(\frac{1}{6}-\frac{7}{72}C_\eta)C_\eta(7C_X)^d}{1+(7C_X)^d}\}C_\rho}{(\frac{1}{4}+\frac{1}{6}\tilde C_\eta C_X^{\frac{d}{\gamma}})^2-(\frac{1}{2}-\frac{7}{12}C_\eta-\frac{1}{3}\tilde C_\eta C_X^{\frac{d}{\gamma}})^2C_\rho^2}\\
    \overset{\triangle}{=}&\frac{C_NC_\rho}{(\frac{1}{4}+\frac{1}{6}\tilde C_\eta C_X^{\frac{d}{\gamma}})^2-(\frac{1}{2}-\frac{7}{12}C_\eta-\frac{1}{3}\tilde C_\eta C_X^{\frac{d}{\gamma}})^2C_\rho^2},
\end{aligned}
\end{equation}

it follows
\[C_\rho\asymp\alpha, \quad\lambda^*_\alpha\asymp\frac{1}{\alpha}.\]
Similar to the proof of Claim \ref{clm:rho}, we can verify that Assumptions~\ref{ass:margin}, \ref{ass:ratio_poly} and \ref{ass:ratio_balance} are satisfied.

\textbf{Construction of $\bar\rho_{a|y}$:}

Define another distribution $\bar P_{X,A,Y}=P_XP_{Y|X}\bar P_{A|X,Y}$ with $\bar\rho_{1|1},\bar\rho_{1|0}$ are defined by replacing $C_\rho$ in $\rho_{1|1},\rho_{1|0}$ by $\bar C_\rho$, where
\[\bar C_\rho=\bigg\{1-c\bigg(\frac{1}{\alpha\sqrt{N}}\wedge1\bigg)\bigg\}C_\rho.\]
Similarly, we define $\bar p_{1,1}$ and $\bar\phi$ accordingly. Then we choose $C_X$ such that
\[\bar\lambda_\alpha^*=\frac{\bar p_{1,1}(p_Y-\bar p_{1,1})(1-2C_\eta-4\tilde C_\eta C_X^{\frac{d}{\gamma}})}{(C_\eta+2\tilde C_\eta C_X^{\frac{d}{\gamma}})(1-\frac{7}{12}C_\eta)\bar C_\rho}.\]
In this case, we have $\bar g_\alpha^*=2\eta-1-\bar\lambda_\alpha^*\bar\phi$ equals
\[\bar g^*_\alpha(x)=\left\{\begin{matrix}
    -\frac{\tilde C_\eta\{2C_X^{\frac{d}{\gamma}}+(\frac{1}{7}-x_1)^{\frac{d}{\gamma}}\}}{C_\eta+2\tilde C_\eta C_X^{\frac{d}{\gamma}}},&{\rm if~}\|x_{-1}\|_1\le\frac{1}{7}-x_1,x_1\in[0,\frac{1}{7}],\\
    -\frac{\tilde C_\eta\{2C_X^{\frac{d}{\gamma}}-(x_1-\frac{1}{7})^{\frac{d}{\gamma}}\}}{C_\eta+2\tilde C_\eta C_X^{\frac{d}{\gamma}}},&{\rm if~}\|x_{-1}\|_1\le x_1-\frac{1}{7},x_1\in[\frac{1}{7},\frac{1}{7}+C_X],\\
    -\frac{\tilde C_\eta(\frac{2}{7}-x_1)^{\frac{d}{\gamma}}}{C_\eta+2\tilde C_\eta C_X^{\frac{d}{\gamma}}},&{\rm if~}\|x_{-1}\|_1\le\frac{2}{7}-x_1,x_1\in[\frac{2}{7}-C_X,\frac{2}{7}],\\
    \frac{\tilde C_\eta(x_1-\frac{2}{7})^{\frac{d}{\gamma}}}{C_\eta+2\tilde C_\eta C_X^{\frac{d}{\gamma}}},&{\rm if~}\|x_{-1}\|_1\le x_1-\frac{2}{7},x_1\in[\frac{2}{7},\frac{3}{7}],\\
    -\frac{(1-2C_\eta-4\tilde C_\eta C_X^{\frac{d}{\gamma}})\{\frac{5}{12}C_\eta+(\frac{1}{3}+\frac{1}{12}C_\eta)\tilde C_\eta C_X^{\frac{d}{\gamma}}\}}{(1-\frac{7}{12}C_\eta)(C_\eta+2\tilde C_\eta C_X^{\frac{d}{\gamma}})},&{\rm if~}x_1\in[\frac{4}{7},\frac{5}{7}],\\
    \frac{\frac{5}{12}C_\eta-\frac{5}{6}C_\eta^2+[\frac{4}{3}-\frac{5}{6}C_\eta-\frac{4}{3}C_\eta^2+\frac{8}{3}(1-C_\eta)\tilde C_\eta C_X^{\frac{d}{\gamma}}]\tilde C_\eta C_X^{\frac{d}{\gamma}}}{(1-\frac{7}{12}C_\eta)(C_\eta+2\tilde C_\eta C_X^{\frac{d}{\gamma}})},&{\rm if~}x_1\in[\frac{6}{7},1].
\end{matrix}\right.\]
Then $C_X$ should satisfies
\begin{align*}
    \alpha=\E\bar\phi(X)\1(\bar g(X)>0)
    =\frac{\{C_N+\frac{(\frac{1}{3}-\frac{7}{36}C_\eta)(C_\eta+\tilde C_\eta C_X^{\frac{d}{\gamma}})(7C_X)^d}{1+(7C_X)^d}\}\bar C_\rho}{(\frac{1}{4}+\frac{1}{6}\tilde C_\eta C_X^{\frac{d}{\gamma}})^2-(\frac{1}{2}-\frac{7}{12}C_\eta-\frac{1}{3}\tilde C_\eta C_X^{\frac{d}{\gamma}})^2\bar C_\rho^2}.
\end{align*}
Comparing with Equation \eqref{eq:lower_alpha_c_rho}, we get
\[C_X^d\asymp\frac{1}{\alpha\sqrt{N}}\wedge1.\]
Therefore we have
\[\bar\lambda_\alpha^*\asymp\frac{1}{\alpha}.\]

\textbf{Verification of $\bar P_{X,A,Y}\in\mathscr{P}$:}

Similarly, we can verify that $\bar P_{X,A,Y}$ satisfies Assumptions~\ref{ass:margin}, \ref{ass:ratio_poly}, \ref{ass:ratio_balance}.

\textbf{Group-Aware Assumptions:}

2) Then we verify the group-aware assumptions. Similar to the analysis of the errors of $\rho_{1|1}$ and $\eta$, we can show $\eta_P^{\rm aware}(\cdot,a),\eta_{\bar P}^{\rm aware}(\cdot,a)$ are both $\beta_Y$-H\"older smooth, where $\eta_P^{\rm aware}(X,A)=\Prob_P(Y=1|X,A)$, and $\U_P\big(\1(2\eta_P^{\rm aware}(X,A)>1)\big)=0$,
\begin{align*}
    \U_{\bar P}\big(\1(2\eta_{\bar P}^{\rm aware}(X,A)>1)\big)=&\frac{(\frac{1}{2}+\bar C_\rho)(1-C_\eta)}{\bar p_{1,1}}-\frac{\frac{3}{4}-\frac{3}{2}\bar C_\rho-\frac{1}{2}C_\eta+\frac{3}{4}\bar C_\rho C_\eta}{\bar p_{1,2}}\\
    &+\mu\bigg(\frac{\frac{1}{2}+\bar C_\rho+\frac{1}{2}\bar C_\rho C_\eta}{2\bar p_{1,1}}+\frac{\frac{1}{2}-\bar C_\rho-\frac{1}{2}\bar C_\rho C_\eta}{2\bar p_{1,2}}\bigg)\\
    \lesssim&c\bigg(\frac{1}{\sqrt{N}}\wedge\alpha\bigg).
\end{align*}
As long as we set the constant $c$ in $\bar C_\rho$ to be small enough, we have $\U_{\bar P}\big(\1(2\eta^{\rm aware}(X,A)>1)\big)<\alpha$. Therefore $\bar g^{*{\rm aware}}_\alpha=2\eta^{\rm aware}_{\bar P}-1$. So $P_{X,A,Y},\bar P_{X,A,Y}$ satisfy the group-aware Assumptions~\ref{ass:margin}, \ref{ass:ratio_poly}, \ref{ass:ratio_balance}.

\textbf{Derivation of the Lower Bound:}

Now we derive the minimax lower bound. For any $\A\in\mathscr{A}$, $\hat f=\A(\D_{\rm all})$, we have
\[\Prob_{\D_{\rm all}\sim{P_{X,A,Y}^{\otimes N}}}\big(\U_{\rm EOO,P}(\hat f)\le\alpha\big)\ge 1-\delta.\]
Note that
\[\bar\phi=\frac{p_{1,1}(p_Y-p_{1,1})}{\bar p_{1,1}(p_Y-\bar p_{1,1})}\bigg\{1-c\bigg(\frac{1}{\alpha\sqrt{N}}\wedge1\bigg)\bigg\}\phi.\]
Under the event $\U_{\rm EOO,P}(\hat f)\le\alpha$, if $\E\phi(X)\hat f(X)\le 0$, then $\E\bar\phi(X)\hat f(X)\le 0$. Otherwise, if $\E\phi(X)\hat f(X)>0$, it follows from $0<\E\phi(X)\hat f(X)\le\alpha$ that
\[\E\bar\phi(X)\hat f(X)\le\alpha-c\bigg(\frac{1}{\sqrt{N}}\wedge \alpha\bigg).\]
Denote the Hellinger distance ${\rm HL}$ between any two distributions $P,Q$ as
\[{\rm HL}(P,Q)=\bigg(\int\big(\sqrt{dP}-\sqrt{dQ}\big)^2\bigg)^{\frac{1}{2}},\]
by Lemma 15.3 and Equation 15.12b in \cite{wainwright2019high}, we can control ${\rm TV}(P_{X,A,Y}^{\otimes N},\bar P_{X,A,Y}^{\otimes N})$ as
\begin{align*}
    {\rm TV}(P_{X,A,Y}^{\otimes N},\bar P_{X,A,Y}^{\otimes N})
    \le&{\rm HL}(P_{X,A,Y}^{\otimes N},\bar P_{X,A,Y}^{\otimes N})\\
    \le&\sqrt{N}{\rm HL}(P_{X,A,Y},\bar P_{X,A,Y})\\
    =&\sqrt{N}\bigg(\int\bigg(\sqrt{\rho_{1|1}(x)}-\sqrt{\bar\rho_{1|1}(x)}\bigg)^2\eta(x)p_X(x)dx\\
    &+\int\bigg(\sqrt{1-\rho_{1|1}(x)}-\sqrt{1-\bar\rho_{1|1}(x)}\bigg)^2\eta(x)p_X(x)dx\\
    &+\int\bigg(\sqrt{\rho_{1|0}(x)}-\sqrt{\bar\rho_{1|0}(x)}\bigg)^2(1-\eta(x))p_X(x)dx\\
    &+\int\bigg(\sqrt{1-\rho_{1|0}(x)}-\sqrt{1-\bar\rho_{1|0}(x)}\bigg)^2(1-\eta(x))p_X(x)dx\bigg)^{\frac{1}{2}}\\
    \lesssim&\sqrt{N}\bigg(\frac{1}{\sqrt{N}}\wedge\alpha\bigg)\\
    \lesssim& 1.
\end{align*}
Then we have
\begin{align*}
    &\Prob_{\D_{\rm all}\sim\bar P^{\otimes N}_{X,A,Y}}\bigg(T_2(\hat f)\ge c|\bar\lambda^*_\alpha|\bigg(\frac{1}{\sqrt{N}}\wedge\alpha\bigg)\bigg)\\
    =&\Prob_{\D_{\rm all}\sim\bar P^{\otimes N}_{X,A,Y}}\bigg(\E\bar\phi(X)\hat f(X)\le \alpha- c\bigg(\frac{1}{\sqrt{N}}\wedge\alpha\bigg)\bigg)\\
    \ge&\Prob_{\D_{\rm all}\sim P^{\otimes N}_{X,A,Y}}\bigg(\E\bar\phi(X)\hat f(X)\le \alpha- c\bigg(\frac{1}{\sqrt{N}}\wedge\alpha\bigg)\bigg)-{\rm TV}(P^{\otimes N}_{X,A,Y},\bar P^{\otimes N}_{X,A,Y})\\
    \ge&\Prob_{\D_{\rm all}\sim P^{\otimes N}_{X,A,Y}}\big(\E\phi(X)\hat f(X)\le \alpha\big)-{\rm TV}(P^{\otimes N}_{X,A,Y},\bar P^{\otimes N}_{X,A,Y})\\
    \ge&1-\delta-{\rm TV}(P^{\otimes N}_{X,A,Y},\bar P^{\otimes N}_{X,A,Y})\\
    \ge&c-\delta.
\end{align*}

\end{proof}

\section{Proof of Theorem~\ref{thm:upper_multi}}

\begin{proof}[Proof of Theorem~\ref{thm:upper_multi}]
    \textbf{Existence of $\hat\lambda^G_\alpha$}:

    At first, we show $\hat f_{\hat\lambda_\alpha}^G$ is well-defined and $\alpha$-fair. Denote the event $E$ as
    \[E=\{\sup_{\lambda\in\R^{\tilde K}}|\hat\U(\hat f^G_\lambda)-\U(\hat f^G_\lambda)|\le\epsilon_\alpha\},\]
    then we know $\Prob(E^c)\le\delta_{\rm post}$. Recall that $U(0)=\U(\1(2\eta^G>1))$. Now we separate the proof into two cases depending on $D_0=U(0)-\alpha$. If $D_0\le-\tilde\epsilon_\eta^G-2\epsilon_\alpha$, it is guaranteed that $\hat\lambda^G_\alpha=0$ leads to a feasible and $\alpha$-fair classifier. If $D_0>-\tilde\epsilon_\eta^G-2\epsilon_\alpha$, we have to carefully choose $\tilde\alpha<\alpha$ such that $\hat\lambda^G_\alpha=\lambda^{*G}_{\tilde\alpha}$ is feasible and thus $\alpha$-fair.

    \textbf{Case (1)}: If $U(0)-\alpha\le-\tilde\epsilon_\eta^G-2\epsilon_\alpha$, we know $\lambda^{*G}_\alpha=0$, $f^{*G}_\alpha=\1(2\eta^G>1)$. Under the event $E$, we have $\hat f^G_0=\1(2\hat\eta^G>1)$ satisfies
    \begin{align*}
        &|\U(\hat f^G_0)-\U(f^{*G}_\alpha)|\\
        =&|\|\E\Phi^G(X,A)\1\big(2\hat\eta^G(X,A)>1\big)\|_\infty-\|\E\Phi^G(X,A)\1\big(2\eta^G(X,A)>1\big)\|_\infty|\\
        \le&\|\E\Phi^G(X,A)\{\1(2\hat\eta^G(X,A)>1)-\1(2\eta^G(X,A)>1)\}\|_\infty\\
        \le&\max_{k\in[\tilde K]}\E|\phi^G_k(X,A)|\1\big(|2\eta^G(X,A)-1|\le2\epsilon_\eta\big)\\
        =&\tilde\epsilon_\eta^G.
    \end{align*}
    Then we have
    \begin{align*}
        \hat\U(\hat f^G_0)\le\U(\hat f^G_0)+\epsilon_\alpha\le\U(f^{*G}_\alpha)+\tilde\epsilon_\eta^G+\epsilon_\alpha\le\alpha-\epsilon_\alpha,
    \end{align*}
    so $\hat f^G_0$ is feasible. 

    \textbf{Case (2)}: If $U(0)-\alpha>-\tilde\epsilon_\eta^G-2\epsilon_\alpha$, under event $E$, for our choice of $\tilde\epsilon_\alpha$ in Equation~\eqref{eq:tilde_epsilon_alpha}, it follows
    \begin{align*}
        \hat\U(\hat f^G_{\lambda^*_{\tilde\alpha}})\le&\U(\hat f^G_{\lambda^*_{\tilde\alpha}})+\epsilon_\alpha\\
        \le&\U(f^{*G}_{\tilde\alpha})+\|\E\Phi^G(X,A)\big\{\hat f^G_{\lambda^*_{\tilde\alpha}}(X,A)-f^{*G}_{\tilde\alpha}(X,A)\big\}\|_\infty+\epsilon_\alpha\\
        \le&\alpha-\tilde\epsilon_\alpha+\epsilon_\alpha+\|\E\Phi^G(X,A)\1\big(0\ge g^{*G}_{\tilde\alpha}(X,A)>2\big(\eta^G(X,A)-\hat\eta^G(X,A)\big)\\
        &\qquad-\lambda^{*G\top}_{\tilde\alpha}\big(\Phi^G(X,A)-\hat\Phi^G(X,A)\big)\big)\|_\infty+\|\E\Phi^G(X,A)\1\big(0<g^{*G}_{\tilde\alpha}(X,A)\le\\
        &\qquad2\big(\eta^G(X,A)-\hat\eta^G(X,A)\big)-\lambda^{*G\top}_{\tilde\alpha}\big(\Phi^G(X,A)-\hat\Phi^G(X,A)\big)\big)\|_\infty\\
        \le&\alpha-\tilde\epsilon_\alpha+\epsilon_\alpha+\max_{k\in[\tilde K]}\E|\phi^G_k(X,A)|\1\big(|g^{*G}_{\tilde\alpha}(X,A)|\le 2\epsilon_\eta+\|\lambda^{*G}_{\tilde\alpha}\|_1\epsilon_\phi\big)\\
        =&\alpha-\tilde\epsilon_\alpha+\epsilon_\alpha+\tilde\epsilon_{g,\tilde\alpha}^G\\
        \overset{\text{Equation~\eqref{eq:tilde_epsilon_alpha}}}{\le}&\alpha-\epsilon_\alpha.
    \end{align*}
    Therefore $\hat f^G_{\lambda^*_{\tilde\alpha}}$ is feasible.

    \textbf{Fairness constraint}:

    For any $\hat\lambda$ such that $\hat\U(\hat f^G_{\hat\lambda})\le\alpha-\epsilon_\alpha$, under $E$, we have
    \[\U(\hat f^G_{\hat\lambda})\le\hat\U(\hat f^G_{\hat\lambda})+\sup_{\lambda\in\R^{\tilde K}}|\hat\U(\hat f^G_\lambda)-\U(\hat f^G_\lambda)|\le\alpha.\]

    \textbf{Excess risk}:

    The analysis of excess risk follows from the theory of empirical risk minimization \citep{massart2006risk}. Denote $Z=(X,A,Y)\sim P_{X,A,Y}$ and $L:\{0,1\}^{\R^d\times[K]}\times\R^d\times[K]\times\{0,1\}$ to be the 0-1 loss function
    \[L(f^G,Z)=\1(Y\ne Y_{f^G}),\quad \Prob(Y_{f^G}=1|X,A)=f^G(X,A),\]
    \[\EE (\lambda)=\risk(\hat f^G_\lambda)-\risk(f^{*G}_{\tilde\alpha}),\quad\EE _{\app}=\risk(\hat f^G_{\lambda^*_{\tilde\alpha}})-\risk(f^{*G}_{\tilde\alpha}).\]
    We start with the following inequality based on Proposition 1 in \cite{tsybakov2004optimal}. For any $\lambda$ that is feasible for Algorithm~\ref{alg:multi_unify}, under the event $E$, we have
    \begin{align*}
        \EE (\lambda)=&\E\big(2\eta^G(X,A)-1\big)\big(f^{*G}_{\tilde\alpha}(X,A)-\hat f^G_\lambda(X,A)\big)\\
        =&\E\big|2\eta^G(X,A)-1-\lambda^{*G\top}_{\tilde\alpha}\Phi^G(X,A)\big|\big|f^{*G}_{\tilde\alpha}(X,A)-\hat f^G_\lambda(X,A)\big|\\
        &+\lambda^{*G\top}_{\tilde\alpha}\E\Phi^G(X,A)\big(f^{*G}_{\tilde\alpha}(X,A)-\hat f^G_\lambda(X,A)\big)\\
        \ge&c\big\{\E\big|f^{*G}_{\tilde\alpha}(X,A)-\hat f^G_\lambda(X,A)\big|\big\}^{\frac{1+\tilde\gamma}{\tilde\gamma}}+\|\lambda^{*G}_{\tilde\alpha}\|_1\tilde\alpha-\|\lambda^{*G}_{\tilde\alpha}\|_1\alpha\\
        \ge& c\big\{{\Var}\big(L(f^{*G}_{\tilde\alpha},Z)-L(\hat f^G_\lambda,Z)\big)\big\}^{\frac{1+\tilde\gamma}{\tilde\gamma}}-\|\lambda^{*G}_{\tilde\alpha}\|_1\tilde\epsilon_\alpha,
    \end{align*}
    so we get
    \[{\Var}\big(L(f^{*G}_{\tilde\alpha},Z)-L(\hat f^G_\lambda,Z)\big)\lesssim\big\{\EE (\lambda)+\|\lambda^{*G}_{\tilde\alpha}\|_1\tilde\epsilon_\alpha\big\}^{\frac{\tilde\gamma}{1+\tilde\gamma}},\]
    it follows
    \[\EE(\lambda)+\|\lambda^{*G}_{\tilde\alpha}\|_1\tilde\epsilon_\alpha\ge 0,\quad \EE_{\app}+\|\lambda^{*G}_{\tilde\alpha}\|_1\tilde\epsilon_\alpha\ge 0.\]
    Denote $\hat \E$ to be the sample average based on data $\D$. For any $t>0$, we denote
    \[V_t=\sup_{\lambda\in\R^{\tilde K}}\frac{(\E-\hat \E)\big(L(\hat f^G_\lambda,Z)-L(\hat f^G_{\lambda^*_{\tilde\alpha}},Z)\big)}{\EE (\lambda)+\EE _{\app}+2\|\lambda^{*G}_{\tilde\alpha}\|_1\tilde\epsilon_\alpha+t},\]
    then we can control the excess risk as
    \begin{align*}
        \EE (\hat\lambda)=&\risk(\hat f^G_{\hat\lambda})-\risk(\hat f^G_{\lambda^*_{\tilde\alpha}})+\risk(\hat f^G_{\lambda^*_{\tilde\alpha}})-\risk(f^{*G}_{\tilde\alpha})\\
        =&\hat\E\big(L(\hat f^G_{\hat\lambda},Z)-L(\hat f^G_{\lambda^*_{\tilde\alpha}},Z)\big)+\big(\E-\hat\E\big)\big(L(\hat f^G_{\hat\lambda},Z)-L(\hat f^G_{\lambda^*_{\tilde\alpha}},Z)\big)+\EE _{\app}\\
        \le& V_t\big\{\EE (\hat\lambda)+\EE _{\app}+2\|\lambda^{*G}_{\tilde\alpha}\|_1\tilde\epsilon_\alpha+t\big\}+\EE _{\app}.
    \end{align*}
    Under the event $\{V_t\le\frac{1}{2}\}$, we get
    \[\EE(\hat\lambda)\le3\EE_{\app}+2\|\lambda^{*G}_{\tilde\alpha}\|_1\tilde\epsilon_\alpha+t.\]
    Then it suffices to control $V_t$ and $\EE_{\app}$. We start with controlling $V_t$ using Talagrand's concentration inequality \citep{Boucheron2013concentration}. 
    Note that
    \begin{align*}
        &\sup_{\lambda\in\R^{\tilde K}}\Var\bigg(\frac{L(\hat f^G_\lambda,Z)-L(\hat f^G_{\lambda^{*}_{\tilde\alpha}},Z)}{\EE(\lambda)+\EE_{\app}+2\|\lambda^{*G}_{\tilde\alpha}\|_1\tilde\epsilon_\alpha+t}\bigg)\\
        \lesssim&\sup_{\lambda\in\R^{\tilde K}}\frac{\big\{\EE(\lambda)+\EE_{\app}+2\|\lambda^{*G}_{\tilde\alpha}\|_1\tilde\epsilon_\alpha\big\}^{\frac{\tilde\gamma}{1+\tilde\gamma}}}{\big(\EE(\lambda)+\EE_{\app}+2\|\lambda^{*G}_{\tilde\alpha}\|_1\tilde\epsilon_\alpha+t\big)^2}\\
        \le&\sup_{\xi\ge 0}\frac{\xi^{\frac{\tilde\gamma}{1+\tilde\gamma}}}{(\xi+t)^2}\\
        \lesssim& t^{-\frac{2+\tilde\gamma}{1+\tilde\gamma}},
    \end{align*}
    \[\sup_{\lambda\in\R^{\tilde K}}\bigg|\frac{L(\hat f^G_\lambda,Z)-L(\hat f^G_{\lambda^*_{\tilde\alpha}},Z)}{\EE(\lambda)+\EE_{\app}+2\|\lambda^{*G}_{\tilde\alpha}\|_1\tilde\epsilon_\alpha+t}\bigg|\le\frac{1}{t},\]
    then Talagrand's concentration inequality \citep{Boucheron2013concentration} implies that with probability at least $1-\delta_{\rm post}$, we have
    \[V_t-\E V_t\lesssim\sqrt{\frac{t^{-\frac{2+\tilde\gamma}{1+\tilde\gamma}}+\frac{1}{t}\E V_t}{n}\log\frac{1}{\delta_{\rm post}}}+\frac{\log\frac{1}{\delta_{\rm post}}}{nt}.\]
    Then it remains to control $\E V_t$. We proceed using the peeling techniques. Denote $\Lambda_j=\{\lambda\in\R^{\tilde K}:\EE(\lambda)+\EE_{\app}+2\|\lambda^{*G}_{\tilde\alpha}\|_1\tilde\epsilon_\alpha\in[2^{j-1}t,~ 2^jt)]\}$ for $j\in\N_+$ and $\Lambda_0=\{\lambda\in\R^{\tilde K}:\EE(\lambda)+\EE_{\app}+2\|\lambda^{*G}_{\tilde\alpha}\|_1\tilde\epsilon_\alpha<t\}$, we know
    \[\sup_{\lambda\in\Lambda_j}\Var\big(L(\hat f^G_\lambda,Z)-L(\hat f^G_{\lambda^*_{\tilde\alpha}},Z)\big)\lesssim \big(2^jt\big)^{\frac{\tilde\gamma}{1+\tilde\gamma}}\wedge 1,\]
    then Theorem 13.7 in \cite{Boucheron2013concentration} implies
    \begin{align*}
        \E V_t\le&\sum_{j\in\N}\E\sup_{\lambda\in\Lambda_j}\frac{(\E-\hat \E)\big(L(\hat f^G_\lambda,Z)-L(\hat f^G_{\lambda^*_{\tilde\alpha}},Z)\big)}{\EE (\lambda)+\EE _{\app}+2\|\lambda^{*G}_{\tilde\alpha}\|_1\tilde\epsilon_\alpha+t}\\
        \lesssim&\frac{t^{\frac{\tilde\gamma}{2+2\tilde\gamma}}}{t}\sqrt{\frac{\tilde K}{n}\log \frac{e}{t^{\frac{\tilde\gamma}{2+2\tilde\gamma}}\wedge 1}}+\sum_{j\in\N_+}\frac{(2^jt)^{\frac{\tilde\gamma}{2+2\tilde\gamma}}}{2^{j-1}t+t}\sqrt{\frac{\tilde K}{n}\log\frac{e}{(2^jt)^{\frac{\tilde\gamma}{2+2\tilde\gamma}}\wedge 1}}\\
        \lesssim&t^{-\frac{2+\tilde\gamma}{2+2\tilde\gamma}}\sqrt{\frac{\tilde K}{n}\log \frac{e}{t\wedge 1}}.
    \end{align*}
    Taking $t\asymp\big(\frac{\tilde K\log n+\log\frac{1}{\delta_{\rm post}}}{n}\big)^{\frac{1+\tilde\gamma}{2+\tilde\gamma}}$, we get $V_t\le\frac{1}{2}$, and thus
    \[\EE(\hat\lambda)\lesssim \EE_{\app}+2\|\lambda^{*G}_{\tilde\alpha}\|_1\tilde\epsilon_\alpha+\bigg(\frac{\tilde K\log n+\log\frac{1}{\delta_{\rm post}}}{n}\bigg)^{\frac{1+\tilde\gamma}{2+\tilde\gamma}}.\]
    Then for $\EE_{\app}$, we have
    \begin{align*}
        &\EE_{\app}\\
        =&\E(2\eta^G(X,A)-1)\big(f^{*G}_{\tilde\alpha}(X,A)-\hat f^G_{\lambda^*_{\tilde\alpha}}(X,A)\big)\\
        =&\underbrace{\E|g^{*G}_{\tilde\alpha}(X,A)||f^{*G}_{\tilde\alpha}(X,A)-\hat f^G_{\lambda^*_{\tilde\alpha}}(X,A)|}_{T_1}+\underbrace{\lambda^{*G\top}_{\tilde\alpha}\E\Phi^G(X,A)\big(f^{*G}_{\tilde\alpha}(X,A)-\hat f^G_{\lambda^*_{\tilde\alpha}}(X,A)\big)}_{T_2}.
    \end{align*}
    We can control $T_1$ as
    \begin{align*}
        &T_1\\
        =&\E|g^{*G}_{\tilde\alpha}(X,A)|\1\big(0<g^{*G}_{\tilde\alpha}(X,A)\le 2\big(\eta^G(X,A)-\hat\eta^G(X,A)\big)-\lambda^{*G\top}_{\tilde\alpha}\big(\Phi^G(X,A)-\hat\Phi^G(X,A)\big)\big)\\
        &+\E|g^{*G}_{\tilde\alpha}(X,A)|\1\big(0\ge g^{*G}_{\tilde\alpha}(X,A)> 2\big(\eta^G(X,A)-\hat\eta^G(X,A)\big)-\lambda^{*G\top}_{\tilde\alpha}\big(\Phi^G(X,A)-\hat\Phi^G(X,A)\big)\big)\\
        \le&\E|g^{*G}_{\tilde\alpha}(X,A)|\1\big(|g^{*G}_{\tilde\alpha}(X,A)|\le 2\epsilon_\eta+\|\lambda^{*G}_{\tilde\alpha}\|_1\epsilon_\phi\big)\\
        \lesssim&\big(\epsilon_\eta+\|\lambda^{*G}_{\tilde\alpha}\|_1\epsilon_\phi\big)^{1+\tilde\gamma}.
    \end{align*}
    For term $T_2$,
    \begin{align*}
        &T_2\\
        =&\lambda^{*G\top}_{\tilde\alpha}\E\Phi^G(X,A)\1\big(0<g^{*G}_{\tilde\alpha}(X,A)\le2\big(\eta^G(X,A)-\hat\eta^G(X,A)\big)-\lambda^{*G\top}_{\tilde\alpha}\big(\Phi^G(X,A)-\hat\Phi^G(X,A)\big)\big)\\
        &-\lambda^{*G\top}_{\tilde\alpha}\E\Phi^G(X,A)\1\big(0\ge g^{*G}_{\tilde\alpha}(X,A)>2\big(\eta^G(X,A)-\hat\eta^G(X,A)\big)-\lambda^{*G\top}_{\tilde\alpha}\big(\Phi^G(X,A)-\hat\Phi^G(X,A)\big)\big)\\
        \le&\E|\lambda^{*G\top}_{\tilde\alpha}\Phi^G(X,A)|\1\big(|g^{*G}_{\tilde\alpha}(X,A)|\le2\epsilon_\eta+\|\lambda^{*G}_{\tilde\alpha}\|_1\epsilon_\phi\big)\\
        =&\sum_{k\in[\tilde K]}|\lambda^{*G}_{\tilde\alpha,k}|\E|\phi_k^G(X,A)|\1\big(|g^{*G}_{\tilde\alpha}(X,A)|\le2\epsilon_\eta+\|\lambda^{*G}_{\tilde\alpha}\|_1\epsilon_\phi\big)\\
        \le&\|\lambda^{*G}_{\tilde\alpha}\|_1\tilde\epsilon_{g,\tilde\alpha}^G.
    \end{align*}

    Combining pieces concludes that with probability at least $1-2\delta_{\rm post}$, we have
    \[\EE(\hat\lambda)\lesssim \big(\epsilon_\eta+\|\lambda^{*G}_{\tilde\alpha}\|_1\epsilon_\phi\big)^{1+\tilde\gamma}+\|\lambda^{*G}_{\tilde\alpha}\|_1\tilde\epsilon_\alpha+\bigg(\frac{\tilde K\log n+\log\frac{1}{\delta_{\rm post}}}{n}\bigg)^{\frac{1+\tilde\gamma}{2+\tilde\gamma}}.\]

\end{proof}

\end{document}